\newif\iflncs
\lncstrue
\iflncs
\documentclass{llncs}
\else
\documentclass{article}[11]
\usepackage{amsthm}
\usepackage{fullpage}
\fi

\usepackage{amsmath}
\usepackage{amsfonts}
\usepackage{amssymb}
\usepackage{ifpdf}
\usepackage{subfig}
\usepackage{wrapfig}
\usepackage{cite}
\usepackage{comment}
\usepackage{appendix}
\usepackage[nofillcomment,noend]{algorithm2e}
\usepackage{commands-tam}
\usepackage{graphicx,wrapfig}

\usepackage{colortbl}
\usepackage{tikz}

\def\proofiness/{{\noindent\color{darkgray}\sffamily\bfseries Proof.}}

\ifpdf

  \usepackage[pdftex]{epsfig}
  \usepackage[pdftex]{hyperref}

\else

    \usepackage[dvips]{epsfig}
    \newcommand{\href}[2]{#2}

\fi

\newif\ifabstract
\newif\iffull

\abstracttrue

\ifabstract
	\fullfalse
\else
	\fulltrue
\fi

\newtoks\magicAppendix
\magicAppendix={}
\newtoks\magictoks
\newif\iflater
\laterfalse
\ifabstract
\long\def\later#1{\magictoks={#1}%
  \edef\magictodo{\noexpand\magicAppendix={\the\magicAppendix \par
    \the\magictoks%
  }}
  \magictodo}
\long\def\both#1{\magictoks={#1}%
  \edef\magictodo{\noexpand\magicAppendix={\the\magicAppendix \par
    \noexpand\setcounter{theorem-preserve}{\noexpand\arabic{theorem}}%
    \noexpand\setcounter{theorem}{\arabic{theorem}}%
    \noexpand\setcounter{section-preserve}{\noexpand\arabic{section}}%
    \noexpand\setcounter{section}{\arabic{section}}%
	\noexpand\let\noexpand\oldsection=\noexpand\thesection
	\noexpand\def\noexpand\thesection{\thesection}
	\noexpand\let\noexpand\oldlabel=\noexpand\label
	\noexpand\let\noexpand\label=\noexpand\blank
    \the\magictoks%
    \noexpand\setcounter{theorem}{\noexpand\arabic{theorem-preserve}}%
    \noexpand\setcounter{section}{\noexpand\arabic{section-preserve}}%
	\noexpand\let\noexpand\thesection=\noexpand\oldsection
	\noexpand\let\noexpand\label=\noexpand\oldlabel
  }}
  \magictodo
  \the\magictoks}
\else
\long\def\later#1{#1}
\long\def\both#1{#1}
\fi
\long\def\magicappendix{
	\latertrue%
	\the\magicAppendix%
}

\def\rtam/{{RTAM}}

\begin{document}

\title{Reflections on Tiles (in Self-Assembly)}

\author{
  Jacob Hendricks%
    \thanks{{Department of Computer Science and Computer Engineering, University of Arkansas,
        \protect\url{jhendric@uark.edu}.
        Supported in part by National Science Foundation Grant CCF-1117672 and CCF-1422152.}}
\and
  Matthew J. Patitz%
    \thanks{Department of Computer Science and Computer Engineering, University of Arkansas,
      \protect\url{patitz@uark.edu}.
      Supported in part by National Science Foundation Grant CCF-1117672 and CCF-1422152.}
\and
  Trent A. Rogers%
    \thanks{Department of Mathematical Sciences, University of Arkansas,
        \protect\url{tar003@email.uark.edu}.
        This author's research was supported by the National Science Foundation Graduate Research Fellowship Program under Grant No. DGE-1450079, and National Science Foundation grants CCF-1117672 and CCF-1422152.}
}
\date{}
\institute{}

\maketitle

\begin{abstract}%
We define the Reflexive Tile Assembly Model (\rtam/), which is obtained from the abstract Tile Assembly Model (aTAM) by allowing tiles to reflect across their horizontal and/or vertical axes.  We show that the class of directed temperature-$1$ \rtam/ systems is not computationally universal, which is conjectured but unproven for the aTAM, and like the aTAM, the \rtam/ is computationally universal at temperature $2$.  We then show that at temperature $1$, when starting from a single tile seed, the \rtam/ is capable of assembling $n \times n$ squares for $n$ odd using only $n$ tile types, but incapable of assembling $n \times n$ squares for $n$ even. Moreover, we show that $n$ is a lower bound on the number of tile types needed to assemble $n \times n$ squares for $n$ odd in the temperature-$1$ \rtam/. The conjectured lower bound for temperature-$1$ aTAM systems is $2n-1$. Finally, we give preliminary results toward the classification of which finite connected shapes in $\Z^2$ can be assembled (strictly or weakly) by a singly seeded (i.e. seed of size $1$) \rtam/ system, including a complete classification of which finite connected shapes be strictly assembled by a \emph{mismatch-free} singly seeded \rtam/ system.
\end{abstract}

\section{Introduction}\label{sec:intro}

Self-assembly is the process by which disorganized components autonomously combine to form organized structures. In DNA-based self-assembly, the combining ability of the components is implemented using complementary strands of DNA as the ``glue''. In \cite{Winf98}, Winfree introduced a useful mathematical model of self-assembling systems called the abstract Tile Assembly Model (aTAM) where the autonomous components are described as square tiles with specifiable glues on their edges and the attachment of these components occurs spontaneously when glues match. The aTAM provides a convenient way of describing self-assembling systems and their resulting assemblies, and serves as the underpinning of many studies of the properties of self-assembling systems.
For a comprehensive survey of tile-based self-assembly including models other than the aTAM, see \cite{PatitzSurvey,DotCACM}.

From the broad collection of results in the aTAM, one property of systems that has been shown to yield enormous power is \emph{cooperation}.  The notion of cooperation captures the phenomenon where the attachment of a new tile to a growing assembly requires it to bind to more than one tile (usually 2) already in the assembly. The requirement for cooperation is determined by a system parameter known as the \emph{temperature}, and when the temperature is equal to $1$ (a.k.a. temperature-$1$ systems), there is no requirement for cooperation.  A long-standing conjecture is that temperature-$1$ aTAM systems are in fact not capable of universal computation or efficient shape building, although it is well-known that temperature $\ge 2$ systems are.
However, in actual laboratory implementations of DNA-based tiles \cite{RothTriangles,OrigamiSeed,SchWin07,MaoLabReiSee00,WinLiuWenSee98}, the self-assembly performed by temperature-$2$ systems does not match the error-free behavior dictated by the aTAM, but instead, a frequent source of errors is the binding of tiles using only a single bond. Thus, temperature-$1$ behavior erroneously occurs and cannot be completely prevented.

Many models of self-assembly can be thought of as extensions of the aTAM (e.g. ~\cite{GeoTiles,DotKarMasNegativeJournal,OneTile,Polyominoes,SingleNegative,CookFuSch11}), and for these models it is common to study the added power that an extra property or constraint gives the extended model. For example, in~\cite{CookFuSch11,DotKarMasNegativeJournal,Polyominoes,SingleNegative}, it is shown that at temperature $1$, when the aTAM is appropriately extended, the resulting models are computationally universal and capable of efficiently assembling shapes.
In this paper, we take the opposite approach and remove a constraint that the aTAM imposes with the goal of modelling physical systems that may be incapable of enforcing these constraints. Tiles in the aTAM are not allowed to flip or rotate prior to attachment to an existing assembly. While this  assumption is a realistic one for many implementations
of DNA-based tiles (e.g. \cite{Winf98}), for certain implementations (e.g. ~\cite{NBlocks}), it is unknown whether or not both of the conditions of this assumption can be physically enforced. (See~\cite{han2013dna,ke2012three,pinheiro2011challenges,C2CC37227D} for more experimentally produced building blocks and systems.) When DNA is used as the binding agent, single stranded DNA can be used to prevent relative tile rotation by encoding a direction (north/south or east/west) in the DNA sequence so that only strands with appropriately matching directions are complementary; on the other hand, preventing tiles from flipping may not always be possible, especially if the glues of different sides of a tile are encoded by disjoint DNA complexes.
Therefore, we consider a model based on the aTAM where tiles may nondeterministically flip horizontally and/or vertically prior to attachment.

We introduce the \emph{Reflexive Tile Assembly Model} (\rtam/), which can be thought of as the aTAM with the relaxed constraint that tiles in the \rtam/ \emph{are} allowed to flip horizontally and/or vertically. Also, unlike most formulations of the aTAM where complementary strands of DNA are represented with the same glue label, the \rtam/ explicitly uses complementary glues.  This importantly prevents copies of tiles of the same type from being able to flip and bind to each other, and is the actual reality with DNA-based tiles. We then show a series of results within the \rtam/.  First we show that at temperature $1$, the class of directed \rtam/ systems -- systems which yield a single pattern up to reflection and ignoring tile orientation -- are only capable of assembling patterns that are essentially periodic. Then, following the thesis set forth in~\cite{jLSAT1}, we conclude that the temperature-1 \rtam/ is not computationally universal.  While the inability of temperature-$1$ aTAM systems to compute is still only conjectured, we are able to conclusively prove it for \rtam/ systems, specifically by using techniques developed in \cite{jLSAT1} to study temperature-$1$ aTAM systems.  We also show that like the aTAM at temperature $2$, the class of directed temperature-$2$ \rtam/ systems is computationally universal. This shows a fundamental dividing line between the powers of \rtam/ temperature-$1$ and temperature-$2$ directed systems.  We then turn our attention to the self-assembly of squares by singly seeded temperature-$1$ \rtam/ systems where we show that for even values of $n \in \mathbb{N}$ it is impossible to self-assemble any $n \times n$ square.  This is exceptional due to the fact that it is the first demonstration of a model of tile assembly in which a finite shape is proven the be impossible to self-assemble in a directed system.  Typically, any finite shape can be self-assembled by a trivial system in which a unique tile type is created for each point of the shape.  However, due to the ability of tiles in the \rtam/ to flip, it is not possible for the \rtam/ systems to effectively constrain the reflections of tiles to produce such even squares without the possibility of tiles growing beyond the boundaries of the squares.  However, for odd values of $n$ and $m$ any $n \times m$ rectangle can be self-assembled using only $\frac{n+m}{2}$ tile types, thus implying that for $n$ odd, an $n\times n$ square can be self-assembled using only $n$ tile types. In addition, we also show that for $n$ odd, an $n\times n$ square cannot be self-assembled using less than $n$ tile types (thus, $n$ is the upper and lower bound for square assembly). This is in contrast to the aTAM at temperature $1$, where the conjectured lower bound for assembling an $n \times n$ square is $2n-1$, and hints that in certain situations the ability of \rtam/ tiles to attach in flipped orientations can be effectively harnessed to more efficiently build shapes than systems in the aTAM.  Finally, we give preliminary results toward the classification of the finite connected shapes in $\Z^2$ that can be assembled (strictly or weakly) by a singly seeded \rtam/ system, including a complete classification of which finite connected shapes be strictly assembled by a \emph{mismatch-free} singly seeded temperature-1 \rtam/ system. We also show that arbitrary shapes with scale factor $2$ can be assembled in the singly seeded temperature-1 \rtam/.  These combined results show that the ability of tiles to bind in flipped orientations is sometimes provably limiting, while at other times can provide advantages, and they provide a solid framework for the study of self-assembling systems composed of molecular building blocks unable to enforce the constraints of the aTAM.

The layout of this paper is as follows.  In Section~\ref{sec-rtam-informal} we present a high-level definition of the \rtam/.  (Due to space constraints, a more technically detailed definition can be found in Section~\ref{sec:RTAM-def-formal} of the appendix, and so can the full proofs of each result.)  Section~\ref{sec:computation} contains the proof that temperature-$1$ \rtam/ systems cannot perform universal computation and that temperature-$2$ systems can.  In Section~\ref{sec:shapes} we present our results related to the self-assembly of shapes in the \rtam/, including our results about assembling squares and classifying the finite connected shapes that self-assemble in the \rtam/.

\iffull
\section{Preliminaries}
\fi
\ifabstract
\section{Definition of the Reflexive Tile Assembly Model}
\fi

\label{sec-rtam-informal}

\iffull
\subsection{Definition of the Reflexive Tile Assembly Model}\label{sec:RTAM-def-informal}
\fi

The Reflexive Tile Assembly Model (RTAM) is essentially equivalent to the abstract Tile Assembly Model (aTAM) \cite{Winf98,RotWin00,Roth01,jSSADST} but with the modification that tiles are allowed to possibly ``flip'' across their horizontal and/or vertical axes before attaching to an assembly.  Also, as in some formulations of the aTAM, it is assumed that glues bind to complementary versions of themselves (so that two tiles of the same type which are flipped relative to each other can't simply bind to each other along the same but reflected side).  We now give a brief definition of the RTAM. See Section~\ref{sec:RTAM-def-formal} for more detailed definitions. Our notation is similar (and where appropriate, identical) to that of \cite{jSSADST}.

\ifabstract
We work in the $2$-dimensional discrete space $\Z^2$. Define the set
$U_2 = \{(0,1),\allowbreak (1,0),\allowbreak (0,-1),\allowbreak (-1,0)\}$ to be the set of all
\emph{unit vectors} in $\mathbb{Z}^2$.
We also sometimes refer to these vectors by their
cardinal directions $N$, $E$, $S$, $W$, respectively.
All \emph{graphs} in this paper are undirected.
A \emph{grid graph} is a graph $G =
(V,E)$ in which $V \subseteq \Z^2$ and every edge
$\{\vec{a},\vec{b}\} \in E$ has the property that $\vec{a} - \vec{b} \in U_2$.
Intuitively, a tile type $t$ is a unit square that can be translated and flipped across its vertical and/or horizontal axes, but not rotated.  This provides each tile type with a pair of North-South ($NS$) sides and a pair of East-West ($EW$) sides, such that either side $s \in NS$ may be facing north while the other is facing south (and vice versa for the $EW$ glues).  For ease of discussion, however, we will talk about tile types as being defined in fixed orientations, but then allow them to attach to assemblies in possibly flipped orientations. 

Fix a finite set $T$ of tile types.
Each side of a tile $t$ in $T$
has a ``glue'' with ``label'' --a string
over some fixed alphabet--and ``strength'' --a nonnegative integer--specified by its type
$t$. Let $R = \{D,V,H,B\}$ be the set of permissible reflections for a tile which is assumed to begin in the default orientation, where $D$ corresponds to no change from the default, $V$ a single vertical flip (i.e. a reflection across the $x$-axis), $H$ a single horizontal flip (i.e. a reflection across the $y$-axis), and $B$ a single horizontal flip and a single vertical flip.  
See Figure~\ref{fig:tile-flips-informal} for an example of each. It is important to note that a glue does not have any particular orientation along the edge on which it resides, and so remains unchanged throughout reflections.

\begin{figure}[htp]
\begin{center}
   \includegraphics[width=1.6in]{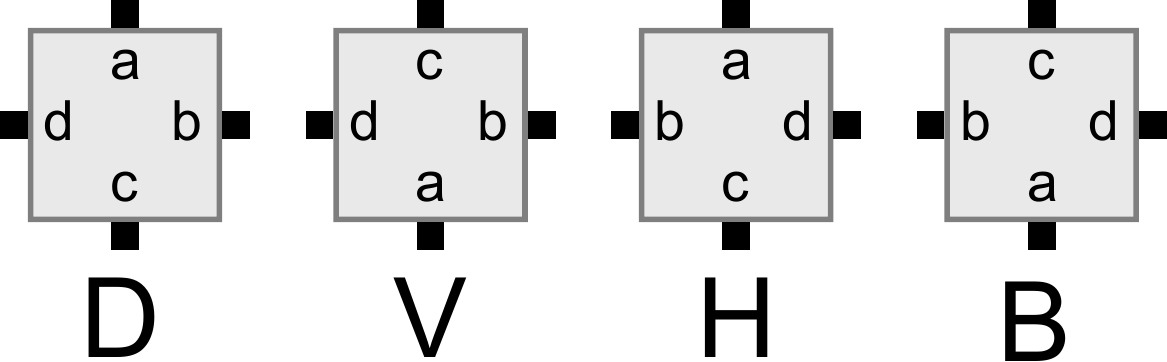}
\caption{Left to right:  (1) Default orientation of an example tile type $t$, (2) $t$ flipped vertically, (3) $t$ flipped horizontally, (4) $t$ flipped across both axes}
\label{fig:tile-flips-informal}
\end{center}
\end{figure}

Glues on adjacent edges of two tiles may bind iff they have complementary labels and the same strength.
An \emph{assembly} is a partial
function $\pfunc{\alpha}{\Z^2}{T \times R}$ defined on at least one input, with points $\vec{x}\in\Z^2$ at
which $\alpha(\vec{x})$ is undefined interpreted to be empty space,
so that $\alpha$ is the set of points with \emph{oriented} tiles. Two assemblies
$\alpha$ and $\beta$ are equivalent iff one of them can be flipped and translated so that they perfectly match at all locations. Each assembly $\alpha$ induces a \emph{binding graph}, a grid graph whose vertices are positions occupied by tiles, according to $\alpha$, with an edge between two vertices if the tiles at those vertices have complementary glues of equal strength.  Then, for some $\tau\in \N$, an assembly $\alpha$ is {\it $\tau$-stable} if every cut of the binding graph of $\alpha$ has weight at least $\tau$, where the weight of an edge is the strength of the glue it represents. When $\tau$ is clear from context, we say $\alpha$ is \emph{stable}.
For a tile set $T$, we let $p_T : T \times R \to T$ be the projection map onto $T$ (i.e. $p_T((t,r)) = t$). A \emph{configuration} given by an assembly $\alpha$ is defined to be the map from $\Z^2$ to $T$ given by $p_T\circ \alpha$.

Self-assembly begins with a {\it seed assembly} $\sigma$, in which each tile has a specified and fixed orientation, and
proceeds asynchronously and nondeterministically, with tiles in any valid reflection in $R$
adsorbing one at a time to the existing assembly in any manner that
preserves $\tau$-stability at all times.  A {\it reflexive tile assembly system}
({\it RTAS} or just {\it TAS} when the context is clear) is an ordered triple $\mathcal{T} = (T, \sigma, \tau)$,
where $T$ is a finite set of tile types, $\sigma$ is a seed assembly
with finite domain in which each tile is given a fixed orientation, and $\tau \in \N$ is the \emph{temperature} of the system which intends to model physical temperature. We use ``temperature-$\tau$ system'' to refer to any TAS with temperature $\tau$.
We write $\prodasm{\mathcal{T}}$ for the set of all $\tau$-stable assemblies that can arise
(in finitely many steps or in the limit) in $\mathcal{T}$.  An
assembly $\alpha \in \prodasm{\mathcal{T}}$ is {\it terminal}, and we write $\alpha \in
\termasm{\mathcal{T}}$, if no tile can be $\tau$-stably added to it. It is clear that $\termasm{\mathcal{T}} \subseteq \prodasm{\mathcal{T}}$. We say that $\mathcal{T}$ is \emph{directed} if and only if all of the assemblies (upto reflection and translation) of a directed system give the same configuration.

A set $X \subseteq \Z^2$ {\it weakly self-assembles} if there exists
a TAS ${\mathcal T} = (T, \sigma, \tau)$ and a set $B \subseteq T$
such that for each $\alpha \in \termasm{T}$ there exists a reflection $r \in R$ and a translation $\vec{v} \in \Z^2$ such that for the assembly, $\alpha^r$ say, corresponding to $\alpha$ reflected according to $r$ then translated by $\vec{v}$ and $\alpha_r^{-1}(B) = X$ holds.  Essentially, weak self-assembly can be thought of
as the creation (or ``painting'') of a pattern of tiles from $B$ (usually taken to be a
unique ``color'' such as black) on a possibly larger ``canvas'' of un-colored tiles.
Also, a set $X$ \emph{strictly self-assembles} if there is a TAS $\mathcal{T}$ such that
for each assembly $\alpha\in\termasm{T}$ there exists a reflection $r \in R$ and a translation $\vec{v} \in \Z^2$ such that $\alpha_r = F(\alpha,r,\vec{v})$ and $\dom \alpha_r =
X$. Essentially, strict self-assembly means that tiles are only placed
in positions defined by the shape.  Note that if $X$ strictly self-assembles, then $X$ weakly
self-assembles.

In this paper, we also consider scaled-up versions of shapes. Formally, if $X$ is a shape and $c \in \mathbb{N}$, then a $c$-\emph{scaling} of $X$ is defined as the set $X^c = \left\{ (x,y) \in \mathbb{Z}^2 \; \left| \; \left( \left\lfloor \frac{x}{c} \right\rfloor, \left\lfloor \frac{y}{c} \right\rfloor \right) \in X \right.\right\}$. Intuitively, $X^c$ is the shape obtained by replacing each point in $X$ with a $c \times c$ block of points. We refer to the natural number $c$ as the \emph{scale factor}. 
\fi

\later{
\ifabstract
\section{Formal Definition of the Reflexive Tile Assembly Model}\label{sec:RTAM-def-formal}
\fi

We work in the $2$-dimensional discrete space $\Z^2$. Define the set
$U_2 = \{(0,1),\allowbreak (1,0),\allowbreak (0,-1),\allowbreak (-1,0)\}$ to be the set of all
\emph{unit vectors} in $\mathbb{Z}^2$.
We also sometimes refer to these vectors by their
cardinal directions $N$, $E$, $S$, $W$, respectively.
All \emph{graphs} in this paper are undirected.
A \emph{grid graph} is a graph $G =
(V,E)$ in which $V \subseteq \Z^2$ and every edge
$\{\vec{a},\vec{b}\} \in E$ has the property that $\vec{a} - \vec{b} \in U_2$.

Intuitively, a tile type $t$ is a unit square that can be translated and flipped across its vertical and/or horizontal axes, but not rotated.  This provides each tile type with a pair of North-South ($NS$) sides and a pair of East-West ($EW$) sides, such that either side $s \in NS$ may be facing north while the other is facing south (and vice versa for the $EW$ glues).  For ease of discussion, however, we will talk about tile types as being defined in fixed orientations, but then allow them to attach to assemblies in possibly flipped orientations.  Therefore, we define each $t$ as having a well-defined ``side
$\vec{u}$'' for each $\vec{u} \in U_2$. Each side $\vec{u}$ of $t$
has a ``glue'' with ``label'' $\textmd{label}_t(\vec{u})$--a string
over some fixed alphabet--and ``strength''
$\textmd{str}_t(\vec{u})$--a nonnegative integer--specified by its type
$t$. Let $R = \{D,V,H,B\}$ be the set of permissible reflections for a tile which is assumed to begin in the default orientation, where $D$ corresponds to no change from the default, $V$ a single vertical flip (i.e. a reflection across the $x$-axis), $H$ a single horizontal flip (i.e. a reflection across the $y$-axis), and $B$ a single horizontal flip and a single vertical flip.  (Note that the ordering of flips for $B$ does not matter as either ordering results in the same orientation, and also that all combinations of possibly many flips across each axis result only in tiles of the orientations provided by $R$.)  See Figure~\ref{fig:tile-flips} for an example of each. Let $S: R \times U_2 \rightarrow U_2$ be a function which takes a type of reflection $r \in R$ and a side $s \in U_2$, and which returns the side of a tile in its default orientation which would appear on side $s$ of the tile when it has been reflected according to $r$.  (E.g. for the tile type shown in Figure~\ref{fig:tile-flips}, $S(H,W)=E$, and $S(H,N)=N$.)  It is important to note that a glue does not have any particular orientation along the edge on which it resides, and so remains unchanged throughout reflections.

\begin{figure}[htp]
\begin{center}
   \includegraphics[width=2.0in]{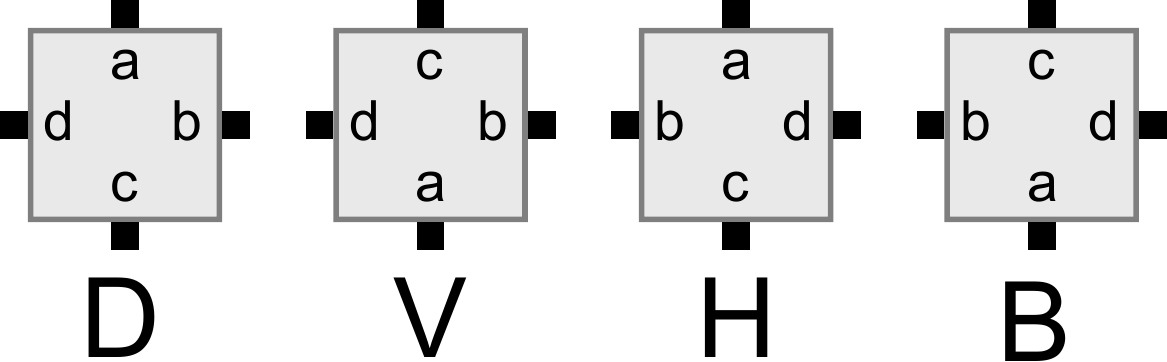}
\caption{Left to right:  (1) Default orientation of an example tile type $t$, (2) $t$ flipped vertically, (3) $t$ flipped horizontally, (4) $t$ flipped across both axes}
\label{fig:tile-flips}
\end{center}
\end{figure}

Two tiles $t$ and $t'$ that are placed at the points $\vec{a}$
and $\vec{a}+\vec{u}$ and reflected by $r \in R$ and $r' \in R$, respectively, \emph{bind} with \emph{strength}
$\textmd{str}_t\left(S(r,\vec{u})\right)$ if and only if
$\left(\textmd{label}_t\left(S(r,\vec{u})\right),\textmd{str}_t\left(S(r,\vec{u})\right)\right)
=
\left(\overline{\textmd{label}_{t'}\left(S(r',-\vec{u})\right)},\textmd{str}_{t'}\left(S(r',-\vec{u})\right)\right)$.  That is, the glues on adjacent edges of two tiles bind iff they have complementary labels and the same strength.

In the subsequent definitions, given two partial functions $f,g$, we write $f(x) = g(x)$ if~$f$ and~$g$ are both defined and equal on~$x$, or if~$f$ and~$g$ are both undefined on $x$.

Fix a finite set $T$ of tile types.
A $T$-\emph{assembly}, sometimes denoted simply as an \emph{assembly} when $T$ is clear from the context, is a partial
function $\pfunc{\alpha}{\Z^2}{T \times R}$ defined on at least one input, with points $\vec{x}\in\Z^2$ at
which $\alpha(\vec{x})$ is undefined interpreted to be empty space,
so that $\dom \alpha$ is the set of points with \emph{oriented} tiles.
We write $|\alpha|$ to denote $|\dom \alpha|$, and we say $\alpha$ is
\emph{finite} if $|\alpha|$ is finite. For a given location $\vec{v} \in \Z^2$, we denote the tile in $\alpha$ at location $\vec{v}$ by $\alpha(\vec{v})$ (if no tile exists there, $\alpha(\vec{v})$ is undefined).  Let $F$ be a function which takes as input an assembly $\alpha$, a reflection $r \in R$, and a translation vector $\vec{v} \in \Z^2$, and which returns the assembly $\alpha^r$, corresponding to $\alpha$ reflected according to $r$ then translated by $\vec{v}$.  We say that two assemblies, $\alpha$ and $\beta$, are equivalent iff there exists some reflection $r \in R$ and translation vector $\vec{v} \in \Z^2$ such that for $\beta' = F(\beta,r,\vec{v})$, $|\alpha| = |\beta'|$ and for all $\vec{v} \in |\alpha|$, $\alpha(\vec{v}) = \beta'(\vec{v})$.  That is, $\alpha$ and $\beta$ are equivalent iff one of them can be flipped and translated so that they perfectly match at all locations.  For assemblies $\alpha$
and $\alpha'$, we say that $\alpha$ is a \emph{subassembly} of
$\alpha'$, and write $\alpha \sqsubseteq \alpha'$, if $\dom \alpha
\subseteq \dom \alpha'$ and $\alpha(\vec{x}) = \alpha'(\vec{x})$ for
all $x \in \dom \alpha$. An assembly $\alpha$ is {\it $\tau$-stable}
For some $\tau\in \N$, an assembly $\alpha$ is {\it $\tau$-stable} if every cut of the binding graph of $\alpha$ has weight at least $\tau$, where the weight of an edge is the strength of the glue it represents. When $\tau$ is clear from context, we say $\alpha$ is \emph{stable}.

For a tile set $T$, we let $p_T : T \times R \to T$ be the projection map onto $T$ (i.e. $p_T((t,r)) = t$). A \emph{configuration} given by an assembly $\alpha$ is defined to  be the map from $\Z^2$ to $T$ given by $p_T\circ \alpha$.

Self-assembly begins with a {\it seed assembly} $\sigma$, in which each tile has a specified and fixed orientation, and
proceeds asynchronously and nondeterministically, with tiles in any valid reflection in $R$
adsorbing one at a time to the existing assembly in any manner that
preserves $\tau$-stability at all times.  A {\it tile assembly system}
({\it TAS}) is an ordered triple $\mathcal{T} = (T, \sigma, \tau)$,
where $T$ is a finite set of tile types, $\sigma$ is a seed assembly
with finite domain in which each tile is given a fixed orientation, and $\tau \in \N$ is the \emph{temperature}.  A {\it generalized tile
assembly system} ({\it GTAS})
is defined similarly, but without the finiteness requirements.  We
write $\prodasm{\mathcal{T}}$ for the set of all assemblies that can arise
(in finitely many steps or in the limit) from $\mathcal{T}$.  An
assembly $\alpha \in \prodasm{\mathcal{T}}$ is {\it terminal}, and we write $\alpha \in
\termasm{\mathcal{T}}$, if no tile can be $\tau$-stably added to it. It is clear that $\termasm{\mathcal{T}} \subseteq \prodasm{\mathcal{T}}$.

An assembly sequence in a TAS $\mathcal{T}$ is a (finite or infinite) sequence $\vec{\alpha} = (\alpha_0,\alpha_1,\ldots)$ of assemblies in which each $\alpha_{i+1}$ is obtained from $\alpha_i$ by the addition of a single tile. The \emph{result} $\res{\vec{\alpha}}$ of such an assembly sequence is its unique limiting assembly. (This is the last assembly in the sequence if the sequence is finite.) The set $\prodasm{T}$ is partially ordered by the relation $\longrightarrow$ defined by
\begin{eqnarray*}
\alpha \longrightarrow \alpha' & \textmd{iff} & \textmd{there is an assembly sequence } \vec{\alpha} = (\alpha_0,\alpha_1,\ldots) \\
                               &              & \textmd{such that } \alpha_0 = \alpha \textmd{ and } \alpha' = \res{\vec{\alpha}}. \\
\end{eqnarray*}
We say that $\mathcal{T}$ is \emph{strongly directed} if and only if either $|\termasm{\mathcal{T}}| = 1$ or if for every pair of terminal assemblies $\alpha,\beta \in \termasm{\mathcal{T}}$, there exists a reflection $r \in R$ and a translation vector $\vec{v} \in \Z^2$ such that $\alpha = F(\beta,r,\vec{v})$. Furthermore, we say that $\mathcal{T}$ is \emph{directed} if and only if for all $\alpha, \beta \in \mathcal{T}$, there exists a reflection $r\in R$ and a translation vector $\vec{v} \in \Z^2$ such that $p_T\circ \alpha = p_T\circ F(\beta, r, \vec{v})$. In other words, all of the assemblies of a directed systems give the same configuration.

A set $X \subseteq \Z^2$ {\it weakly self-assembles} if there exists
a TAS ${\mathcal T} = (T, \sigma, \tau)$ and a set $B \subseteq T$
such that for each $\alpha \in \termasm{T}$ there exists a reflection $r \in R$ and a translation $\vec{v} \in \Z^2$ such that $\alpha_r = F(\alpha,r,\vec{v})$ and $\alpha_r^{-1}(B) = X$ holds.  Essentially, weak self-assembly can be thought of
as the creation (or ``painting'') of a pattern of tiles from $B$ (usually taken to be a
unique ``color'' such as black) on a possibly larger ``canvas'' of un-colored tiles.

A set $X$ \emph{strictly self-assembles} if there is a TAS $\mathcal{T}$ such that
for each assembly $\alpha\in\termasm{T}$ there exists a reflection $r \in R$ and a translation $\vec{v} \in \Z^2$ such that $\alpha_r = F(\alpha,r,\vec{v})$ and $\dom \alpha_r =
X$. Essentially, strict self-assembly means that tiles are only placed
in positions defined by $X$.  Note that if $X$ strictly self-assembles, then $X$ weakly
self-assembles. $X$ in the definition of strict or weak self-assembly is called a \emph{shape} in $\Z^2$. 

In this paper, we also consider scaled-up versions of shapes. Formally, if $X$ is a shape and $c \in \mathbb{N}$, then a $c$-\emph{scaling} of $X$ is defined as the set $X^c = \left\{ (x,y) \in \mathbb{Z}^2 \; \left| \; \left( \left\lfloor \frac{x}{c} \right\rfloor, \left\lfloor \frac{y}{c} \right\rfloor \right) \in X \right.\right\}$. Intuitively, $X^c$ is the shape obtained by replacing each point in $X$ with a $c \times c$ block of points. We refer to the natural number $c$ as the \emph{scale factor}.

\subsection{Paths in the Binding Graph and as Assemblies}\label{sec:def-paths}

Given an assembly $\alpha$ and locations $\vec{x}$ and $\vec{y}$ such that $\vec{x}$,$\vec{y}\in \dom \alpha $, we define a \emph{path in $\alpha$ from $\vec{x}$ to $\vec{y}$} (or simply a \emph{path from $\vec{x}$ to $\vec{y}$}) as a simple directed path in the binding graph of $\alpha$ with the first location being $\vec{x}$ and the last $\vec{y}$.  We refer to such a path as $\pi_x^y$, and for $k = |\pi_x^y|$ (i.e. $k$ is the length of, or number of tiles on, $\pi_x^y$) and $0 \leq i < k$, let $\pi_x^y(i)$ be the $i$th location of $\pi_x^y$.  Thus, $\pi_x^y(0) = \vec{x}$, and $\pi_x^y(k-1) = \vec{y}$.  We can thus refer to the $i$th tile on $\pi_x^y$ and its reflection as $\alpha(\pi_x^y(i))$, and as shorthand will often refer to locations and/or tiles along a path.
Regardless of the order in which the tiles of $\pi_x^y$ were placed in $\alpha$, we define \emph{input} and \emph{output} sides for each tile in $\pi_x^y$ (except for the first and last, respectively) in relation to their position on $\pi_x^y$.  The input side of the $i$th tile of $\pi_x^y$, $\alpha(\pi_x^y(i))$, is that which binds to $\alpha(\pi_x^y(i-1))$, and the output side is that which binds to $\alpha(\pi_x^y(i+1))$.  We denote these sides as $IN(\pi_x^y(i))$ and $OUT(\pi_x^y(i))$, respectively.  (Thus, $\alpha(\pi_x^y(0))$ has no input side, and $\alpha(\pi_x^y(k-1))$ has no output side.)  Note that in a temperature-$1$ system, an assembly $\alpha'$ exactly representing $\pi_x^y$ would be able to grow solely from $\alpha(\pi_x^y(0))$, in the order of $\pi_x^y$, with each tile having input and output sides as defined for $\pi_x^y$.

} %

\section{The \rtam/ is Not Computationally Universal at $\tau=1$}\label{sec:computation}

In this section, we show that directed \rtam/ systems are not computationally universal by showing that any shape weakly assembled by a directed \rtam/ system is ``simple''. We will first define our notion of simple. Many of the following definitions can also be found in~\cite{jLSAT1}.

\begin{definition}
\label{def-doubly-periodic-set} A set $X \subseteq \mathbb{Z}^2$ is
\emph{semi-doubly periodic} if there exist three vectors $\vec{b}$,
$\vec{u}$, and $\vec{v}$ in $\mathbb{Z}^2$ such that
$$\label{eq-doubly-periodic-def}
X = { \setl{\vec{b} + n\cdot \vec{u} + m\cdot \vec{v}}{n,m \in \mathbb{N} }}.$$
\end{definition}

Less formally, a semi-doubly periodic set is a set that repeats infinitely along two vectors (linearly independent vectors in the non-degenerate case), starting at some base point $\vec{b}$.
Now, let $\mathcal{T} = (T,\sigma, 1)$ refer to a directed, temperature-1 \rtam/ system. We show that any such $\mathcal{T}$ weakly self-assembles a set $X \subseteq \Z^2$ that is a finite union of semi-doubly periodic sets.

\begin{theorem}
\label{thm-no-computation} Let $\mathcal{T} = (T,\sigma,1)$ be a
directed \rtam/ system. If a set $X \subseteq \mathbb{Z}^2$ weakly self-assembles in $\mathcal{T}$, then $X$ is a finite union of semi-doubly periodic sets.
\end{theorem}

\ifabstract

\begin{proof} (sketch)
Here we give a high-level sketch of the proof of Theorem~\ref{thm-no-computation}. See Section~\ref{sec:thm-no-computation-proof} for a rigorous proof. The basic idea of the proof is as follows. For an \rtam/ system $\mathcal{T} = (T,\sigma, 1)$ we consider all of the paths of $n$ tiles (for $n$ to be defined) that can assemble from each exposed glue of $\sigma$ such that each consecutive tile that binds forming the path attach via a north or west glue (and we say that such a path ``extends to the north-west'').

\begin{figure}[htp]

\centering
\includegraphics[scale=.17]{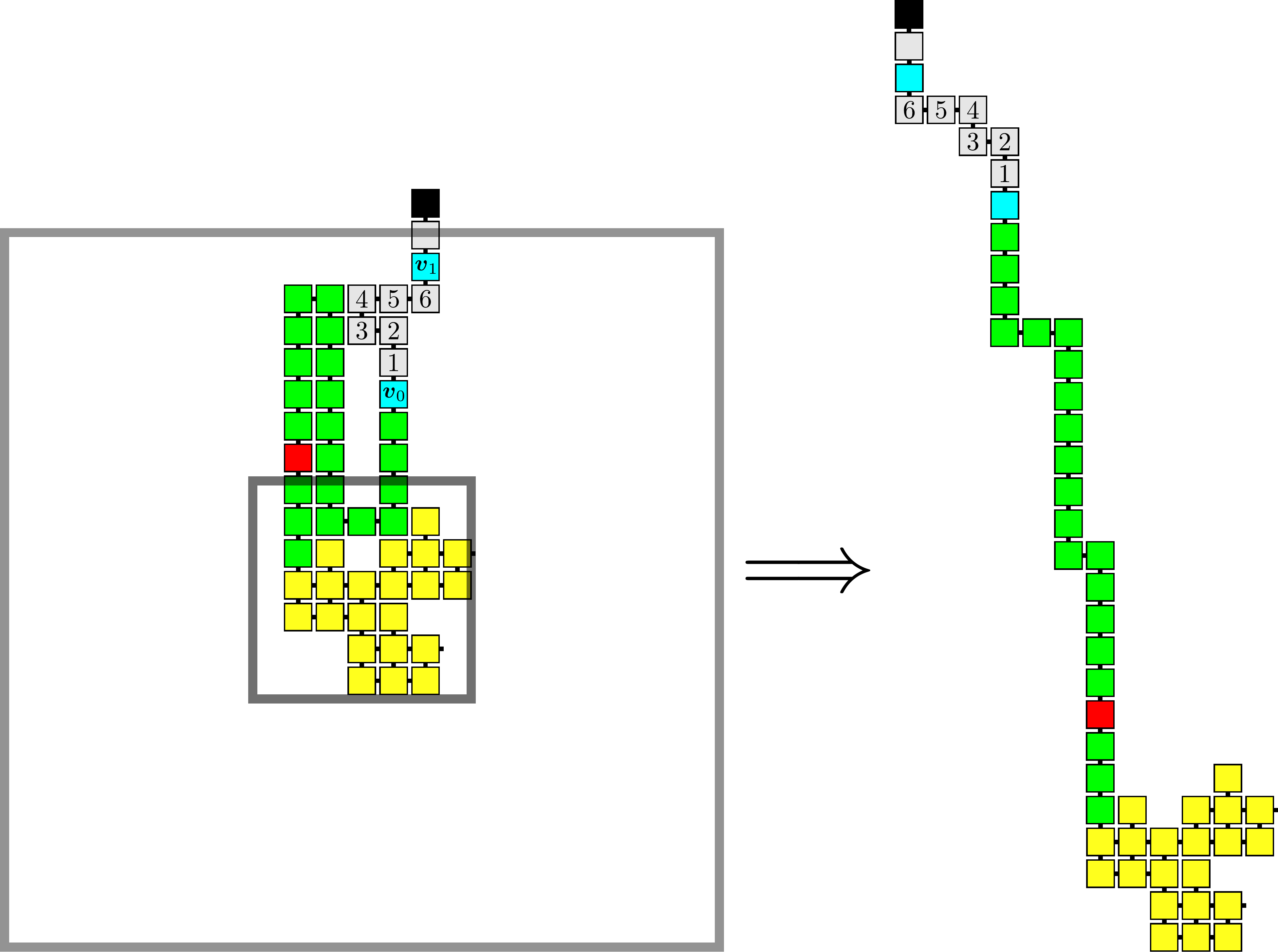}
\caption{A depiction of tiles for a path that can assemble in $\mathcal{T}$. The original path is on the left. The path on the right is a modification to the path on the left that must also be able to assemble in $\mathcal{T}$. The blue tiles labeled $\vec{v}_0$ and $\vec{v}_1$ are of the same tile type and orientation. The tiles labeled $1$ through $6$ can be repeated indefinitely as depicted in Table~\ref{tbl:paths-informal}(a).}
\label{fig:pumpable-path-informal}

\end{figure}

Any finite path is trivially the union of semi-doubly periodic sets. Then, for $n$ sufficiently large, a path of $n$ tiles that extends to the north-west must contain two distinct tiles $t_1$ and $t_2$ of the same tile type in the same orientation. If for every such path, every two distinct tiles $t_1$ and $t_2$ of the same tile type in the same orientation lie on a horizontal or vertical line, then it can be argued that the terminal configuration of $\mathcal{T}$ must consist of finitely many infinitely long horizontal or vertical paths connected to $\sigma$, and is therefore the finite union of semi-doubly periodic sets. On the other hand, if there is a path such that the two distinct tiles $t_1$ and $t_2$ of the same tile type in the same orientation do not lie on a horizontal or vertical path, then we argue that the terminal assembly of $\mathcal{T}$ is the finite union of semi-doubly periodic sets as follows. First, we note that for such a path, $\pi$ say, the tiles between $t_1$ and $t_2$ can be repeated indefinitely. This is shown in Table~\ref{tbl:paths-informal}(a).
\newcolumntype{M}{>{\centering\arraybackslash}m{\dimexpr.25\linewidth-2\tabcolsep}}
\begin{table}[htp]
\centering
\begin{tabular}{| M | M | M | M |}
	\hline
	\includegraphics[scale=.07]{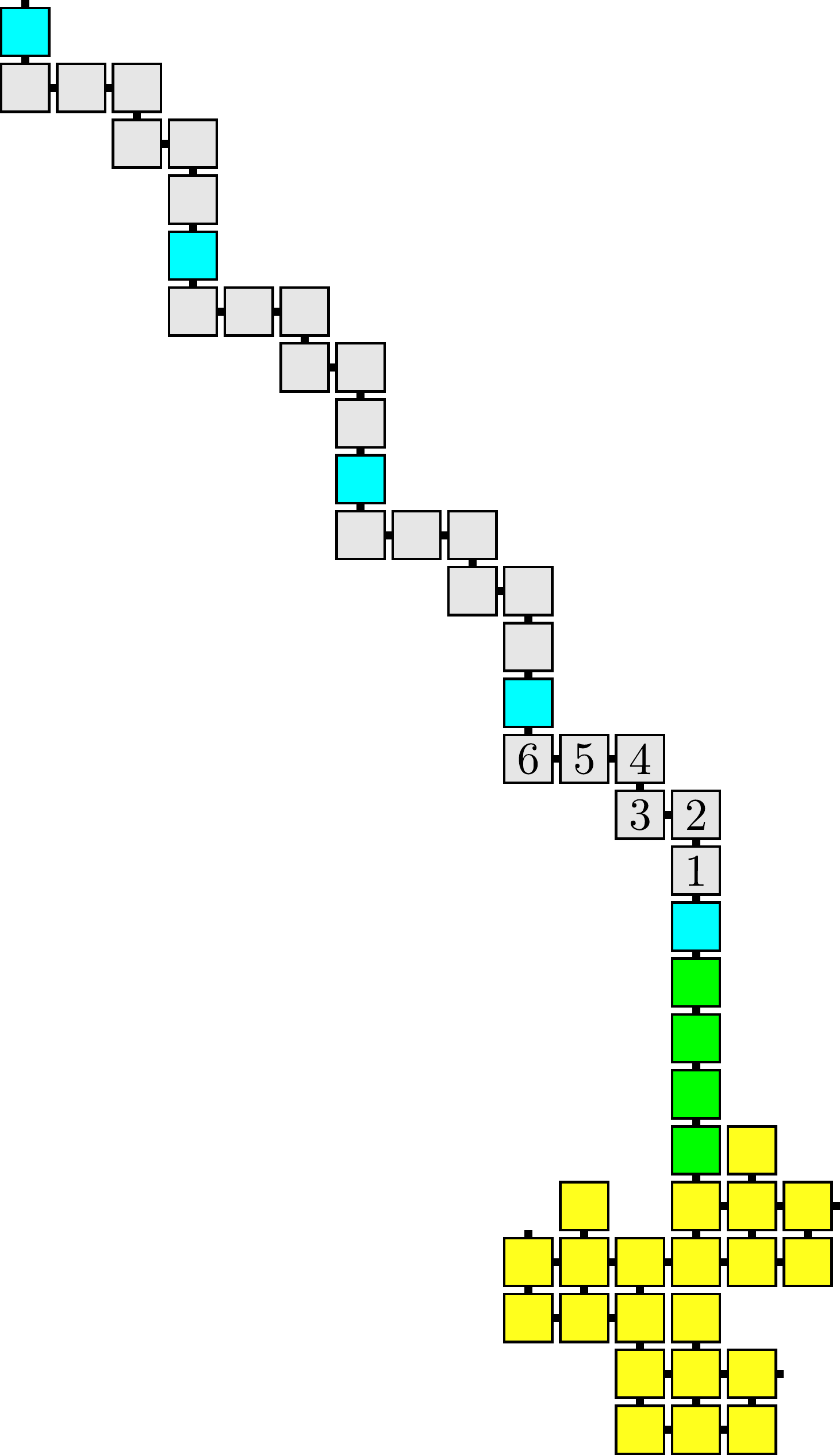}  & \includegraphics[scale=.07]{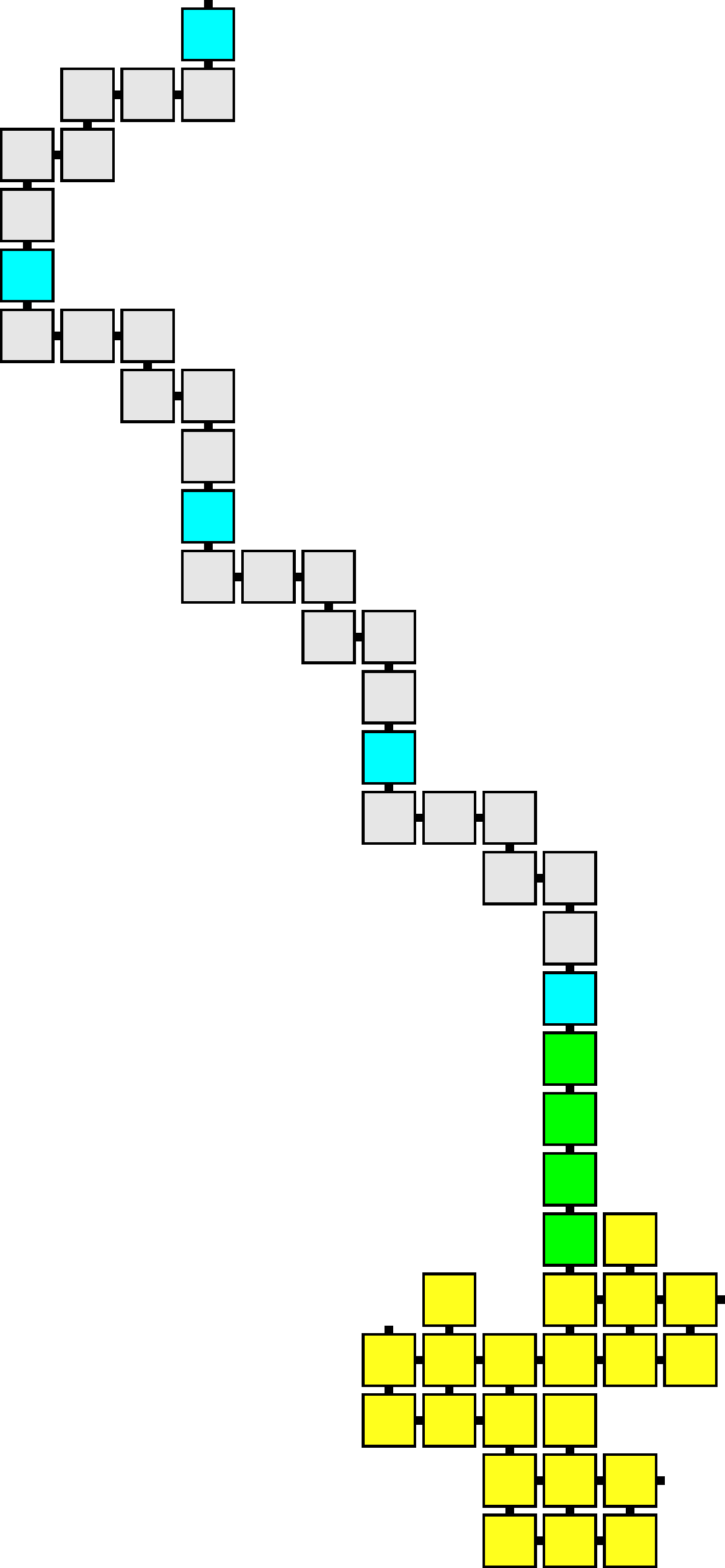} &
	\includegraphics[scale=.07]{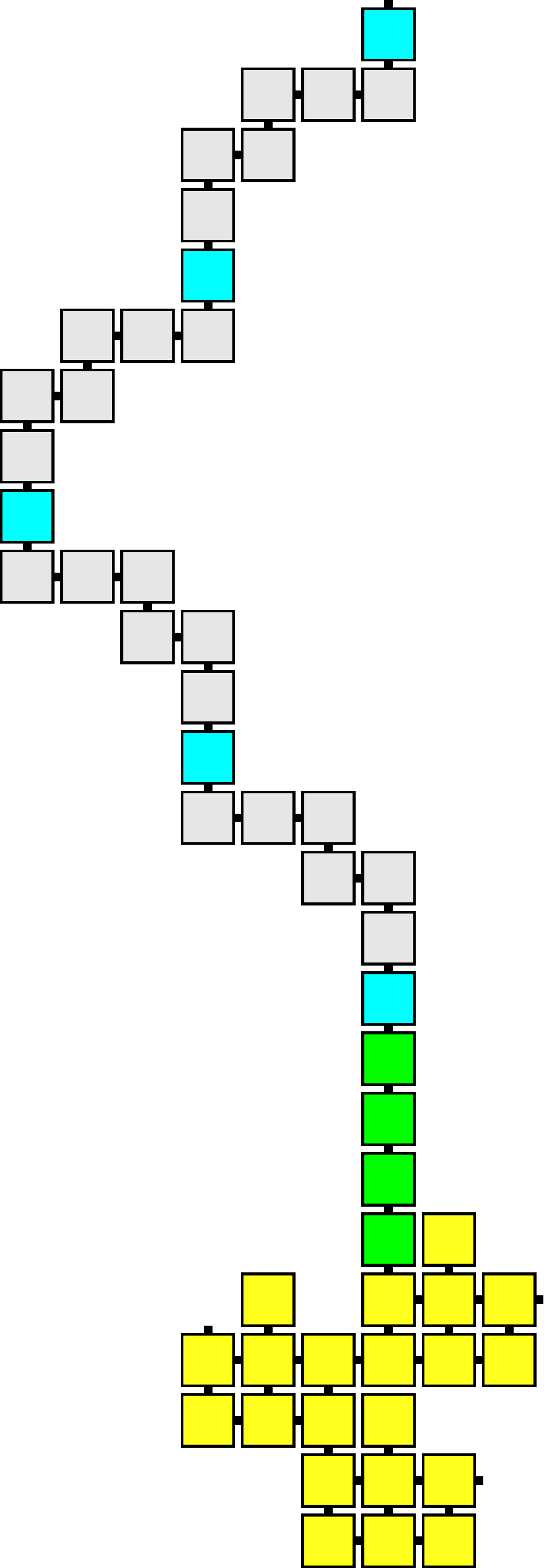}  & \includegraphics[scale=.07]{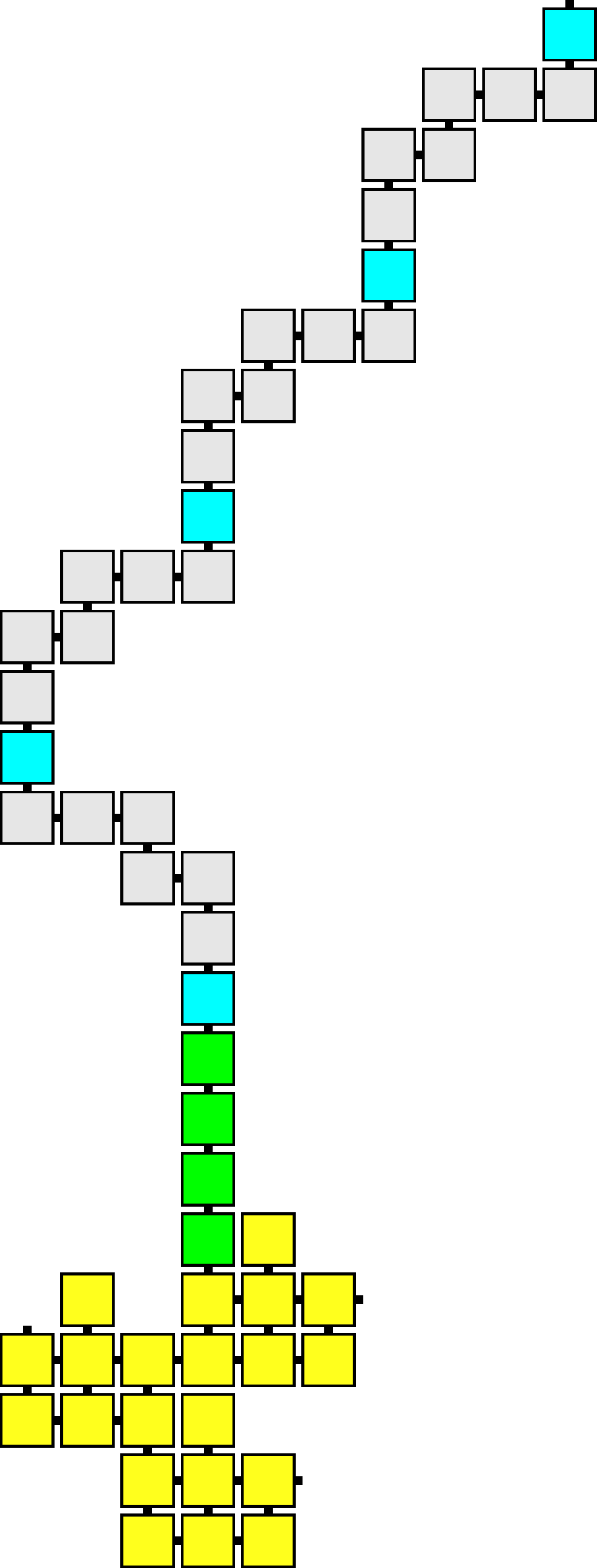} \\
	(a) & (b) & (c) & (d) \\\hline
	\includegraphics[scale=.07]{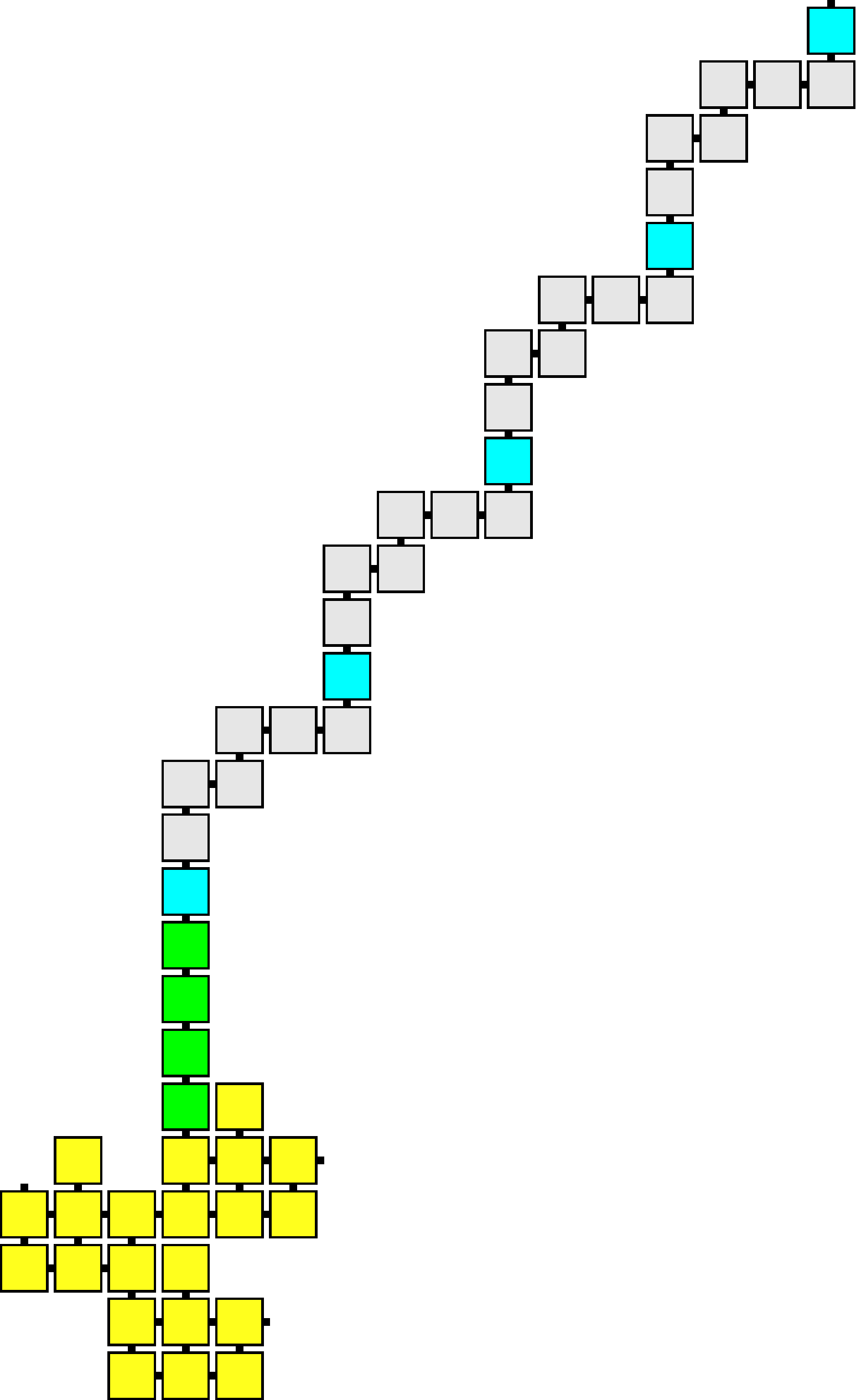}  & \includegraphics[scale=.07]{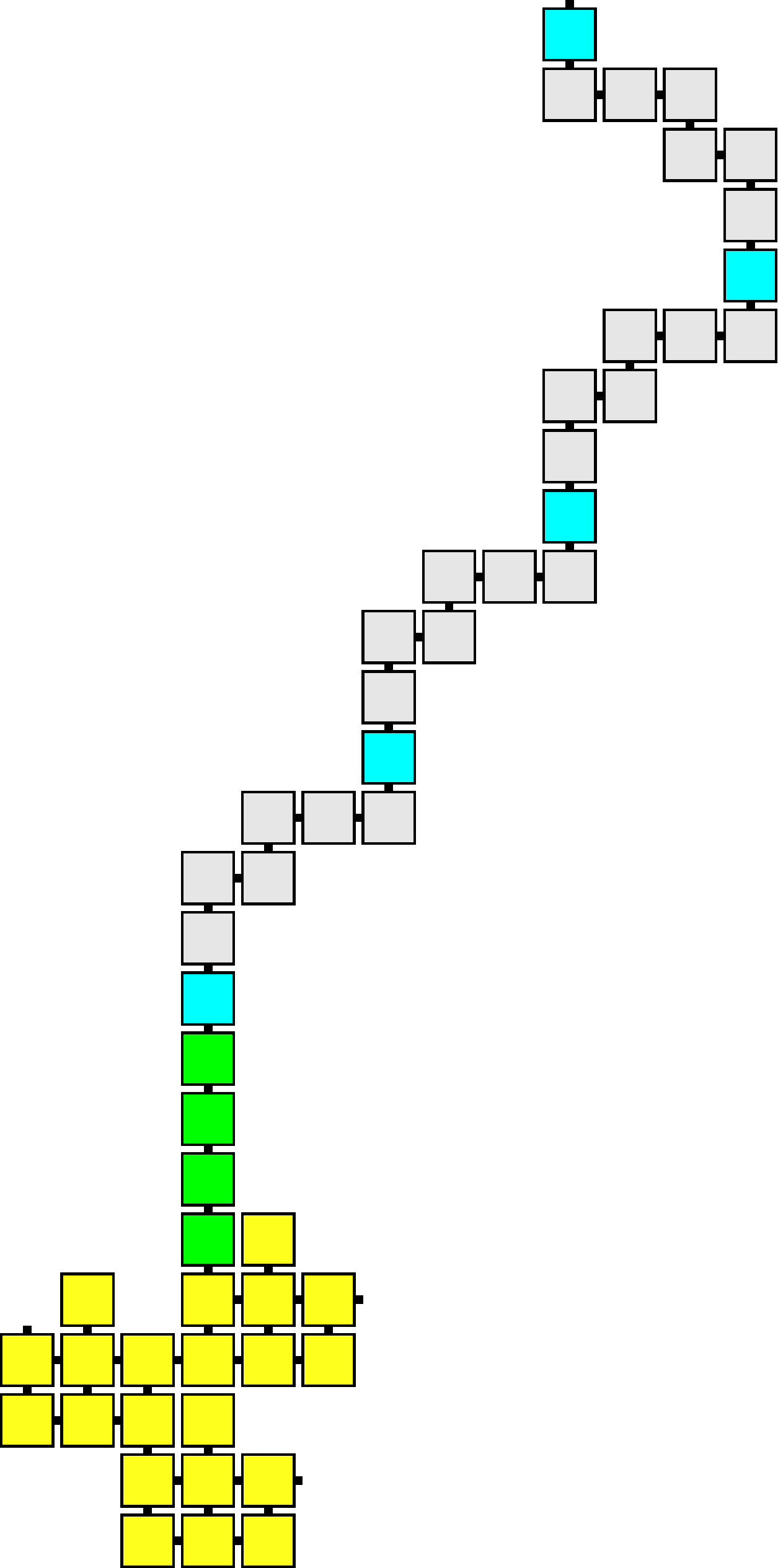} &
	\includegraphics[scale=.07]{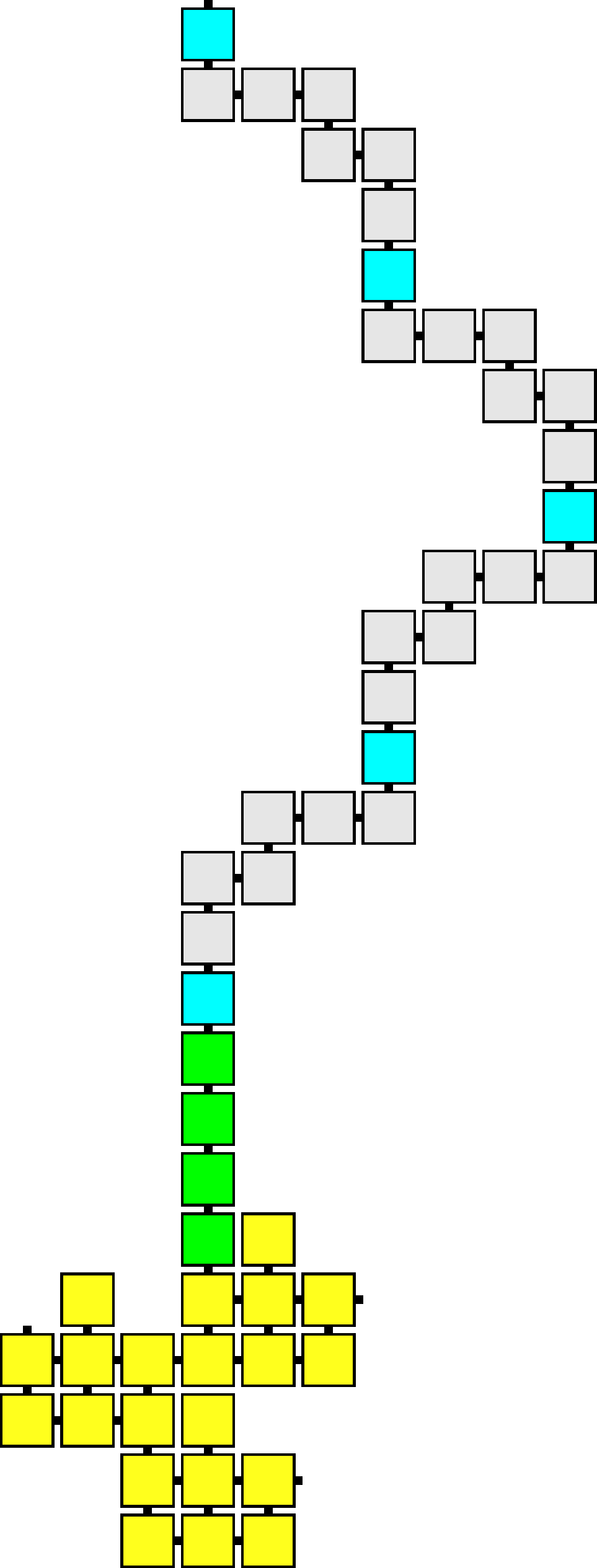}  & \includegraphics[scale=.07]{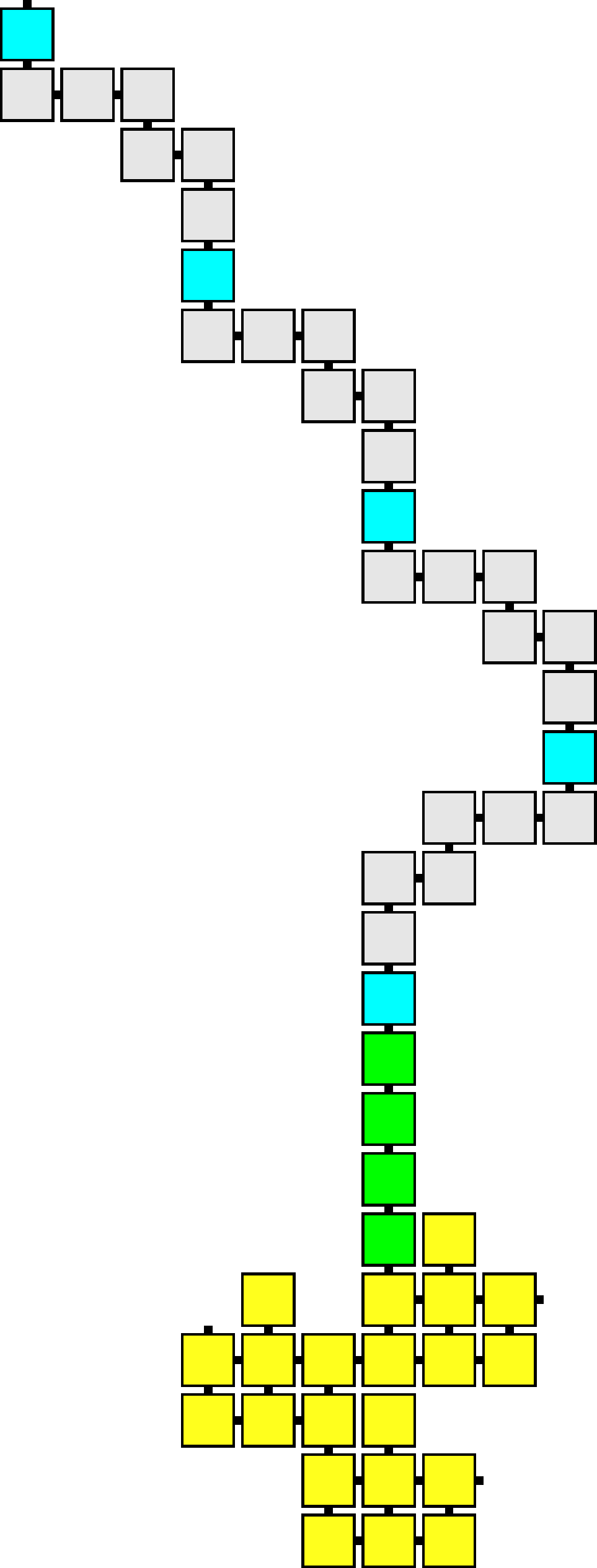} \\
	(e) & (f) & (g) & (h) \\\hline
\end{tabular}\caption{Each figure in this table depicts a possible path that can assemble in $\mathcal{T}$. The yellow tiles make up the seed and the green tiles are a path leading to the repeated tile that allows the path to repeat.}\label{tbl:paths-informal}
\end{table}
Then we show how to modify $\pi$ by reflecting tiles to obtain an infinite family of paths. Examples of such modified paths are shown in Table~\ref{tbl:paths-informal}(b)-(h). Now, if a tile $t$ belongs to one of these paths and has location $l$ say, then, since $\mathcal{T}$ is directed the terminal assembly of $\mathcal{T}$ must contain a tile of the same type as $t$ at each such location $l$.
\begin{figure}[htp]
\centering
  \subfloat[][]{%
        \label{fig:combs-informal}%
        \makebox[2in][c]{ \includegraphics[scale=.08]{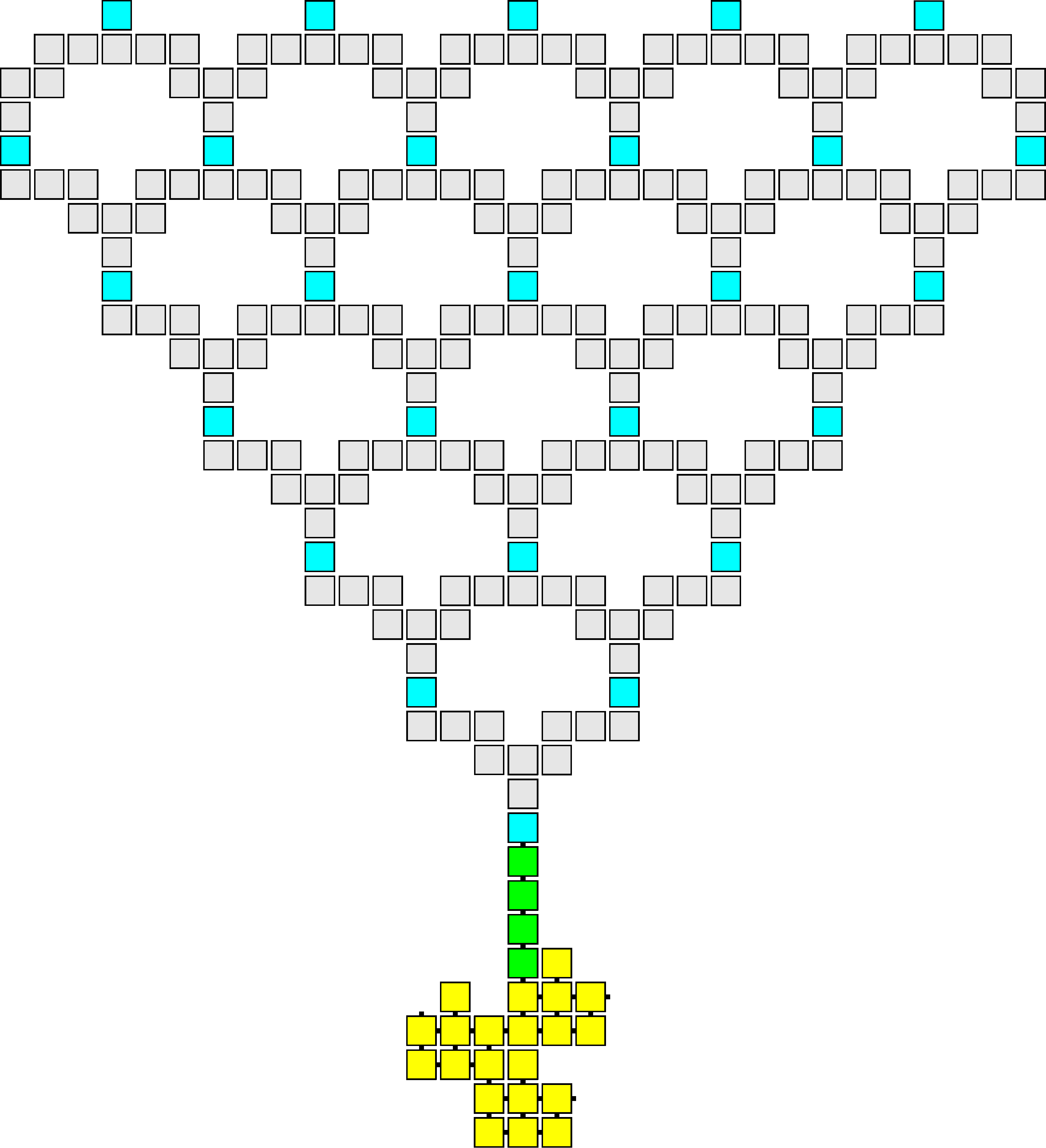}}
        }%
        \quad
  \subfloat[][]{%
        \label{fig:pump-with-seed-informal}%
        \makebox[2in][c]{ \includegraphics[width=1.7in]{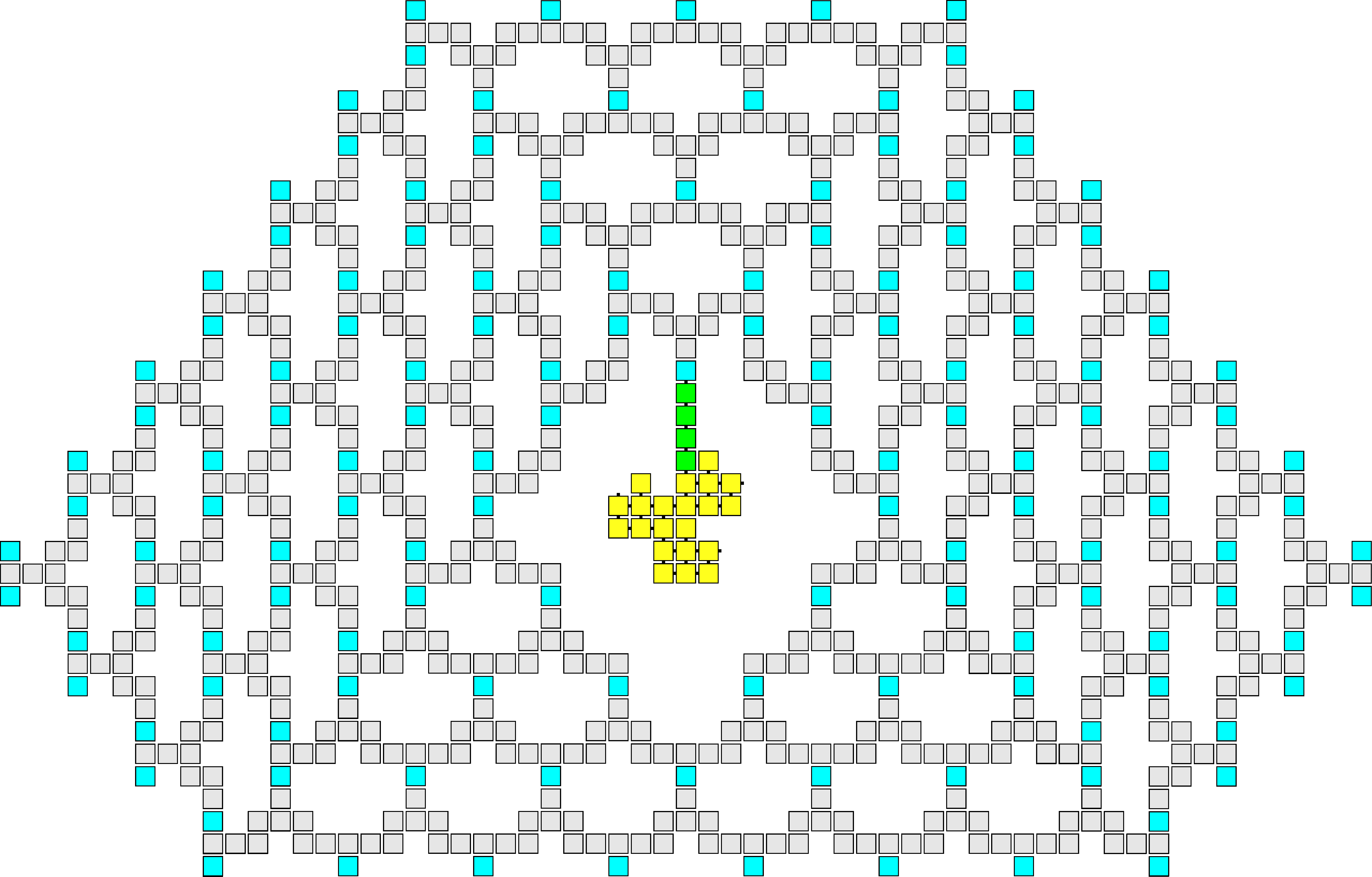}}
        }%
  \caption{(a) A configuration that can be thought of as the ``union'' of the type-consistent assemblies depicted in Table~\ref{tbl:paths-informal}. (b) A configuration of tiles that must weakly self-assemble in $\mathcal{T}$.}
  \label{fig:periodic}
\end{figure}
Finally, we note that all of these paths taken together form a semi-doubly periodic set. This is depicted in Figure~\ref{fig:combs-informal}. Continuing this line of reasoning, we show that the terminal assembly of $\mathcal{T}$ is the finite union of semi-doubly periodic sets. A portion of such an assembly is shown in Figure~\ref{fig:pump-with-seed-informal}.
\qed
\end{proof}
\fi

\later{
\ifabstract
\section{Proof of Theorem~\ref{thm-no-computation} and Additional Corollaries}\label{sec:thm-no-computation-proof}
\fi

In the proof of Theorem~\ref{thm-no-computation}, we will be considering many different assemblies and will use the following definition to form a ``union'' of the configurations given by these assemblies.

\begin{definition}
Two assemblies $\alpha$ and $\beta$ are \emph{type-consistent} iff for each $\vec{l}\in \dom \alpha \cap \dom \beta$, $p_T(\alpha(\vec{l})) = p_T(\beta(\vec{l}))$.
\end{definition}

Recall that $p_T$ is the projection from $T \times R$ onto $T$. Notice that if two assemblies $\alpha$ and $\beta$ are type-consistent, then we can give a well-defined partial function $f$ from $\Z^2$ to $T$ by $x \mapsto t$ iff either $x\in \dom \alpha$ and $\alpha(x)=t$ or $x\in \dom \beta$ and $\beta(x)=t$, and otherwise $f$ is undefined. The following definitions will also be useful in the proof of Theorem~\ref{thm-no-computation}.

Let $c \in \mathbb{N}$ and $\vec{v} \in \mathbb{Z}^2$. The \emph{box of radius} $c$ \emph{centered about the point} $\vec{v}$ is the set of points defined as $B_c\left(\vec{v}\right) = \setl{(x,y) + \vec{v}}{|x| \leq c \text{ and } |y| \leq c}$.
Finally, we say that a path in the binding graph \emph{turns} if there are two nodes $\vec{n}_0 = (x_0,y_0)$ and $\vec{n}_1 = (x_1,y_1)$ in the path such that $x_0 \neq x_1$ and $y_0\neq y_1$. In other words, a path turns if its vertices do not lie on a horizontal or vertical line. A path of tiles in an assembly is said to turn if the corresponding path in the binding graph turns.

Without loss of generality, assume that $\sigma$ contains the point $\vec{0} \in \Z^2$ where $\vec{0} = (0,0)$. Let $c\in \N$ be the minimal number such that $\sigma \subseteq B_c(\vec{0})$ and let $c'\in \N$ be such that $c'= 4|T| + c + 1$. Now let $\alpha \in \termasm{T}$ be a terminal assembly such that $\alpha$ and $\sigma$ are type-consistent. Note that a finite assembly trivially weakly assembles a finite union of semi-doubly periodic sets. Therefore, we only consider the case where $\overline{B_{c'}(\vec{0})} \cap \dom \alpha \neq \emptyset$. Let $\vec{x} \in \Z^2$ be a tile location for some tile in $\overline{B_{c'}(\vec{0})} \cap \dom \alpha$, and let $\pi_{\sigma}^{\vec{x}}$ be a path in $G_\alpha$ from a node corresponding to a tile in $\sigma$ to $\vec{x}$. (Notice the abuse of notation here. We are using $\sigma$ in the notation $\pi_{\sigma}^{\vec{x}}$ where we should be using the location of the tile in $\sigma$.) Notice that $\pi_{\sigma}^{\vec{x}}$ must contain at least $4|T|+1$ nodes. Now, at temperature $1$, $\sigma\cup \pi_{\sigma}^{\vec{x}}$ is a producible assembly of $\mathcal{T}$; denote this assembly by $\beta$. We will modify $\beta$ by modifying the path $\pi_{\sigma}^{\vec{x}}$.
\begin{figure}[htp]
\begin{center}
	\includegraphics[scale=.18]{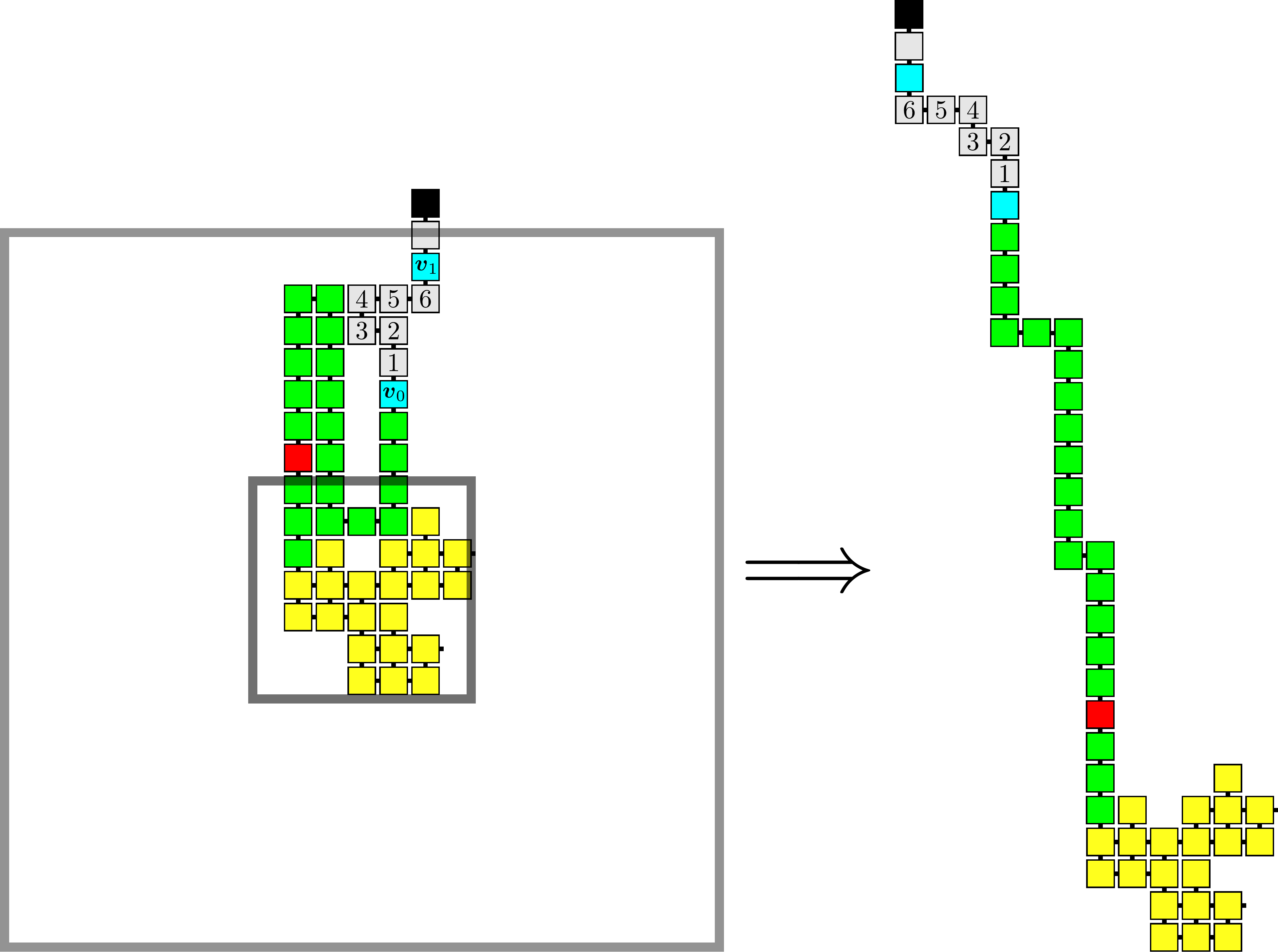}
\caption{A depiction of tiles for a path that can assemble in $\mathcal{T}$. The original path is on the left. The path on the right is a modification to the path on the left that must also be able to assemble in $\mathcal{T}$. On the left, the light grey inner box delineates $B_c(\vec{0})$ and the outer one delineates $B_{c'}(\vec{0})$. The blue tiles labeled $\vec{v}_0$ and $\vec{v}_1$ are at the respective locations $\vec{v}_0$ and $\vec{v}_1$ in $\Z^2$. The red tile is the first tile ($t$) along the path with location outside of $B_c(\vec{0})$ and the black tile is the first tile along the path with location ($\vec{x}$) outside of $B_{c'}(\vec{0})$.}
\label{fig:pumpable-path}
\end{center}

\end{figure}
Let $\vec{\beta}$ be the assembly sequence for $\beta$, and let $t$ be the first tile placed outside of $B_c(\vec{0})$ using this assembly sequence. This tile is marked as the red tile in Figure~\ref{fig:pumpable-path}. We will denote the path from $t$ to $\vec{x}$ by $\pi_t^{\vec{x}}$, and denote the subpath of $\pi_{\sigma}^{\vec{x}}$ that does not contain the tiles of $\pi_t^{\vec{x}}$ by $\pi_{\sigma}^{t}$. (Again we abuse notion here. We are using $t$ in the notation $\pi_t^{\vec{x}}$ where we should technically use the location of $t$.)  Without loss of generality, suppose that for $\vec{l} = (l_x,l_y)$, the location of $t$, for all $(x,y) \in B_c(\vec{0})$, $l_y > y$. In other words, suppose that $t$ lies to the north of $B_c(\vec{0})$. This is depicted in Figure~\ref{fig:pumpable-path}. Now we modify the path $\pi_{t}^{\vec{x}}$ as follows.

Let $\vec{\alpha} = (\alpha_0, \alpha_1,...,\alpha_n)$ be the assembly sequence such that $\alpha_0 = \sigma$ and $\alpha_i$ is obtained from $\alpha_{i-1}$ by the attachment of the tile $\pi_{\sigma}^{\vec{x}}(i-1)$ of $\pi_{\sigma}^{\vec{x}}$. We modify $\vec{\alpha}$ as follows. Suppose that $k$ is such that $t$ is the tile at the location in $\dom \alpha_k \setminus \dom \alpha_{k-1}$. That is, $k = |\pi_{\sigma}^{t}|$, and so $t = \pi_{\sigma}^{\vec{x}}(k-1) = \pi_t^{\vec{x}}(0)$.
Starting from $\alpha_{k-1}$, orient and attach $t$ to form $\alpha_{k}'$ such that a tile with the same type as $t_{k+1} = \pi_t^{\vec{x}}(1)$ can attach to $\alpha'_k$ at a location that is north or west of $t$. That is, attach $t$ in such a way that the output glue is exposed on the north or west edge of $t$. Note that this is possible since $t$ is the first tile located outside of $B_c(\vec{0})$ and the location of $t$ is to the north of $B_c(\vec{0})$. Then, $\alpha_{k+1}'$ is obtained from $\alpha'_k$ by attaching a tile with the same type as $t_{k+1}$ to the north or west of $t$ in $\alpha'_k$ such that the output glue of $t_{k-1}$ is also exposed on the north or west edge. In general, for a step $j$ such that $k+1<j<n$, $\alpha'_j$ is obtained from $\alpha'_{j-1}$ by attaching a tile with the same type as $\pi_t^{\vec{x}}(j-1)$ to the tile $t'$ at the location in $\dom \alpha'_{j-1} \setminus \dom \alpha'_{j-2}$ to the north or west of $t'$. In other words, we form a path with the same tile types as $\pi_{\sigma}^{\vec{x}}$, only as tiles attach, flip them so that the next tile to attach does so at a tile location that is north or west of the previously attached tile. Denote the resulting path as $\pi'$. Figure~\ref{fig:pumpable-path} (left) gives an example of a modified path.

\newcolumntype{M}{>{\centering\arraybackslash}m{\dimexpr.25\linewidth-2\tabcolsep}}
\begin{table}[htp]
\centering
\begin{tabular}{| M | M | M | M |}
	\hline
	\includegraphics[scale=.12]{images/line1}  & \includegraphics[scale=.12]{images/line2} &
	\includegraphics[scale=.12]{images/line3}  & \includegraphics[scale=.12]{images/line4} \\
	(a) & (b) & (c) & (d) \\\hline
	\includegraphics[scale=.12]{images/line5}  & \includegraphics[scale=.12]{images/line6} &
	\includegraphics[scale=.12]{images/line7}  & \includegraphics[scale=.12]{images/line8} \\
	(e) & (f) & (g) & (h) \\\hline
\end{tabular}\caption{Each figure in this table depicts a possible path that can assemble in $\mathcal{T}$. The yellow tiles make up the seed and the green tiles are a path leading to the repeated tile that allows the path to repeat.}\label{tbl:paths}
\end{table}

Note that since $\pi_t^{\vec{x}}$ contains at least $4|T|+1$ nodes, $\pi'$ contains at least $4|T| + 1$. By the pigeonhole principle, there must be two nodes $\vec{v}_0$ and $\vec{v}_1$ contained in $\pi'$ with locations in $\overline{B_{c}(\vec{0})}$ that correspond to the same tile type in the same orientation. These are depicted as blue tiles in Figure~\ref{fig:pumpable-path}. Let $\pi_{\vec{v}_0}^{\vec{v}_1}$ denote the path in $G_{\alpha}$ that is the subpath of $\pi_{\sigma}^{\vec{x}}$ from $\vec{v}_0$ to $\vec{v}_1$. Because each tile of $\pi'$ attaches to the north or west of a previous tile in the path, $\pi'$ can be modified to extend indefinitely to the north-west since the subpath $\pi_{\vec{v}_0}^{\vec{v}_1}$ can be repeated an arbitrary number of times. See Table~\ref{tbl:paths}(a) for an example of repeating $\pi_{\vec{v}_0}^{\vec{v}_1}$.

Now, either $\pi_{\vec{v}_0}^{\vec{v}_1}$ turns or it does not. If it does not turn, then  $\pi_{\sigma}^{\vec{x}}$ can be modified so that the tile types of the subpath $\pi_{\vec{v}_0}^{\vec{v}_1}$ are repeated indefinitely starting from $\vec{v}_0$ and this modified path lies on a vertical (or horizontal) line. If each such path $\pi_{\sigma}^{\vec{x}}$ from the seed $\sigma$ to some point $\vec{x}$ is such that for each pair of distinct nodes of $\pi_{\sigma}^{\vec{x}}$ with corresponding tiles of the same type and orientation, the path segment between these nodes does not turn, then every path containing a node in $\overline{B_{c'}(\vec{0})}$ can be modified to assemble an arbitrarily long path that lies on a vertical or horizontal line. Then, it can easily be seen that $\alpha$ is the finite union of semi-doubly periodic sets. This is due to the fact that in the limit, $\alpha$ must consist of $\sigma$ and finitely many infinitely long paths connected to $\sigma$. Each of the infinitely long paths is a semi-doubly periodic set, and there can only be finitely many paths that lie on vertical or horizontal lines that can assemble outside of $B_c(\vec{0})$. On the other hand, if there is a $\pi_{\sigma}^{\vec{x}}$ such that the subpath $\pi_{\vec{v}_0}^{\vec{v}_1}$ turns, then we can use this turn and the non-orientable nature of tiles in the \rtam/ to show that $\alpha$ is the finite union of semi-doubly periodic sets.

Let $\beta'$ denote $\sigma\cup \pi'$. Notice that $\beta' \in \prodasm{T}$. We now use $\pi_{\vec{v}_0}^{\vec{v}_1}$ to form an infinite set of producible assembles in $\mathcal{T}$ by repeating the tiles of $\pi_{\vec{v}_0}^{\vec{v}_1}$ and changing tile orientations by reflecting.
Without loss of generality, we assume that output glues of the tiles at locations $\vec{v}_0$ and $\vec{v}_1$ are on the north edges of the tiles at those locations. (An analogous argument holds when the output glues are on the west edges.) Table~\ref{tbl:paths} depicts this. Table~\ref{tbl:paths} (b) through (h) shows seven of the various paths that can assembly in $\mathcal{T}$ from the path depicted in Table~\ref{tbl:paths}(a).  We can construct these paths as follows. Since tiles of $\pi_{\vec{v}_0}^{\vec{v}_1}$ always attach to the north or west, tile types of $\pi_{\vec{v}_0}^{\vec{v}_1}$ can be repeated (in the order that they appear in the path $\pi_{\vec{v}_0}^{\vec{v}_1}$ and with the same orientations) indefinitely in the north-west direction to yield a producible assembly. This gives (a) in Table~\ref{tbl:paths}. Now, as as tile types of $\pi_{\vec{v}_0}^{\vec{v}_1}$ are repeated, at any arbitrary number of repetitions, we can reflect the tile types of $\pi_{\vec{v}_0}^{\vec{v}_1}$ about a vertical axis so that the reflection of the path $\pi_{\vec{v}_0}^{\vec{v}_1}$ assembles instead of $\pi_{\vec{v}_0}^{\vec{v}_1}$. Let $\hat{\pi}_{\vec{v}_0}^{\vec{v}_1}$ denote the reflection of $\pi_{\vec{v}_0}^{\vec{v}_1}$ about a vertical axis. Then, tiles of $\hat{\pi}_{\vec{v}_0}^{\vec{v}_1}$ can be repeated an indefinite number of times to give (b) through (d) in Table~\ref{tbl:paths}. Note that each path of tiles is producible in $\mathcal{T}$ (though perhaps not in the same assembly). Similarly, we can produce paths (e) through (h) in $\mathcal{T}$ by first forming a path where each successive tile attaches to the north-east of the previous tile. This gives an infinite set of of paths such that each path belongs to a producible assembly of $\mathcal{T}$. Denote this set by $S$. Note that $S$ is an infinite subset of $\prodasm{T}$.

\begin{figure}[htp]
\begin{center}
\includegraphics[scale=.11]{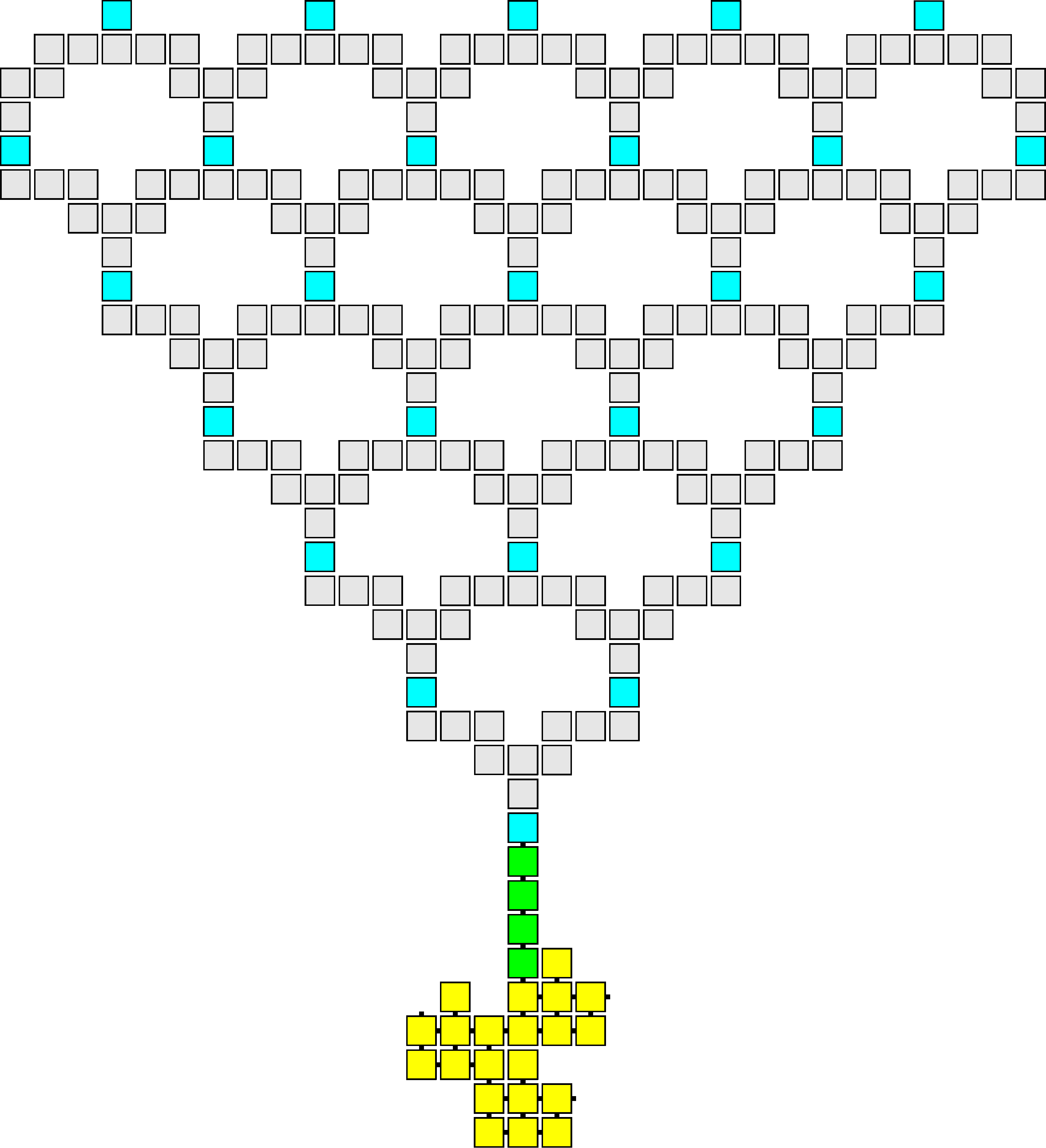}
\caption{A portion of the configuration $\Gamma$.  This is a map from $\Z^2$ to $T$ obtained by considering tile types from each path given in Table~\ref{tbl:paths}.}
\label{fig:combs}
\end{center}

\end{figure}

Notice that the assemblies in $S$ must be pairwise type-consistent since $\mathcal{T}$ is directed. Consider the map $\Gamma:\Z^2 \to T$ defined by $\vec{x} \mapsto a$ iff $\exists \gamma \in S$ such that $p_T(\gamma(\vec{x})) = a$. It should be noted that the fact that assemblies in $S$ are pairwise type-consistent implies that $\Gamma$ is well-defined. It should also be noted that $\Gamma$ is equal to the map $p_T\circ \alpha|_{\cup_{\gamma\in S} \dom \gamma}$. Notice that it may be the case that there does not exist an assembly $\phi \in \prodasm{T}$ such that $\Gamma = p_T\circ \phi$ since orientations of tiles of the same type at a particular location may differ. Figure~\ref{fig:combs} depicts a portion of this configuration. The points in the domain of $\Gamma$ consist of the union of a finite number of points and a semi-doubly periodic set. This can be seen by considering the vectors defined by $\vec{u} = \vec{v_1} - \vec{v_0} = (a,b)$ and $\vec{v} = (-a,b)$.

\newcolumntype{M}{>{\centering\arraybackslash}m{\dimexpr.25\linewidth-2\tabcolsep}}
\begin{table}[htp]
\centering
\begin{tabular}{| M | M | M |}
	\hline
	\includegraphics[scale=.1]{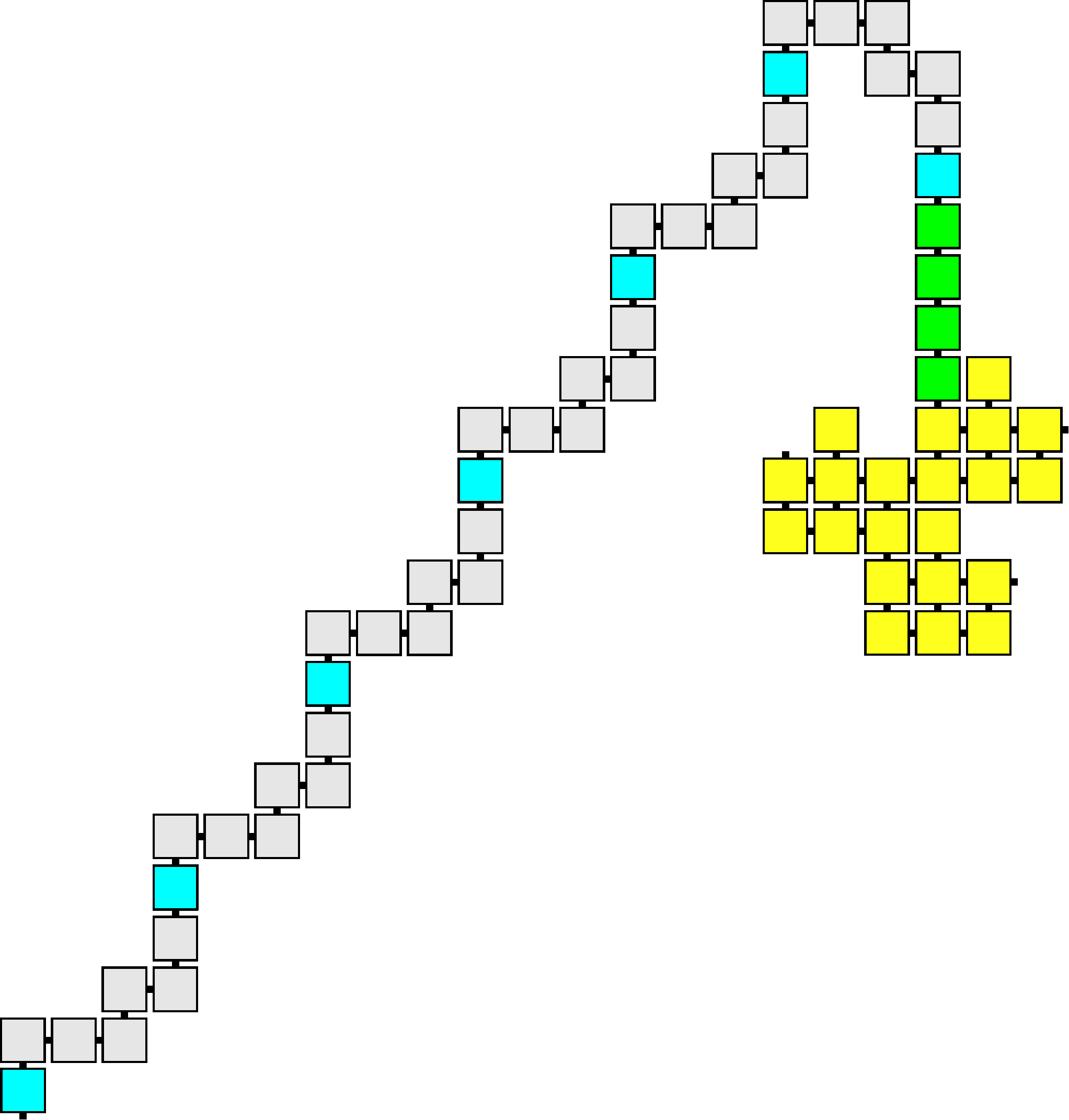}  & \includegraphics[scale=.1]{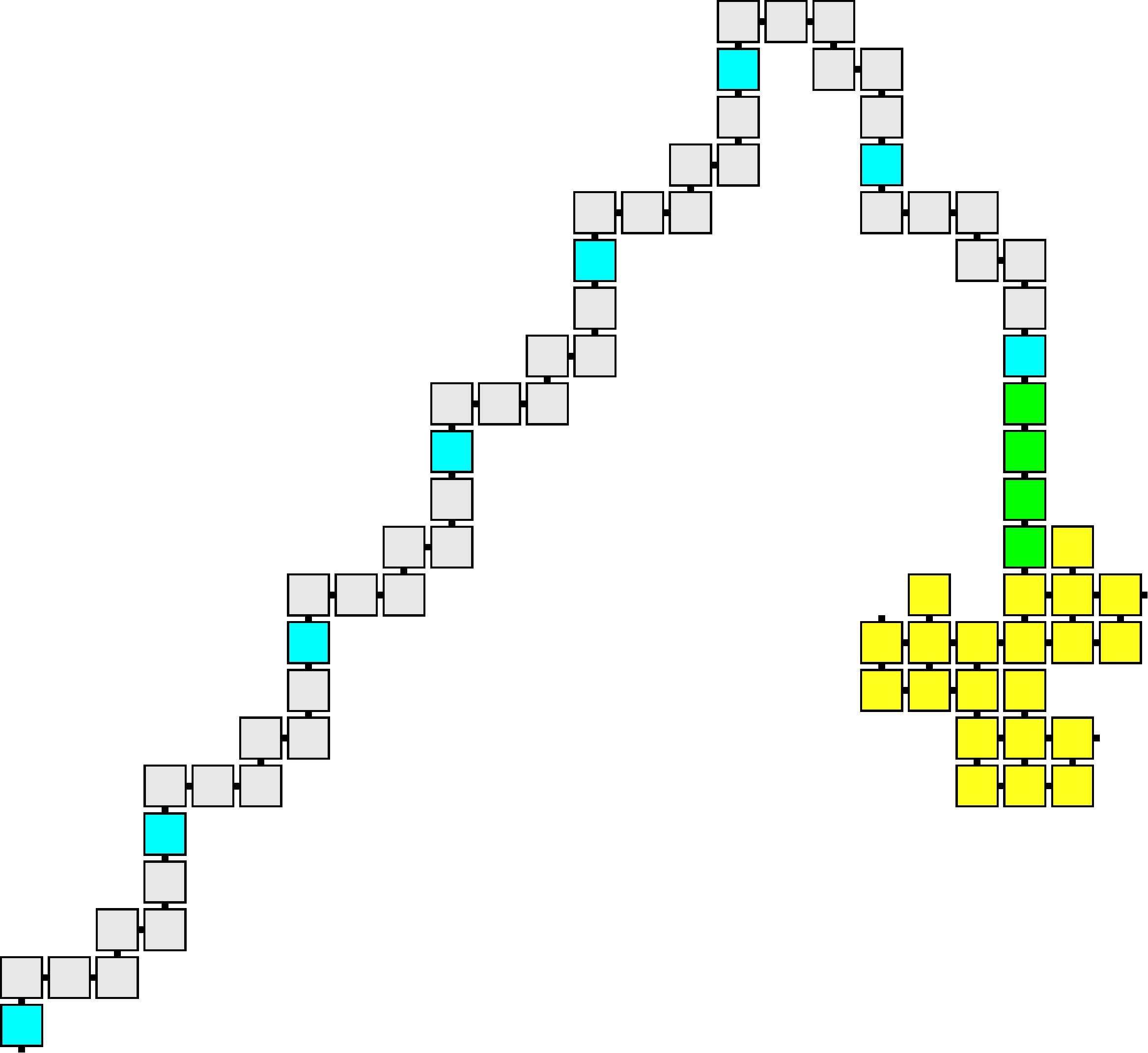} & 	\includegraphics[scale=.1]{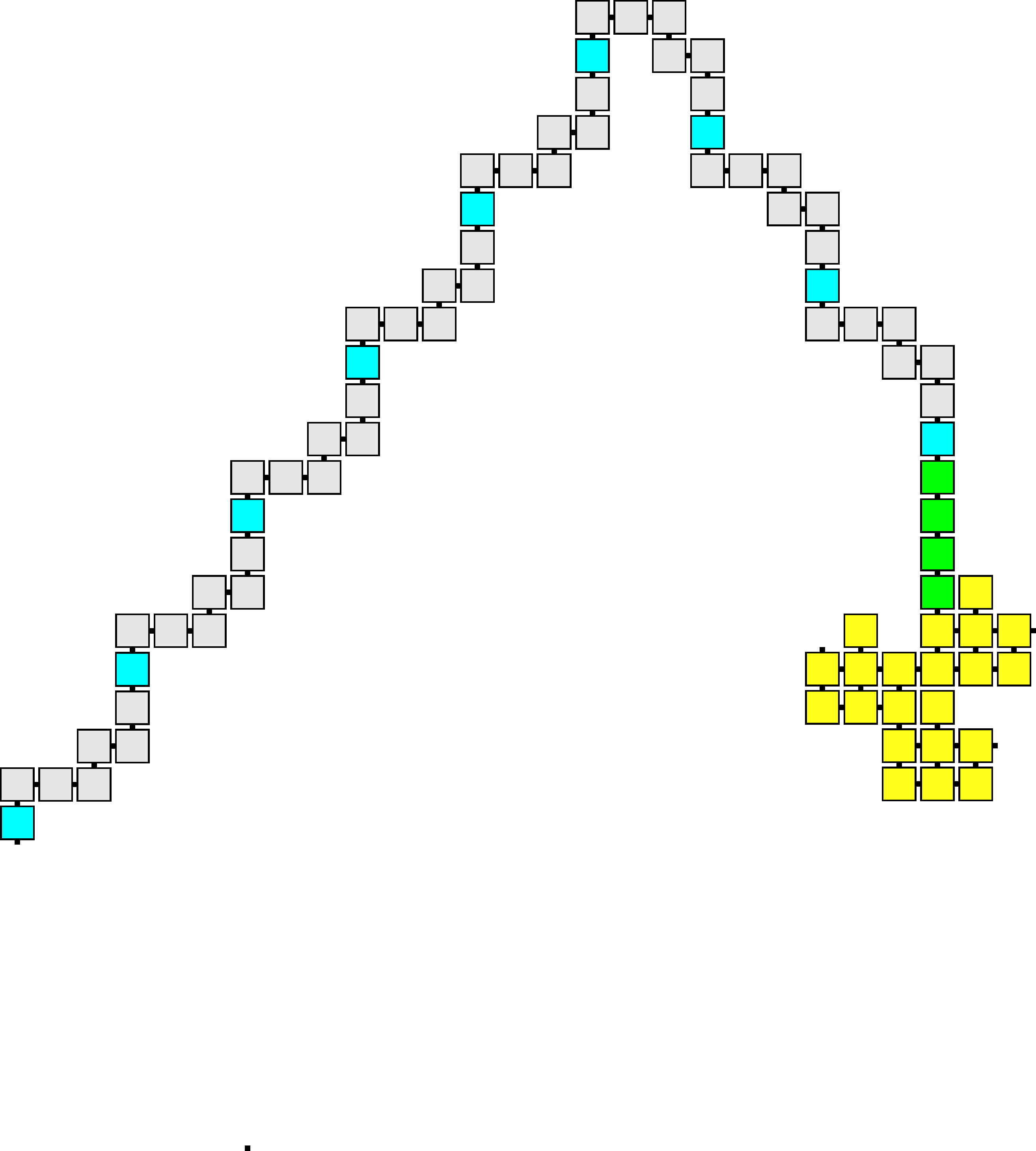} \\
	(a) & (b) & (c) \\\hline
	\includegraphics[scale=.1]{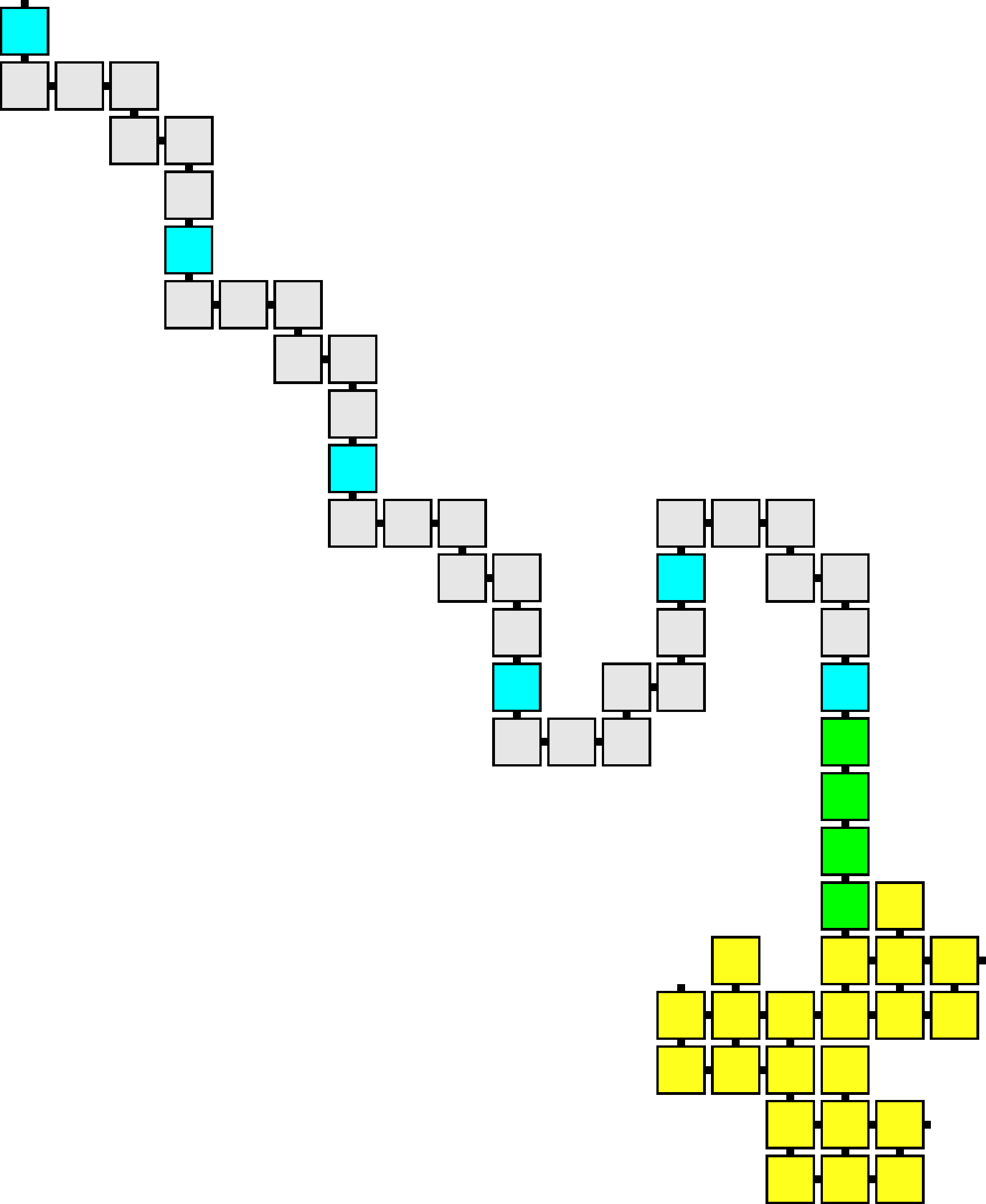}  & \includegraphics[scale=.1]{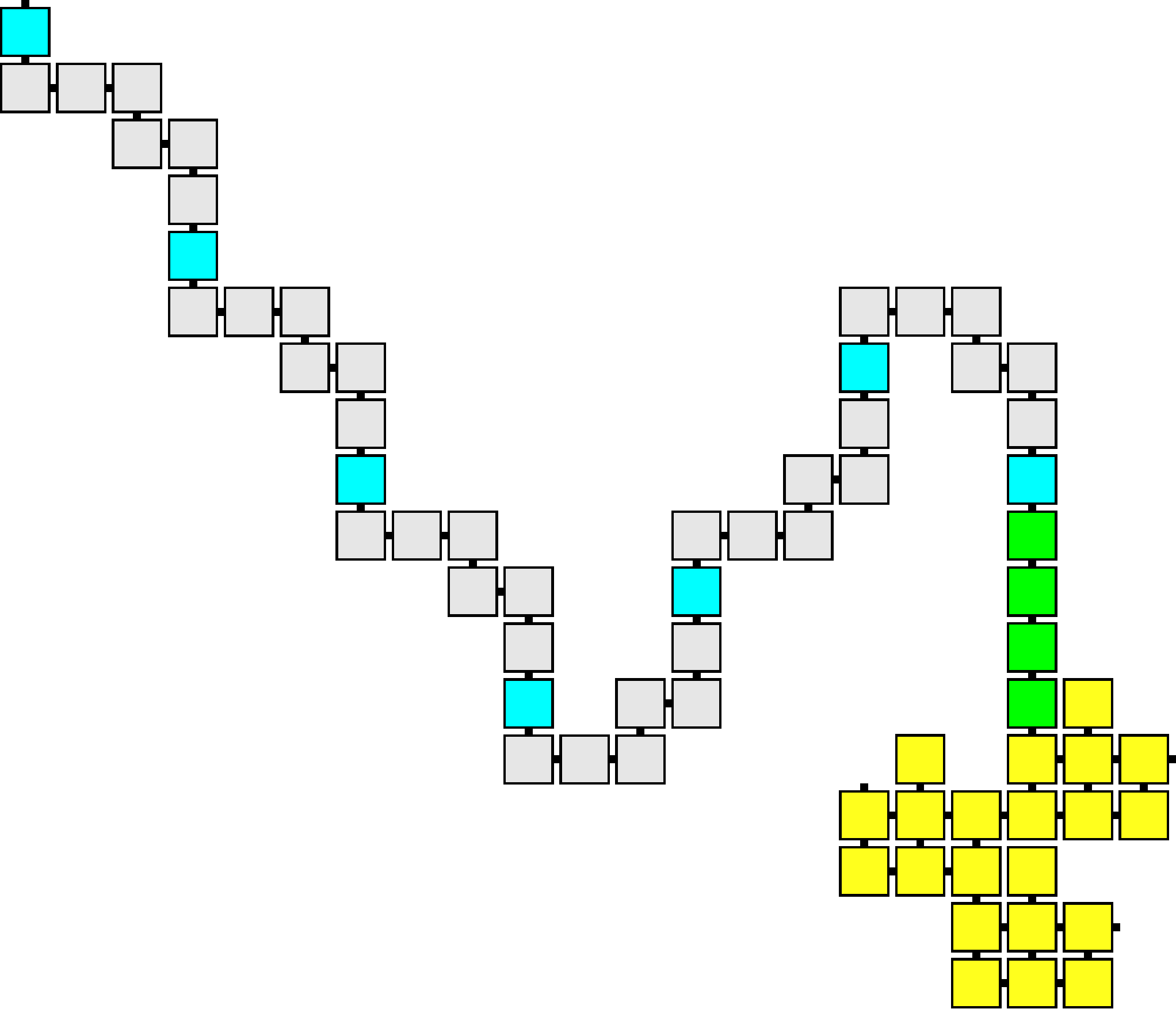} &	\includegraphics[scale=.1]{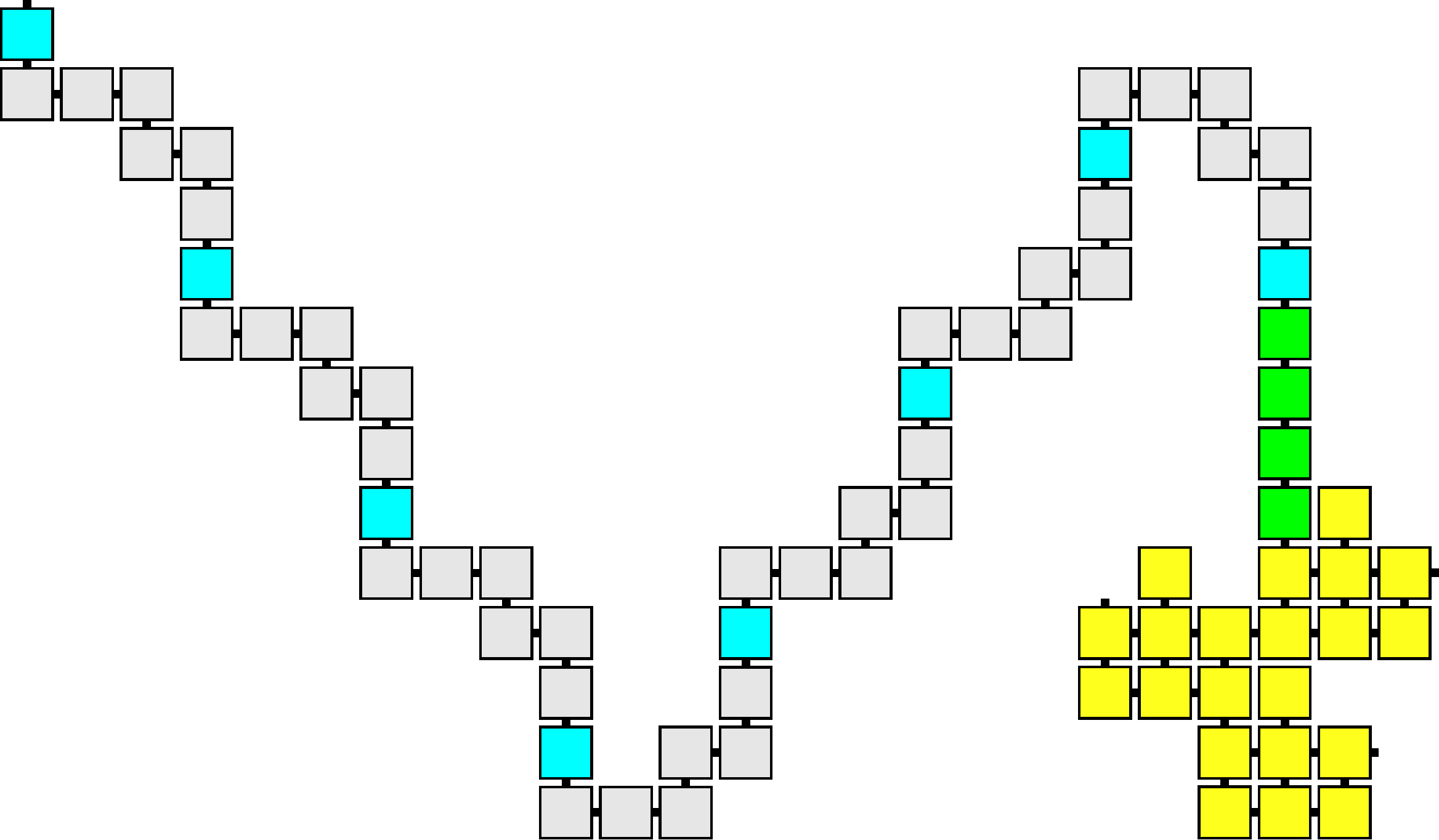} \\
	(e) & (f) & (g) \\\hline
\end{tabular}\caption{Each figure in this table depicts a possible path that can assemble in $\mathcal{T}$. The yellow tiles make up the seed and the green tiles are a path leading to the repeated tile that allows the path to repeat.}\label{tbl:paths2}
\end{table}

Using a similar technique that is used to form the assemblies of $S$, the path $\pi_{\vec{v}_0}^{\vec{v}_1}$ can also repeated (with appropriate reflections) to form type-consistent assemblies containing paths whose domains collectively form a set of points $W$ such that $W$ is the union of four semi-doubly infinite sets. Table~\ref{tbl:paths2} shows six examples of such paths. Because the assemblies whose domains make up $W$ are type-consistent, we can once again give a well-defined partial map $p_T\circ \alpha|_W$ from $\Z^2$ to $T$. Figure~\ref{fig:pump-with-seed} depicts a portion of a configuration given by the map $p_T\circ \alpha|_W$ from $\Z^2$ to $T$. Note that any point of $\Z^2$ is either in $W$ or is in some finite region bounded by points of $W$.

\begin{figure}[htp]
\begin{center}
   \includegraphics[width=3in]{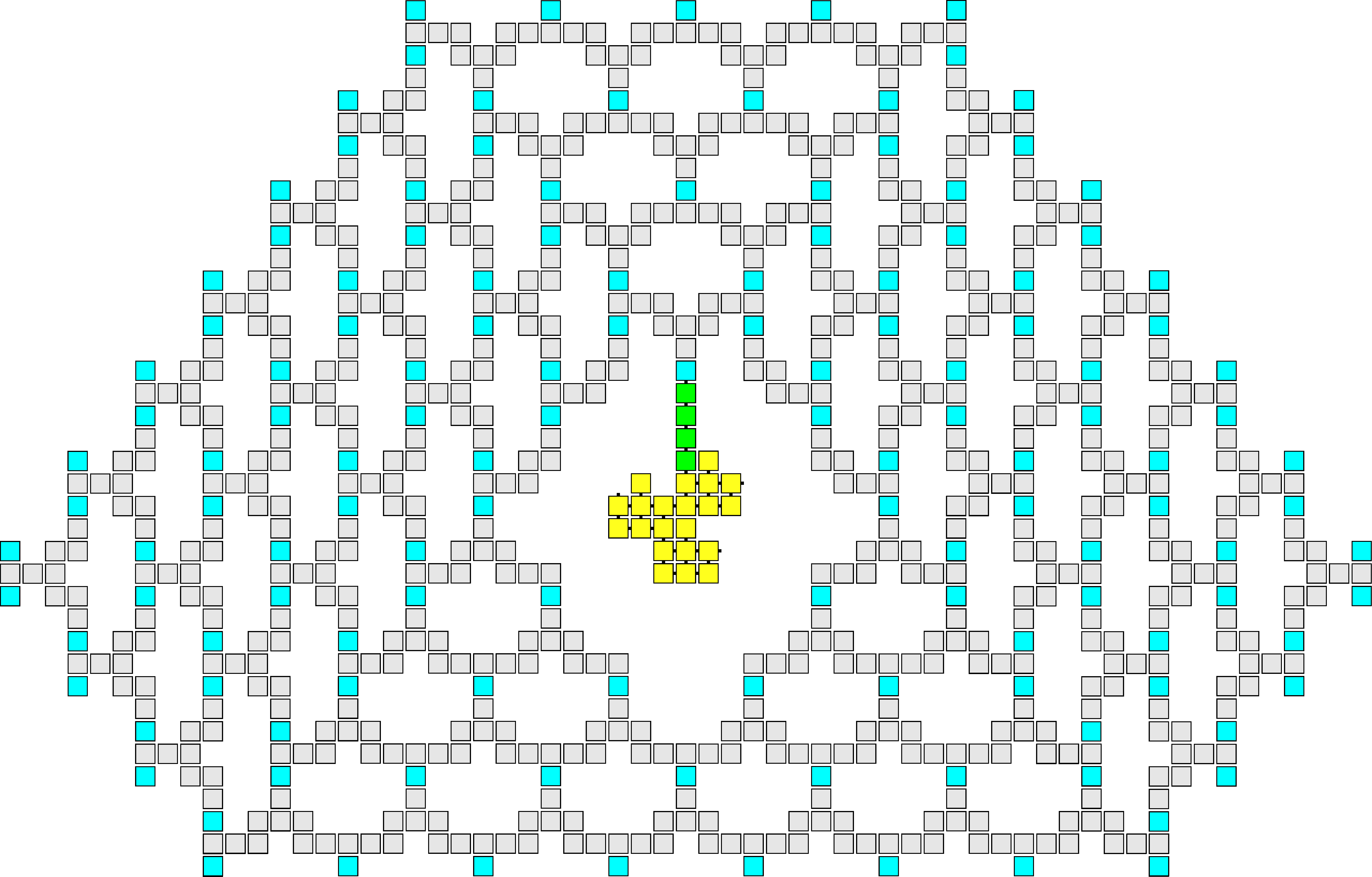}
\caption{A configuration of tiles that must weakly self-assemble in $\mathcal{T}$.}
\label{fig:pump-with-seed}
\end{center}

\end{figure}

To finish the proof, it now suffices to show that the set of all points in $\dom \alpha$ (not just the points in $W$) is the finite union of semi-doubly periodic sets. First, note that except for possibly a finite number of points, any point in $\dom \alpha$ is contained in some ``parallelogram'' region $R$ bounded by points of $W$, where these bounding points are contained in a semi-doubly periodic set with vectors $\vec{b}$, $\vec{u}$, and $\vec{v}$ as in Definition~\ref{def-doubly-periodic-set}. Note that for each path whose domain makes up $W$ (see Table~\ref{tbl:paths} for examples of these paths), we can orient the tiles of the path so that the exposed glues of each tile of each path are always oriented the same way. This ensures that if some tile can be placed in the region $R$, then the same tile can be placed in the region $R_{n,m} = \{  \vec{r} + n \cdot \vec{u} + m \cdot \vec{v} \mid \vec{r} \in R \}$ for each $n,m\in \N$. See Figure~\ref{fig:periodic-filling} for more details. Therefore, the set of all points in $\dom \alpha$ (even those points that are not contained in $W$) is the finite union of semi-doubly periodic sets, and hence, $\alpha$ is the finite union of semi-doubly periodic sets.

\begin{figure}[htp]
\begin{center}
\def\svgwidth{3in}
   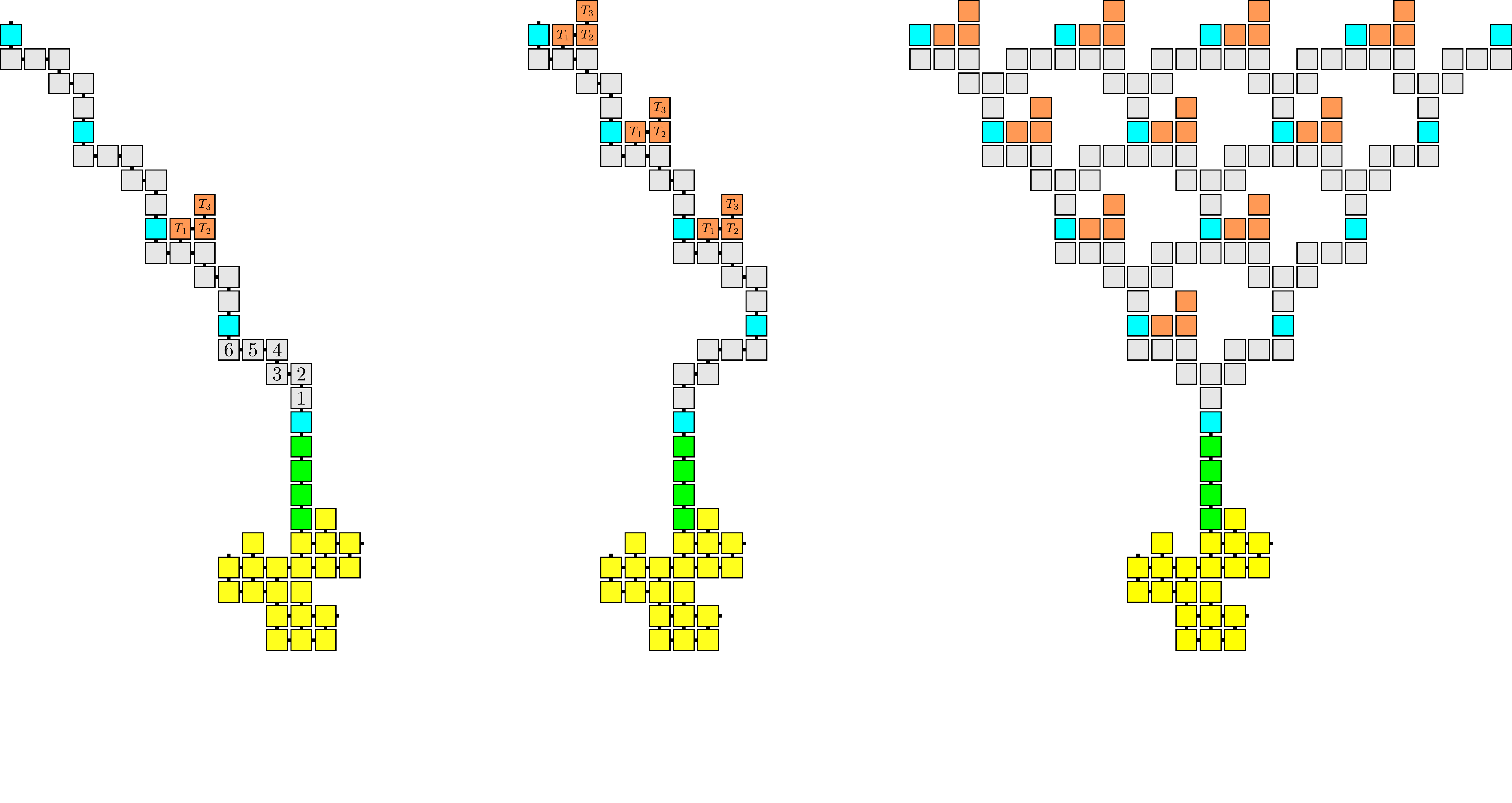
\caption{(a) Tiles $T_1$, $T_2$, and $T_3$ have locations in a region $R$ bounded by points in $W$. (b) We can choose tile orientations such that tiles with the same types as $T_1$, $T_2$, and $T_3$ can bind in regions $R_{n,m}$ for each $n, m \in \N$. (c) A portion of a semi-doubly periodic set containing $T_1$, $T_2$, and $T_3$. }
\label{fig:periodic-filling}
\end{center}

\end{figure}

\iffull
 \end{proof}
\fi

We note that in the proof of Theorem~\ref{thm-no-computation}, we actually explicitly describe the semi-doubly periodic sets that can assemble in an \rtam/ system. This is summed up by the following corollary.

\begin{corollary}
\label{thm-no-computation2} Let $\mathcal{T} = (T,\sigma,1)$ be a
directed \rtam/ system. If a set $X \subseteq \mathbb{Z}^2$ weakly self-assembles in $\mathcal{T}$, then $X$ is a finite union of a box of constant (which depends only on $|T|$ and $|\sigma|$) radius containing the seed and at most $4$ (possibly empty) semi-doubly periodic sets. Moreover, each of these $4$ semi-doubly periodic sets $S_i$ for $1\leq i\leq 4$ is described as follows.

\begin{enumerate}
	\item $S_i$ is empty.
	\item $S_i$ is the finite union of $d\in\N$ horizontal or vertical paths of tiles where $d$ is bounded by $|\sigma|$.
	\item $S_i$ is the set obtained from the proof of Theorem~\ref{thm-no-computation} by modifying a path that does not lie on a horizontal or vertical line. Note that the number of such sets in this case is bounded by $|T|$.
\end{enumerate}
\end{corollary}

\begin{proof}
Note that the number of paths $\pi$ as described in the proof of Theorem~\ref{thm-no-computation} that may assemble from a single glue exposed by $\sigma$ is bounded by $|T|$. Therefore, there is a constant $\kappa$ that depends only on $|\sigma|$ and $|T|$ such that for the terminal assembly $\alpha$ of $\mathcal{T}$, $\alpha|_{\overline{B_\kappa(0)}}$ (the portion of $\alpha$ outside of a box of radius $\kappa$) is the union of $4$ semi-doubly periodic sets. We may take $\kappa$ to be $c'$ given in the proof of Theorem~\ref{thm-no-computation} given in Section~\ref{thm-no-computation}.
\end{proof}
} %

An intuitive reason that Theorem~\ref{thm-no-computation} supports the conclusion that the \rtam/ is not computationally universal is as follows. Let $H = \{(i,x)\mid \text{program }$\ $ i \text{ halts when run on input } x\}$ be the halting set and let $M$ be a Turing machine that outputs a $1$ if $(i,x)\in H$. The typical way of expressing the computation of $M$ in tile assembly is as follows. For a fixed tileset $T$, a seed assembly $\sigma$ encodes an ``input'' to the computation, while the ``output'' of $1$ by $M$ corresponds to the translation of some configuration being contained in the terminal assembly $\alpha_\sigma$ of $(T,\sigma, 1)$. For this sense of computation, the following corollary says that the set of seed assemblies that ``output'' a $1$ is a recursive set (not just a recursively enumerable set). This would contradict the fact that the halting set is not recursive. This is stated in the following corollary. For a more formal statement of this corollary and a proof, see Section~\ref{sec:thm-no-computation-proof}.

\begin{corollary}\label{cor:recusive-set}
For any tileset $T$ in the \rtam/ and fixed finite configuration $C$, let $S$ be the set of seed assemblies $\sigma$ such that (1) the \rtam/ system $(T, \sigma, 1)$ is directed and (2) the terminal assembly of $\mathcal{T}$ contains $C$. Then, $S$ is a recursive set.
\end{corollary}

\later{
Now we give a corollary that provides evidence to the conclusion that the directed temperature-1 \rtam/ is not computationally universal.
For a tileset $T$, fix a finite configuration $C_T$ over $T$, and let $C_T^{\vec{v}}$ be the translation of $C_T$ by the vector $\vec{v}$. In other words, $C_T$ is a partial function from $\Z^2$ to $T$, and $C_T^{\vec{v}}$ is the configuration given by the partial map $C_T^{\vec{v}}(x) = C_T(x - \vec{v})$ for $x\in \{z \mid z - \vec{v}\in\dom(C_T)\}$. Then, let $A(C_T)$ be the set of all seed assemblies $\sigma$ such that
the \rtam/ system $(T,\sigma,1)$ is directed and the terminal assembly $\alpha_\sigma$ of the \rtam/ system $(T,\sigma,1)$ contains some translation of the configuration $C_T$. More formally, $A(C_T)$ is the set of all seed assemblies $\sigma$ such that the \rtam/ system $(T, \sigma, 1)$ is directed and the terminal assembly $\alpha_\sigma$ of \rtam/ system $(T,\sigma,1)$ has the property that there exists a vector $\vec{v}$ in $\Z^2$ such that $p_T\circ \alpha_{\sigma}(x) = C_T^{\vec{v}}(x)$ for all $x\in \dom(C_T^{\vec{v}})$.

Intuitively, we are considering the set $C_T$ (and its translations) for the following reason. Let $H = \{(i,x)\mid \text{program } i \text{ halts when run on input } x\}$ be the halting set and let $M$ be a Turing machine that outputs a $1$ if $(i,x)\in H$. The typical way of expressing the computation of $M$ in tile assembly is as follows. For a fixed tileset $T$, a seed assembly $\sigma$ encodes an ``input'' to the computation, while the ``output'' of $1$ by $M$ corresponds to some configuration being contained in the terminal assembly $\alpha_\sigma$ of $(T,\sigma, 1)$. For this sense of computation, the following corollary says that the set of seed assemblies that ``output'' a $1$ is a recursive set (not just a recursively enumerable set). This would contradict the fact that the halting set is not recursive.
The following corollary is a restatement of Corollary~\ref{cor:recusive-set}.

\begin{corollary}
For any tileset $T$ in the \rtam/ and configuration $C_T$ over $T$, $A(C_{T})$ is a recursive set.
\end{corollary}

\begin{proof}(sketch)
In each of the cases for the sets $S_i$ described in Corollary~\ref{thm-no-computation2}, to determine if some $C_{T}^{\vec{v}}$ is contained in one of the $S_i$ subassemblies one need only check whether or not a finite portion of $S_i$ contains a translation of $C_{T}$, which can be done in a finite number of steps depending only on $|\sigma|$ and $|T|$.
\end{proof}

} %

\subsection{Universal computation at $\tau=2$}\label{sec:temp2-comp}
\vspace{-8pt}
In this section give a theorem that states that universal computation is possible in the RTAM at temperature $2$. We give an example of simulating a binary counter in the \rtam/ and give the general proof in Section~\ref{sec-temp2-computation} in the appendix.
\vspace{-4pt}
\begin{theorem}\label{thm:comp-univ}
The \rtam/ is computationally universal at $\tau=2$. Moreover, the class of directed \rtam/ systems is computationally universal at $\tau=2$.
\end{theorem}
\vspace{-4pt}

First we note that given a Turing machine $M$, we use Lemma 7 of \cite{CookFuSch11} to obtain a tile set which simulates $M$ using a \emph{zig-zag system}.  In fact, as noted in \cite{SingleNegative}, we can find a singly seeded \emph{compact zig-zag system} $\calT = (T, \sigma, 2)$ with $\termasm{T}=\{\alpha\}$ which simulates $M$.
Then the proof of Theorem~\ref{thm:comp-univ} relies on showing that any compact zig-zag system in the aTAM at temperature $2$ can be converted into a directed \rtam/ system $\cal S$ that is ``almost'' compact zig-zag. The \rtam/ system that we construct differs from a compact zig-zag system in that when the length of a row of the growing zig-zag assembly increases by a tile, a strength-$2$ glue is exposed that allows a tile to bind below the row. This results in the possibility of a single ``misplaced'' tile per row, but nevertheless, this is enough to simulate a Turing machine. The proof of Theorem~\ref{thm:comp-univ} can be found in Section~\ref{sec-temp2-computation}

\later{
\section{Proof of Theorem~\ref{thm:comp-univ}}\label{sec-temp2-computation}

In this section we present the proof that the \rtam/ is computationally universal at $\tau=2$.
As previously stated, we note that given a Turing machine $M$, we use Lemma 7 of \cite{CookFuSch11} to obtain a tile set which simulates $M$ using a \emph{zig-zag system}.  In fact, as noted in \cite{SingleNegative}, we can find a singly seeded \emph{compact zig-zag system} $\calT = (T, \sigma, 2)$ with $\termasm{T}=\{\alpha\}$ which simulates $M$.
Then the proof of Theorem~\ref{thm:comp-univ} relies on showing that any compact zig-zag system in the aTAM at temperature $2$ can be converted into a directed \rtam/ system $\cal S$ that is ``almost'' compact zig-zag. The \rtam/ system that we construct differs from a compact zig-zag system in that when the length of a row of the growing zig-zag assembly increases by a tile, a strength-$2$ glue is exposed that allows a tile to bind below the row. This results in the possibility of a single ``misplaced'' tile per row, but nevertheless, this is enough to simulate a Turing machine.
In addition, notice that the orientation of the first tile which binds to the seed tile is such that a strength 2 glue is exposed on either the west or the east.  Since the two assemblies obtained from the different binding orientations of this tile are the same (up to reflection), this does not affect the assembly which $\mathcal{S}$ produces. First, we give an example of how to simulate a zig-zag.

\parshape 1
0pt \dimexpr0.6\textwidth
Figure~\ref{fig:t2Example} shows how we convert a tile set used in a system which simulates a counter in the aTAM at $\tau =2$ (part (a) of the figure) into a tile set used in a system which simulates a counter in the RTAM at $\tau=2$ (part (b) of the figure).  Part (c) of Figure~\ref{fig:t2Example} shows a part of the assembly of the RTAM system $S$ which simulates the counter.

\parshape 1
0pt \dimexpr\textwidth\

\begin{figure}[htp]
\begin{center}
\vspace{-120pt}
   \includegraphics[width=3.5in]{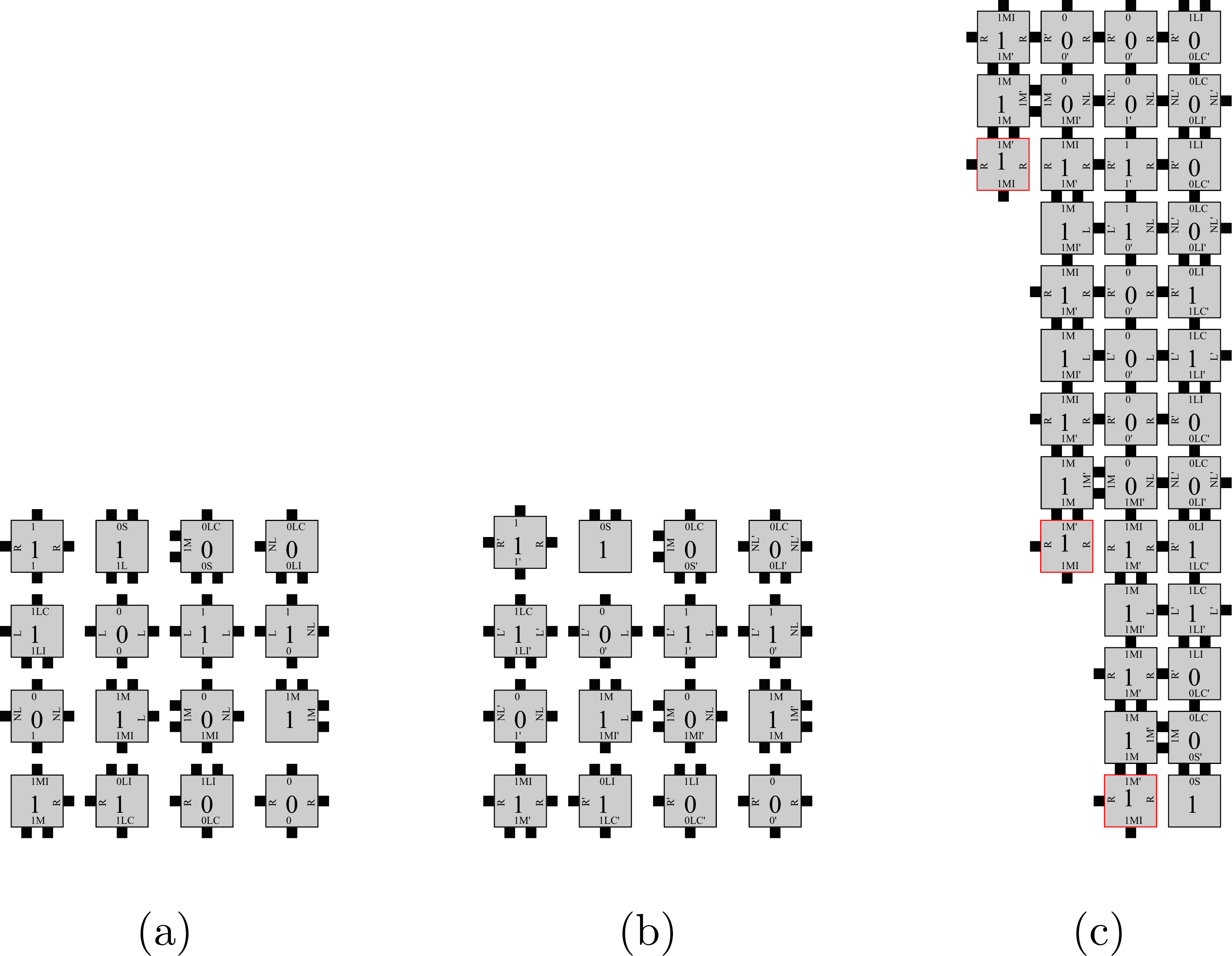}
\caption{Part (a) of this figure shows a tile set which assembles a binary counter in the aTAM, part (b) shows the converted tile set from (a) which assembles a binary counter in the RTAM, and part (c) shows a part of the binary counter assembly from an RTAM system which has the tile set from part (b). A single black box on the side of a tile represents a strength 1 glue, two black boxes on the side of a tile represents a strength 2 glue, and a glue mismatch is represented by a glue label on the side of a tile and the absence of a black box on that side.  The tiles with red borders represent ``misplaced'' tiles.}
\label{fig:t2Example}
\end{center}

\end{figure}

\begin{proof}(sketch proof of Theorem~\ref{thm:comp-univ})
For the remainder of the proof, we say that two tiles $t$ and $t'$ are of the same form provided that they have glues of the same strengths on the same sides. Given a Turing machine $M$, we use Lemma 7 of \cite{CookFuSch11} to obtain a tile set which simulates $M$ using a zig-zag system.  Furthermore, as noted in \cite{SingleNegative}, we can find a singly seeded compact zig-zag system  $\calT = (T, \sigma, 2)$ with $\termasm{T}=\{\alpha\}$ which simulates $M$.  We assume that $\calT$ grows to the north.  Let $\mathcal{S} = (S, \sigma', 2)$ be an RTAM system defined a follows.

First, for clarity, we rewrite the glues contained on tiles of $T$ in terms of their complementary glues. We use the convention that if a glue $g$ appears on the north or east side of a tile, then we simply write $g$.  If a glue $g$ appears on the south or west of a tile, then we rewrite that glue as $g'$.  Next, we set $\sigma'$ to be the same as the tile which composes $\sigma$.  It follows from the constructive proof of Lemma 7 in \cite{CookFuSch11}, that for any $t \in T$ such that $t$ has a strength 2 glue on its south, it must be of the same form as tile $A$ shown in Figure~\ref{fig:t2casesS}(a) up to reflection.  For each $t \in T$ with a strength 2 glue on the south, we create a tile $t' \in S$ which consists of the same south glue as the one on $t$, but has a copy of the east/west glue of $t$ such that $t'$ has identical east and west glues.  The form of $t'$ is that of tile $A'$ shown in Figure~\ref{fig:t2casesS}(b).  Henceforth, we refer to tiles of the same form (up to reflection across the vertical axis) as the tile labeled $B$ in Figure~\ref{fig:t2casesS}(a) as \emph{corner tiles}.  For each corner tile $t \in T$, we form a tile $t' \in S$ which consists of the same east and west glues as $t$, but has a copy of the north/south glue of $t$ such that $t'$ has identical south and north glues.  Thus, in our RTAM tile set $S$, corner tiles in $T$ take on the form of the tile labeled $B'$ in Figure~\ref{fig:t2casesS}(b).  Notice that tiles which require cooperation to bind are necessarily oriented upon binding.  Consequently, for all tiles $t \in T$ which bind cooperatively (i.e. those that are not of the form $A$ or $B$ shown in Figure~\ref{fig:t2casesS}(a)), we add tile $t$ to our RTAM tile set $S$.

\begin{figure}[htp]
\begin{center}
   \includegraphics[scale=.5]{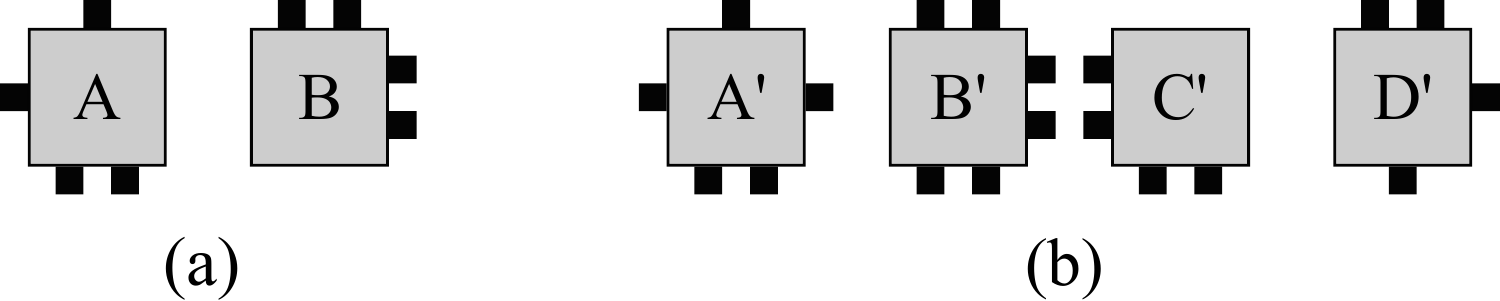}
\caption{(a) Tiles of this form (upto reflection) allow a row of zig-zag growth in the aTAM to increase in length by one tile before proceeding to the growth of the next row. (b) Tiles of this form (upto reflection) allow a row of ``almost'' zig-zag growth in the \rtam/ to increase in length by one tile before proceeding to the growth of the next row.}
\label{fig:t2casesS}
\end{center}

\end{figure}

Now, we show that the northern most row of tiles in the assembly obtained from $\mathcal{S}$ contains the final configuration of the tape of $M$ and its final state just as it is in $\calT$. Because tiles which bind cooperatively in $\mathcal{S}$ always orient themselves in the same manner in which they appear in $\calT$, we only examine the cases where $\tau$ strength glues are used in binding.  Up to reflection, all of the cases of tiles which bind with strength 2 are of the same form as tiles $A'-C'$ in Figure~\ref{fig:t2casesS}(b).  Tiles that are of the same form as the tile labeled $A'$ in Figure~\ref{fig:t2casesS}(b) do not pose a problem since the extra glue it has relative to its counterpart in $t$ is strength 1 and exposed on the east/west edges of the assembly in such a way that no other tiles can bind to expose a glue with which it can cooperate.

Next, we examine the binding of a tile of the same form as $B'$ in Figure~\ref{fig:t2casesS}(b).  Upon binding, a tile of the same form as $B'$ will allow for the binding of a tile which its counterpart in $T$ does not. That is, it will allow for the binding of a tile to its south.  But, the constructive proof of Lemma 7 in \cite{CookFuSch11} is such that only tiles of the same form as $A'$ in Figure~\ref{fig:t2casesS}(b) can flip and bind to the south of the $B'$ tile.  Now, observe that due to the zig-zag nature of growth in $S$, the tile located diagonally to the south-east of the misplaced $A'$ tile is of the same form as the tiles $B'$ or $D'$ in Figure~\ref{fig:t2casesS}(b).  Notice that this means the misplaced $A'$ tile cannot cooperatively place a tile with any other glue.

It follows from the compact zig-zag Turing machine simulation construction used in \cite{SingleNegative}, that tiles of the form $C'$ appear only once in $\alpha$.  More specifically, tiles of the same form as $C'$ always attach to the seed tile.  Thus, in our case it does not matter which orientation $C'$ has upon binding since the two assemblies obtained from allowing $C'$ to bind in its two orientations are the same up to reflection.

Consequently, the ``misplaced'' tiles in the assembly produced by $\mathcal{S}$, that is the tiles that differ from the tiles in the assembly produced by $\calT$ in location or label, are tiles that bind to the south of corner tiles.  Since this does not interfere with the simulation of $M$, $\mathcal{S}$ is computationally universal. To prove the second statement of the theorem, notice that the constructed \rtam/ system is directed.
\end{proof}
} %

\section{Self-assembly of shapes in the RTAM at $\tau=1$}\label{sec:shapes}

In this section, we discuss the self-assembly of shapes in the RTAM, especially the commonly used benchmark of squares.  At temperature $2$, using a zig-zag binary counter similar to that used in Section~\ref{sec:temp2-comp}, $n \times n$ squares can be built using the optimal $\log n/(\log \log n)$ tile types following the construction of \cite{RotWin00} with only trivial modifications.  Similarly, the majority of shapes which can be weakly self-assembled in the temperature-$2$ aTAM can be built in the temperature-$2$ RTAM, although shapes with single-tile-wide branches which are not symmetric are impossible to strictly self-assemble in the RTAM.

At temperature-$1$, however, the differences between the powers of the aTAM and RTAM appear to increase. Here we will demonstrate that squares whose sides are of even length cannot weakly (or therefore strictly) self-assemble in the RTAM at $\tau=1$, although any square can strictly self-assemble in the $\tau=1$ aTAM.  We then prove a tight bound of $n$ tile types required to self-assemble an $n \times n$ square for odd $n$ in the RTAM at $\tau=1$.  (Which is, interestingly, better than the conjectured lower bound of $2n-1$ for the $\tau=1$ aTAM.)  

\subsection{For even $n$, no $n \times n$ square self-assembles in the $\tau=1$ RTAM}\label{sec:even-squares}

\begin{theorem}\label{thm:square-even}
For all $n \in \Z^+$ where $n$ is even, there exists no RTAM system $\mathcal{T}=(T,\sigma,1)$ where $|\sigma|=1$ and $\mathcal{T}$ weakly (or strictly) self-assembles an $n \times n$ square.
\end{theorem}

\begin{proof}(sketch)
We prove Theorem~\ref{thm:square-even} by contradiction, and here give a sketch of the proof.  (See Section~\ref{sec:square-even-proof} for the full proof.)   Therefore, assume that for some $n \in \Z^+$ such that $(n \mod 2) = 0$, there exists an RTAM system $\mathcal{T} = (T,\sigma,1)$ such that $|\sigma|=1$ and $\mathcal{T}$ weakly self-assembles an $n \times n$ square $S$.  We take $\alpha \in \termasm{\calT}$ and consider the corners of the square which is weakly self-assembled.

\begin{figure}[htp]

\centering
 \includegraphics[height = 1.7in]{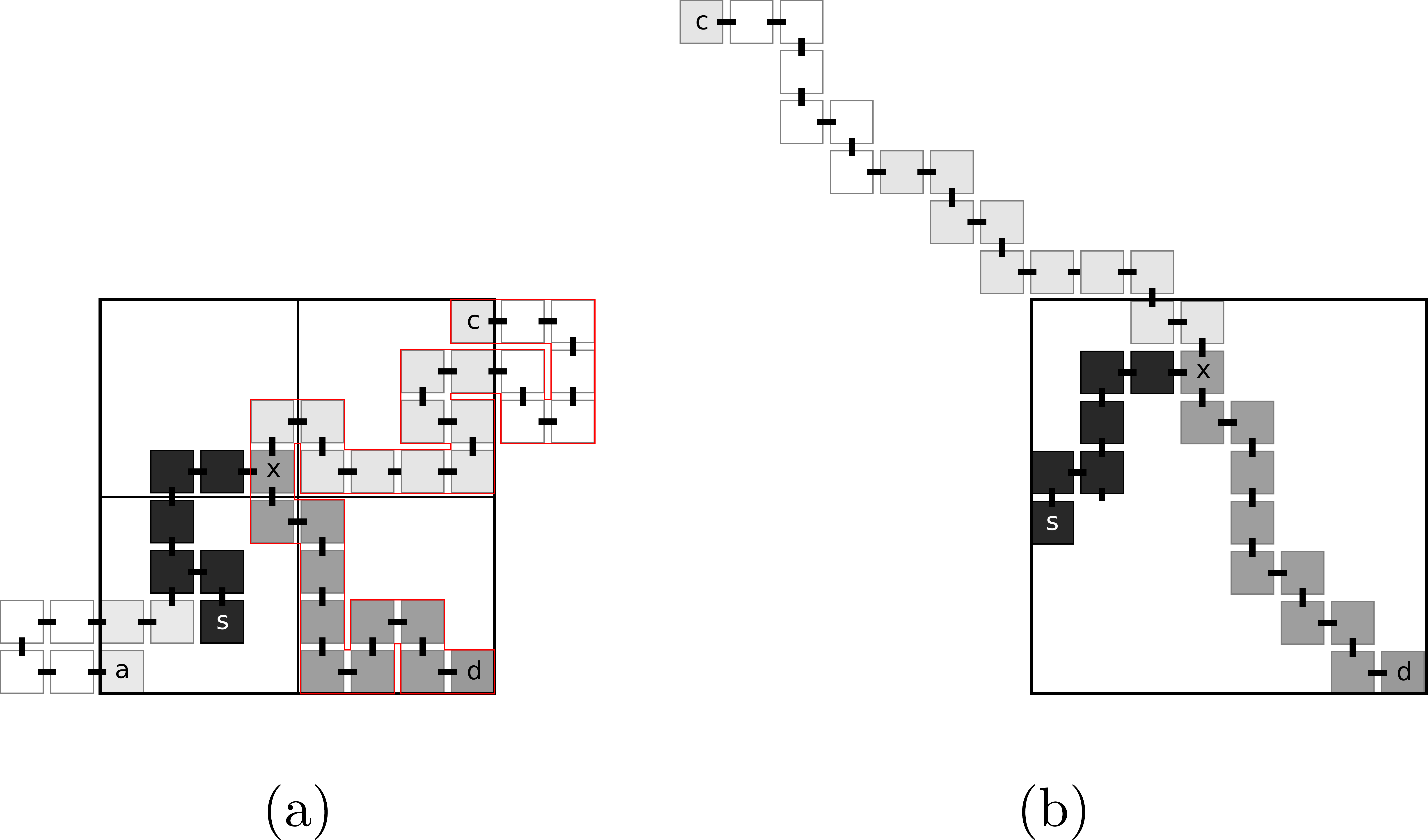}%
  \caption{(a) Example paths in an $8 \times 8$ square (i.e. of even dimension).  The path from corners $\vec{a}$ and $\vec{c}$ is composed of tiles of different colors.  The path from corner $\vec{d}$ to a point on that path is dark grey. The path from the seed to that intersection is in black. The path to be stretched out is outlined in red. (b) The stretched out version of the paths.}
  \label{fig:weak-s-a}

\end{figure}
There must exist a path which connects two diagonal corners, and that path must travel through 3 quadrants of the square since the dimensions are even length and diagonal paths are not possible in the grid graph of an assembly.  We then find a path connecting the corner of the quadrant (possibly) unvisited by that path, and note that since $\alpha$ must be connected that we can find a new path which connects the new corner to one of the original corners via a path which crosses across the midpoint of the square from either corner on the new path. (See Figure~\ref{fig:weak-s-a} (a) for an example.) Finally, we demonstrate that there is an assembly sequence which starts from the seed and grows to that new path, then builds that path in a way such that it is maximally ``stretched'' out by appropriately flipping the tiles.  (See Figure~\ref{fig:weak-s-a} (b) for an example.)  This stretched out path is producible by a valid assembly sequence but must grow beyond the bounds of the square (by at least one position), so $\calT$ does not weakly (or strictly) self-assemble the square.
\qed
\end{proof}

\later{
\section{Proof of Theorem~\ref{thm:square-even}}\label{sec:square-even-proof}

In this section, we give the full proof of Theorem~\ref{thm:square-even}.

\begin{proof}
We prove Theorem~\ref{thm:square-even} by contradiction.  Therefore, assume that for some $n \in \Z^+$ such that $(n \mod 2) = 0$, there exists an RTAM system $\mathcal{T} = (T,\sigma,1)$ such that $|\sigma|=1$ and $\mathcal{T}$ weakly self-assembles an $n \times n$ square $S$.  It must therefore be the case that for some subset of tile types $B \subseteq T$, for all $\alpha \in \termasm{\mathcal{T}}$ there exist $r \in R$ and $\vec{v} \in \Z^2$ such that for $\alpha_r = F(\alpha,r,\vec{v})$, $\alpha_r^{-1}(B) = S$ (that is, every terminal assembly contains an $n \times n$ square of ``black'' tiles, and no ``black'' tiles anywhere else).

\begin{figure}[htp]
\begin{center}
   \includegraphics[width=2.0in]{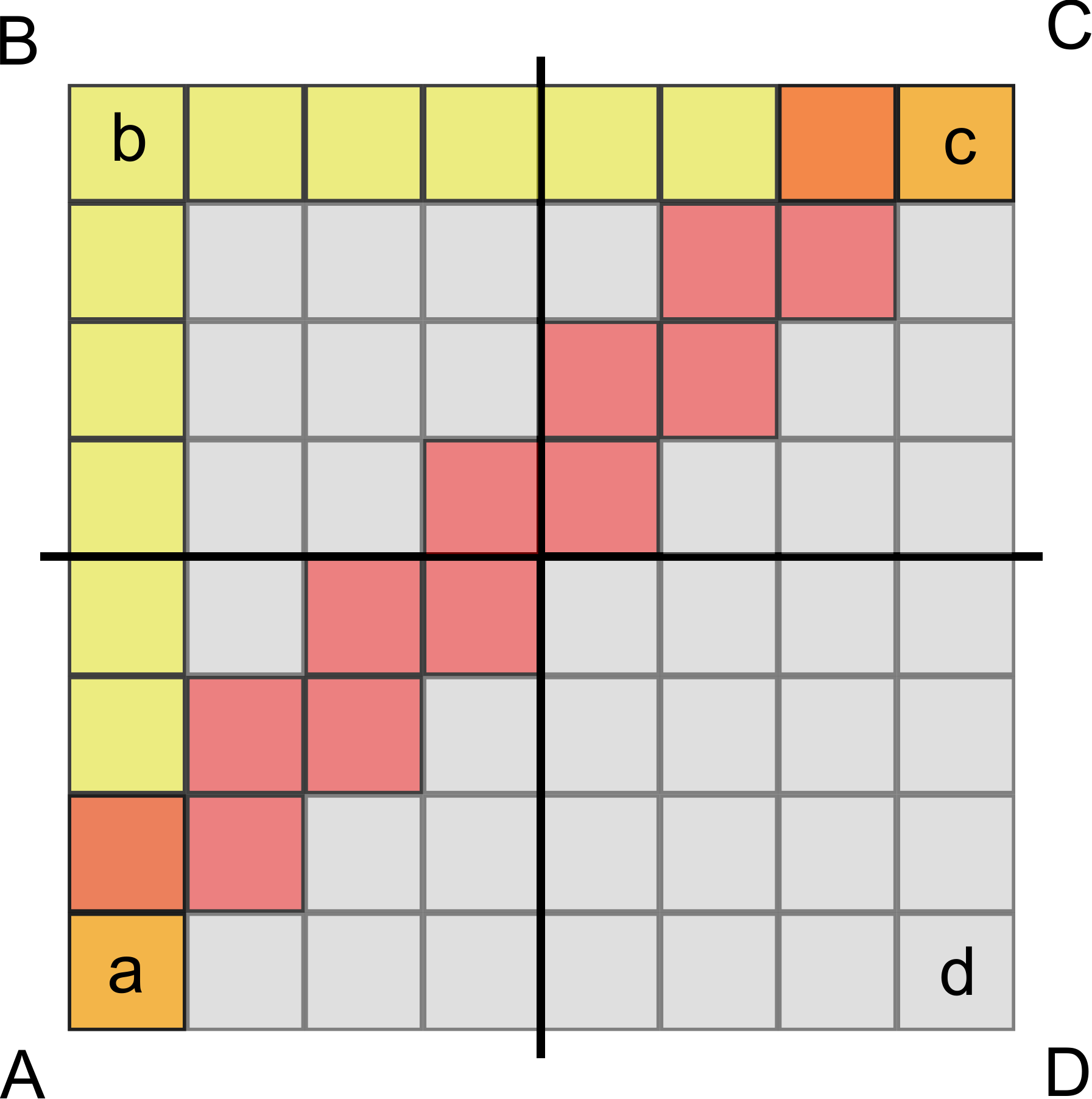}
\caption{The shortest path between two opposite corners of an $n \times n$ square has length $2n-1$.}
\label{fig:even-square-corners}
\end{center}

\end{figure}

We choose an arbitrary $\alpha \in \termasm{\mathcal{T}}$ and refer to the four corners of $S$ in $\alpha$ as $\vec{a}$, $\vec{b}$, $\vec{c}$, and $\vec{d}$, with $\vec{a}$ being the southwest and moving clockwise to name the remaining corners.  (See Figure~\ref{fig:even-square-corners} for an example using an $8\times8$ square.)  Since $\alpha$ must be a stable assembly, there must exist a set of (possibly overlapping) paths $P$ in the binding graph of $\alpha$ which contain the tiles at all $4$ corners.  Let $\pi_a^c$ be a shortest path which contains both corners $\vec{a}$ and $\vec{c}$ (there may be more than one), and note that $|\pi_a^c| \ge 2n-1$ because the tiles are located on a grid.  (Two possible such paths can be seen in Figure~\ref{fig:even-square-corners} as the red and yellow paths.)  Let the uppercase letters $A$, $B$, $C$, and $D$ represent the quadrants which contain the corners with the corresponding lowercase labels.  It must be the case that $\pi_a^c$ contains at least one tile in either quadrant $B$ or $D$, else it could not reach quadrant $C$.  This is due to the facts that $n$ is even and that $\pi_a^c$ is a path through a grid graph where diagonal travel is not allowed.  Without loss of generality, assume that $\pi_a^c$ enters quadrant $B$.

We now know that some path $\pi_a^c$ through the binding graph of $\alpha$ connects $\vec{a}$ and $\vec{c}$ and also contains at least one point in quadrant $B$.  Additionally, we know that within the overall binding graph of $\alpha$, $\pi_a^c$ must somehow also be connected to the tile at $\vec{d}$.  Let $\pi_d^{ac}$ be a path in the binding graph which contains both $\vec{d}$ and some tile in $\pi_a^c$.  Note that this may just be the tile at $\vec{d}$ since $\pi_a^c$ may contain that location.  We now define path $\pi'$ as the longest of the following two paths formed by possibly combining subpaths of $\pi_a^c$ and $\pi_d^{ac}$:  1) a path which connects $\vec{d}$ and $\vec{c}$ and contains a tile in $B$, or 2) a path which connects $\vec{a}$ and $\vec{d}$ and contains a tile in quadrant $B$.  For ease of discussion, we will talk about $\pi_a^c$ as being a directed path from $\vec{a}$ to $\vec{c}$, while noting that we don't know the actual ordering of its growth.  Such a $\pi'$ must exist because the point at which $\pi_d^{ac}$ intersects $\pi_a^c$ must be either 1) before it has entered $B$, which yields case 1 for $\pi'$, 2) after it has left $B$, which yields case 2 for $\pi'$, or 3) within $B$ in which case either holds.  Note that $\pi_a^c$ could perhaps enter and leave $B$ multiple times, which would only strengthen our argument, but for simplicity we will simply make use of the first time it enters quadrant $B$ and the last time it leaves quadrant $B$ for the above argument.

\begin{figure}[htp]
\begin{center}
	\includegraphics[width=2.0in]{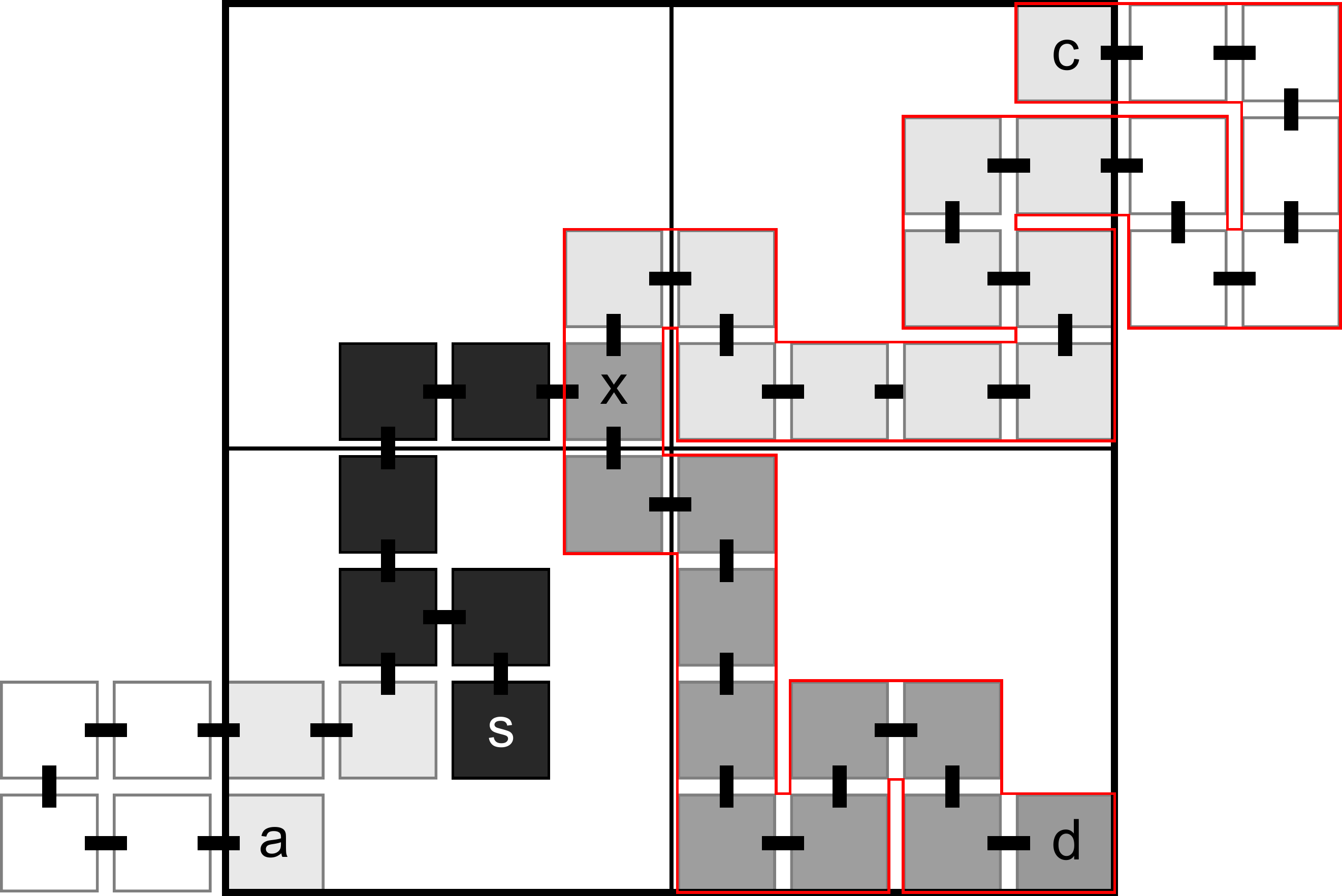}
\caption{Example paths in an $8 \times 8$ square (i.e. of even dimension).  White tiles represent tile types not in $B$ (and thus tiles not in $S$), black tiles represent path $\pi_{\sigma}'$, dark grey tiles represent path $\pi_d^{ac}$, and $\pi_a^c$ is represented by a combination of tiles of all colors.  Path $\pi'$ connects $\vec{d}$ and $\vec{c}$ by passing through $\vec{x}$ (which is in quadrant $B$) and is outlined in red.}
\label{fig:s-x-a-c-paths-even}
\end{center}

\end{figure}

Without loss of generality, we will assume case 1 for path $\pi'$, and note that we now have found the following path in the binding graph of $\alpha$.  Path $\pi'$ includes the tile at location $\vec{d}$, contains at least one tile in quadrant $B$, and also includes the tile at location $\vec{c}$.  (See Figure~\ref{fig:s-x-a-c-paths-even} for an example.) For a square of dimension $n$, it must be the case that path $\pi'$ contains a minimum of $n$ pairs of tiles bound to each other on their east/west edges.

Now we check to see if the seed $\sigma$ is contained within $\pi'$.  If not, we define $\pi_{\sigma}'$ as the minimum path in $\alpha$'s binding graph which contains $\sigma$ and some tile in $\pi'$, else $\pi_{\sigma}'$ is simply $\sigma$.  Now we define a new valid assembly sequence, $\vec{\alpha'}$ for $\calT$ as follows.

\begin{figure}[htp]
\begin{center}
	\includegraphics[width=3.0in]{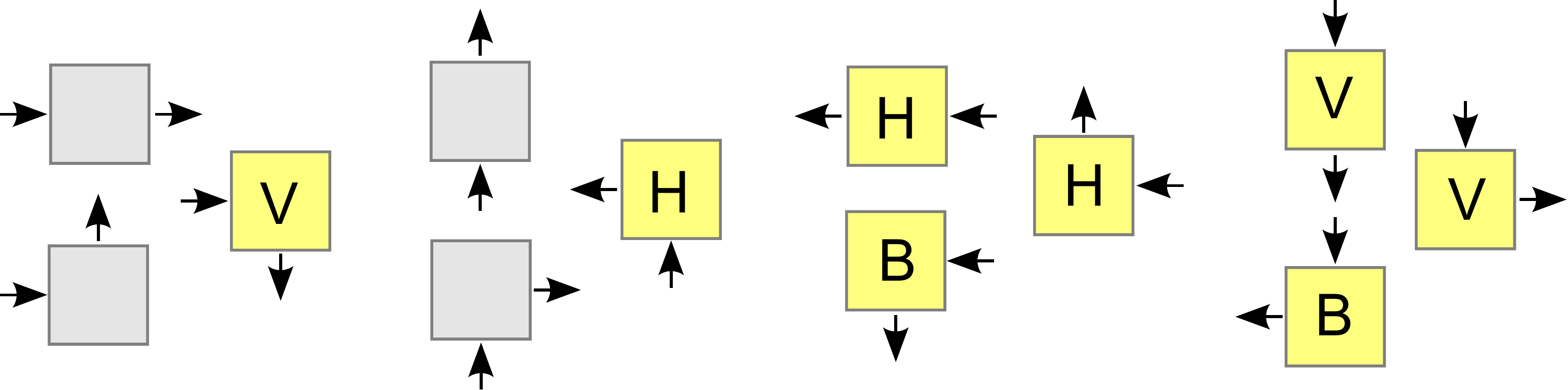}
\caption{Every possible pairing of an input and output side of a tile on a path.  Those in grey have as input sides either south or west, and as output sides either north or east.  Those in yellow do not, but are labeled with the minimal reflection necessary to orient them so that they do.}
\label{fig:input-output-flips}
\end{center}

\end{figure}

Let $\vec{x}$ denote the location of the intersection of $\pi_{\sigma}'$ and $\pi'$, and $\pi_x^c$ be the portion $\pi'$ from $\vec{x}$ to $\vec{c}$ and let $\pi_x^d$ be the portion of $\pi'$ from $\vec{x}$ to $\vec{d}$.  Without loss of generality, we will assume that the tile on $\pi_{\sigma}'$ which binds to $\sigma$ attaches to the north of $\sigma$.  (If this is not the case, then all following directions can simply be rotated by the same angle as the difference of the direction of the first binding from north.)  Starting from $\sigma$, $\vec{\alpha'}$ attaches one tile at a time from $\pi_{\sigma}'$ so that output glues are exposed on the north or east edge of each tile.  In other words, as tiles attach, flip them so that the next tile to attach does so at a tile location that is north or east of the previously attached tile.  (Figure~\ref{fig:input-output-flips} shows how this must be possible for each tile.)  This causes all inputs along $\pi_{\sigma}'$ to be on the south or west sides of tiles, and outputs on the north and east.  Furthermore, place the tile from location $\vec{x}$ so that the exposed glue used for the binding of $\pi_x^c$ is on its the north or west, and the exposed glue used for the binding of $\pi_x^d$ is on its south or east.  Figure~\ref{fig:input-output-flips2} shows how this can be done.  (Note that depending on where and from which direction $\pi_{\sigma}'$ intersects $\pi'$, it could be the case that the output sides which bind to $\pi_x^c$ and $\pi_x^d$ are instead south or east, and north or west, respectively.  However, again this doesn't change the argument as the following directions can be rotated appropriately, so we assume the output for $\pi_x^c$ is on the north or west for ease of discussion.)

\begin{figure}[htp]
\begin{center}
	\includegraphics[width=3.0in]{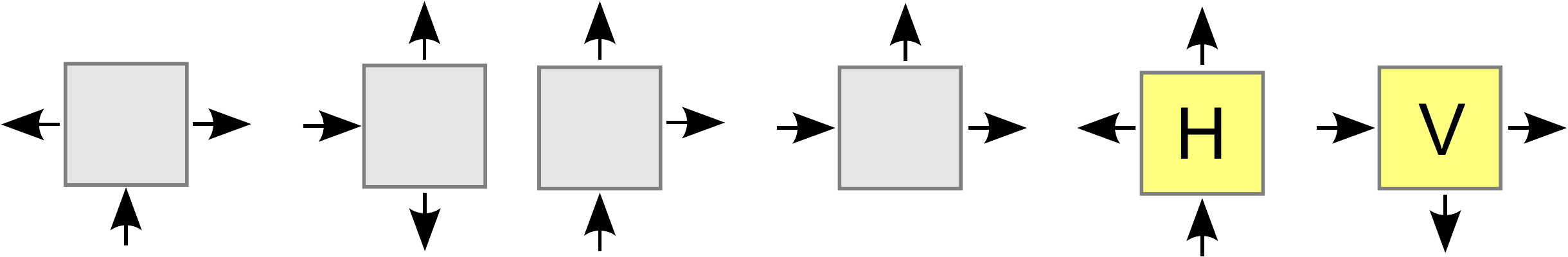}
\caption{Every possible pairing of a south or west input side and 2 output sides of the final tile of $\pi_{\sigma}'$.  Those in grey have one output side on the north or west, and the other on the south or east.  Those in yellow do not, but are labeled with the minimal reflection necessary to orient them so that they do.}
\label{fig:input-output-flips2}
\end{center}

\end{figure}

Now we assemble the path $\pi_x^c$ as follows.  Starting from the tile of location $\vec{x}$, attach one tile at a time so that the output glues are exposed on the north or west edge of each tile.  This results in a version of $\pi_x^c$ which grows strictly up and to the left.  Such growth is possible because of the set of reflections possible and the fact that no tile previously placed from $\pi_x^c$ or $\pi_{\sigma}'$ can block the placement.  The fact that blocking can't occur is due to the fact that the stretched version of $\pi_x^c$ grows strictly up and/or to the left so all previous tiles placed for $\pi_x^c$ cannot block, and after the first tile placed after the final tile of $\pi_{\sigma}'$, the newly forming path is either completely above $\pi_{\sigma}'$ or to the left of the topmost tile of that path.  Next we assemble the path $\pi_x^d$ in an analogous manner, but down and to the right.  Starting from the tile of location $\vec{x}$, attach one tile at a time so that the output glues are exposed on the south or east edge of each tile.  For the same but reflected reasons as for the growth of the modified version of $\pi_x^c$, $\pi_x^d$ can be grown in this way.  See Figure~\ref{fig:s-x-a-c-paths-stretched-even} for an example of these ``stretched'' paths.

\begin{figure}[htp]
\begin{center}
	\includegraphics[width=3.5in]{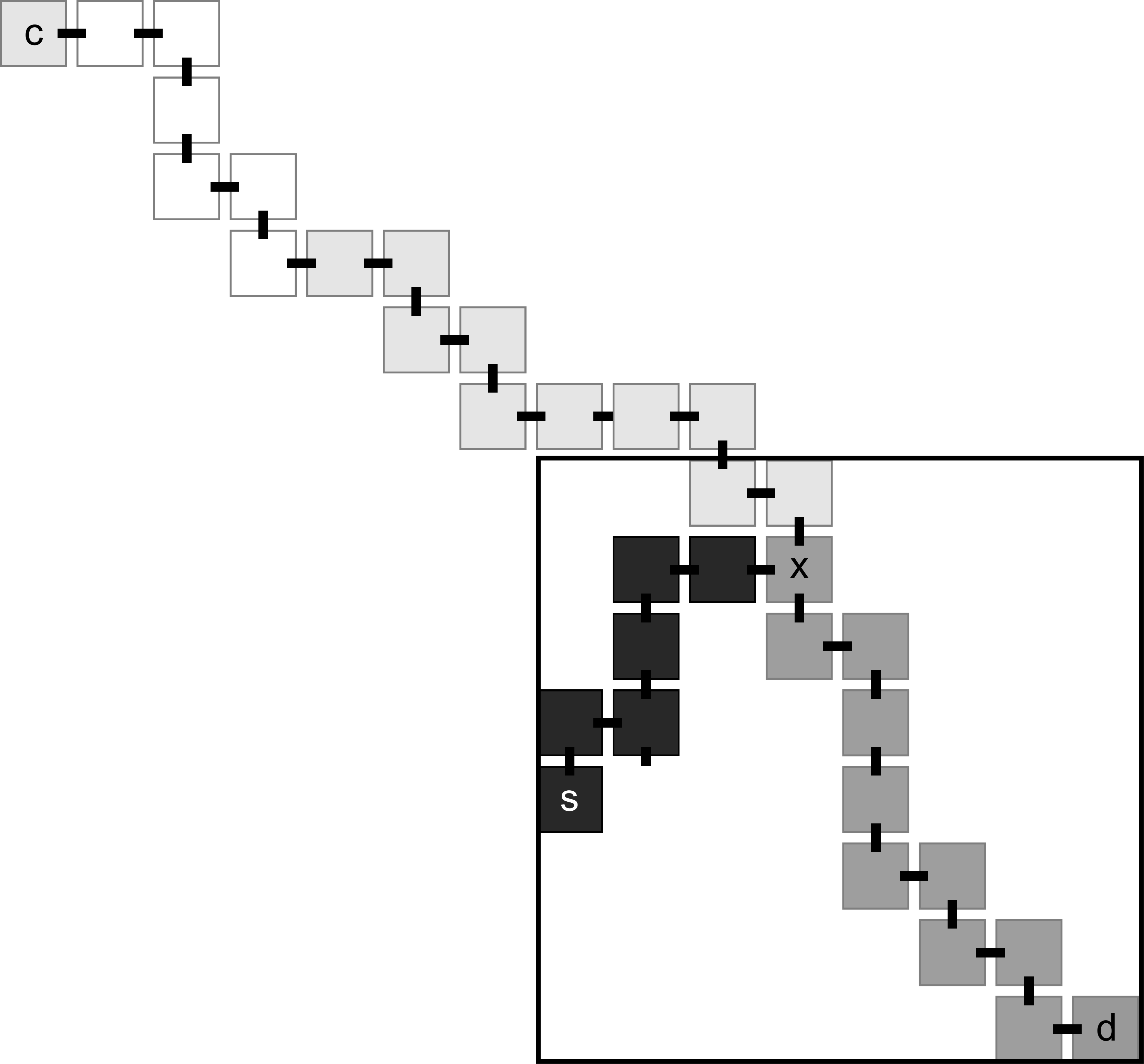}
\caption{Example paths $\pi_{\sigma}^x$, $\pi_x^c$, and $\pi_x^d$ stretched out when assembled by $\vec{\alpha}$.}
\label{fig:s-x-a-c-paths-stretched-even}
\end{center}

\end{figure}

The tile types which are at the ends of $\pi_x^c$ and $\pi_x^d$ were at locations $\vec{c}$ and $\vec{d}$ in $\alpha$, so they must be of types which are in $B \subseteq T$ (i.e. ``black'' tile types).  However, since there are $\ge n$ pairs of tiles connected via east/west glues in $\pi'$ and it is now maximally stretched, it must be $\ge n+1$ tiles wide (i.e. from leftmost to rightmost tiles).  Thus, this is a producible assembly which has ``black'' tiles at a distance which is further apart than any two points in the square $S$.  Therefore $\vec{\alpha'}$ does not weakly (or, therefore, strictly) self-assemble the $n \times n$ square $S$ and so neither does $\calT$.  This is a contradiction, and thus $\calT$ does not exist.
\end{proof}

}  %

\subsection{Tight bounds on the tile complexity of squares of odd dimension}\label{sec:odd-squares}

In this section, we prove tight bounds on the number of tiles necessary to self-assemble a square of odd dimension.

\begin{theorem}\label{thm:square-upper}
For all $n \in \Z^+$ where $n$ is odd, an $n \times n$ square strictly self-assembles in an RTAM system $\mathcal{T}=(T,\sigma,1)$ where $|T|=n$ and $|\sigma|=1$.
\end{theorem}

\later{

\section{Proof of Theorem~\ref{thm:square-upper}}\label{sec:square-upper-proof}

In this section, we give the details of the proof of Theorem~\ref{thm:square-upper}.

\begin{proof}

The general scheme for the construction can be seen in
Figure~\ref{fig:odd-square}.  Essentially, the seed, which is in the center of the square, presents some glue on both the north and south, and presents some second glue on both the east and west.  Above the seed, a series of $(n-1)/2$ hard-coded tile types form a vertical column, with the east and west edges of all tiles exposing the same glue.  Using the same tile types, a reflected version of the upper column grows downward to the bottom of the square.  To both the east and west of that column, a hard-coded series of another $(n-1)/2$ tile types form reflected rows out to the sides of the square.  This requires $1+2((n-1)/2)=n$ tile types and strictly self-assembles an $n \times n$ square.

\begin{figure}[htp!]

\centering
   \includegraphics[width=2in]{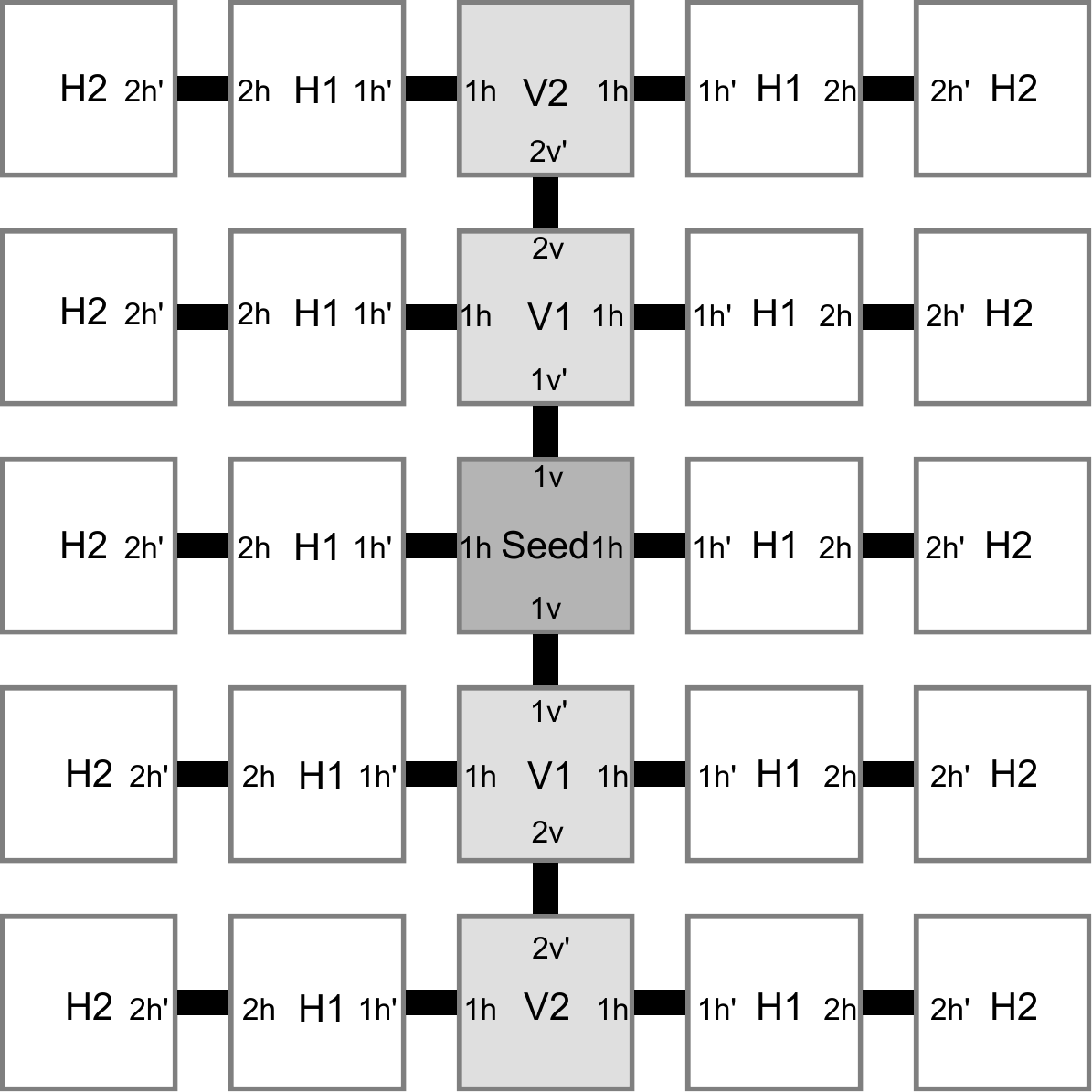}
\caption{A $5 \times 5$ square built by an RTAM TAS at $\tau=1$ using only $5$ tile types.}
\label{fig:odd-square}

\end{figure}

To prove Theorem~\ref{thm:square-upper}, we provide the following construction which demonstrates how to build an RTAM TAS which strictly self-assembles an $n \times n$ square for any odd $n$.  The basic idea is that we place the seed tile in the middle of the square, and from both its top and bottom it grows copies of a reflected column of tiles which grow to the top and bottom of the square, respectively.  This requires one tile type for the seed and $(n-1)/2$ tile types for the column (which grows in both directions).  Each of the seed and the tiles of the vertical column have the same glue on both east and west sides.  From each of those a row of $(n-1)/2$ unique tile types grows to the left or right boundary of the square.  This construction is robust to any valid rotation of the tiles and strictly self-assembles an $n \times n$ square using exactly $1 + 2(n-1)/2 = n$ tile types.

\begin{figure}[htp]
\begin{center}
   \includegraphics[width=2.0in]{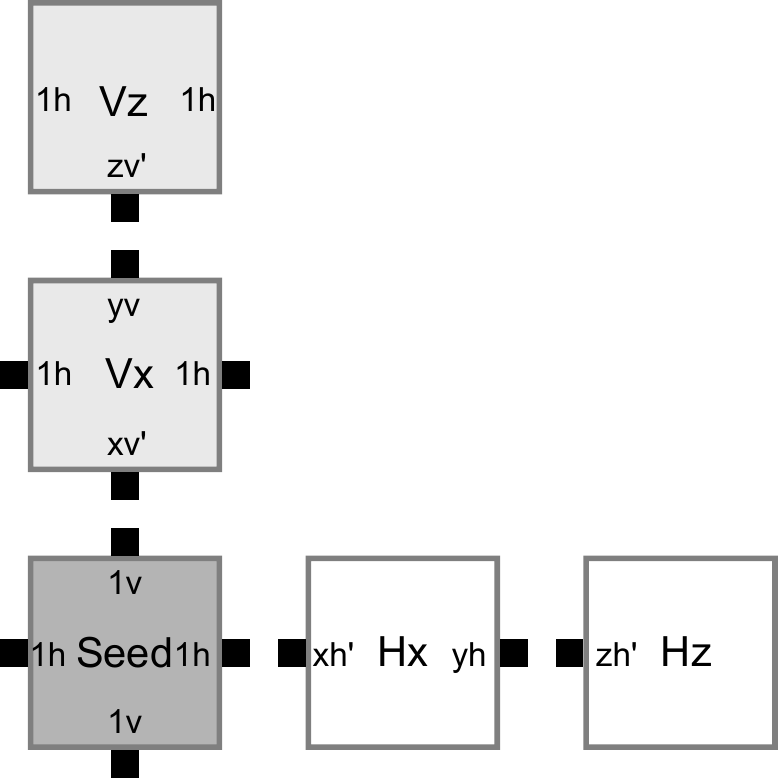}
\caption{The tile type templates for building a square whose sides are odd length in the RTAM at $\tau=1$ using $n$ unique tile types.}
\label{fig:odd-square-tiles}
\end{center}

\end{figure}

In Figure~\ref{fig:odd-square-tiles} we show the templates used to generate the necessary tile types to self-assemble an $n \times n$ square for odd $n$, and in Figure~\ref{fig:odd-square} we show an example of the terminal assembly for $n=3$.  For all odd $n$, the seed is exactly as in Figure~\ref{fig:odd-square-tiles}.  Then for each $m$ where $0<m<(n-1)/2$, we create a unique tile type from the tile type template labeled ``Vx''.  We do this by replacing each ``x'' (i.e. those of the tile label and the south glue) with the value ``$m$'', and the ``y'' of the north glue with the value ``$m+1$''.  (Note the ``prime'' markings of some glues which denote they are complementary to the ``unprimed'' versions, and for all tile types created from these templates we keep those markings unchanged, which is important in restricting the potential interactions.)  In this way, we create the tile types for the central column between the seed and the top or bottom most locations.  To make the top/bottom tile type, we start with the ``Vz'' template and replace the ``z'' of the label and south glue with the value $(n-1)/2$.  To make the tile types for the horizontal rows which grow outward from the central column, for each $m$ where $0<m<(n-1)/2$ we create a unique tile type from the tile type template labeled ``Hx''.  We do this by replacing each ``x'' (i.e. those of the tile label and the west glue) with the value ``$m$'', and the ``y'' of the east glue with the value ``$m+1$''.  In this way we create the tile types for the horizontal row between the central column and left or right most locations.  To make the left/right tile type, we start with the ``Hz'' template and replace the ``z'' of the label and west glue with the value $(n-1)/2$.  This completes the creation of exactly $n$ tile types: 1 for the seed, $2(((n-1)/2) - 1)$ for the interiors of the column and rows, and 1 each for the top/bottom tile type and left/right tile type.  The RTAM TAS $\calT$ consisting of this tile set, a seed consisting of a copy of the seed tile in any valid orientation at the origin, and $\tau=1$ strictly self-assembles an $n \times n$ square and is directed.  This is because the only tiles which can attach to the north and south of the seed are properly reflected versions of the ``V1'' tile type, and then to the north/south of those two tiles properly reflected copies of ``V2'', etc. until copies of the ``V$(n-1/2)$'' tile type cap off the north and south growth of the column.  Regardless of the reflections of the ``Vx'' and ``Vz'' tiles across the vertical axis, the only subassemblies that can grow to the left and right are the arms consisting of the ``Hx'' and ``Hz'' tile types.  All producible terminal assemblies of $\calT$ are equivalent and in the shape of an $n \times n$ square.
\end{proof}

} %

We prove Theorem~\ref{thm:square-upper} by giving a scheme for obtaining the tileset for any given $n$ that exploits the fact that for $n$ odd, an $n\times n$ square in $\Z^2$ is symmetric across a row and column points. Although Theorem~\ref{thm:square-upper} pertains to squares, a simple modification of the proof shows the following corollary.

\begin{corollary}
For all $n,m \in \Z^+$ where $n$ is odd, an $n \times m$ square strictly self-assembles in an RTAM system $\mathcal{T}=(T,\sigma,1)$ where $|T|=\frac{n+m}{2}$ and $|\sigma|=1$.
\end{corollary}

We also prove that the upper bound of Theorem~\ref{thm:square-upper} is tight, i.e. an $n \times n$ square, where $n$ is odd, cannot be self-assembled using less than $n$ tile types. The proof of this theorem can be found in Section~\ref{sec:square-lower-proof}.

\begin{theorem}\label{thm:square-lower}
For all $n \in \Z^+$ where $n$ is odd, there exists no RTAM system $\mathcal{T}=(T,\sigma,1)$ where $|T|<n$ and $|\sigma|=1$ such that $\mathcal{T}$ weakly (or strictly) self-assembles an $n \times n$ square.
\end{theorem}

\later{

\section{Proof of Theorem~\ref{thm:square-lower}}\label{sec:square-lower-proof}

In this section, we give the full proof of Theorem~\ref{thm:square-lower}.

\begin{proof}
We prove Theorem~\ref{thm:square-lower} by contradiction.  Therefore, assume that for some $n \in \Z^+$ such that $(n \mod 2) = 1$, there exists an RTAM system $\mathcal{T} = (T,\sigma,1)$ such that $|\sigma|=1$ and $|T|<n$, and $\mathcal{T}$ weakly self-assembles an $n \times n$ square $S$.  It must therefore be the case that for some subset of tile types $B \subseteq T$, for all $\alpha \in \termasm{\mathcal{T}}$ there exist $r \in R$ and $\vec{v} \in \Z^2$ such that for $\alpha_r = F(\alpha,r,\vec{v})$, $\alpha_r^{-1}(B) = S$.

We choose an arbitrary $\alpha \in \termasm{\mathcal{T}}$ and, using the notation of Figure~\ref{fig:even-square-corners}, note that in the binding graph of $\alpha$, there is a path $\pi_a^c$ which connects the tiles in locations $\vec{a}$ and $\vec{c}$ (i.e. opposite corners of the square).  Furthermore, the tile types of the tiles at locations $\vec{a}$ and $\vec{c}$ must be in $B \subseteq T$, and tiles of types in $B$ can be no further apart from each other (in the plane) than a Manhattan distance of $2n-1$.  Since the seed $\sigma$ may not be on $\pi_a^c$, define the location $\vec{x}$ to be either 1) the location of $\sigma$ if it is on $\pi_a^c$, or 2) the location of the intersection of $\pi_a^c$ and the shortest path in the binding graph which contains $\sigma$ and some tile in $\pi_a^c$.  Note that such a location $\vec{x}$ must exist since $\alpha$ must be connected.  Define the path $\pi_{\sigma}^x$ as the path from $\sigma$ to $\vec{x}$ and note that if $\sigma$ is on $\pi_a^c$ then $\pi_{\sigma}^x$ consists of a single tile.  Also, we can now define $\pi_x^a$ and $\pi_x^c$ as the paths from $\vec{x}$ to $\vec{a}$ and $\vec{c}$, respectively.  (See Figure~\ref{fig:s-x-a-c-paths} for a possible example of such paths in $\alpha$.)

\begin{figure}[htp]
\begin{center}
	\includegraphics[width=2.0in]{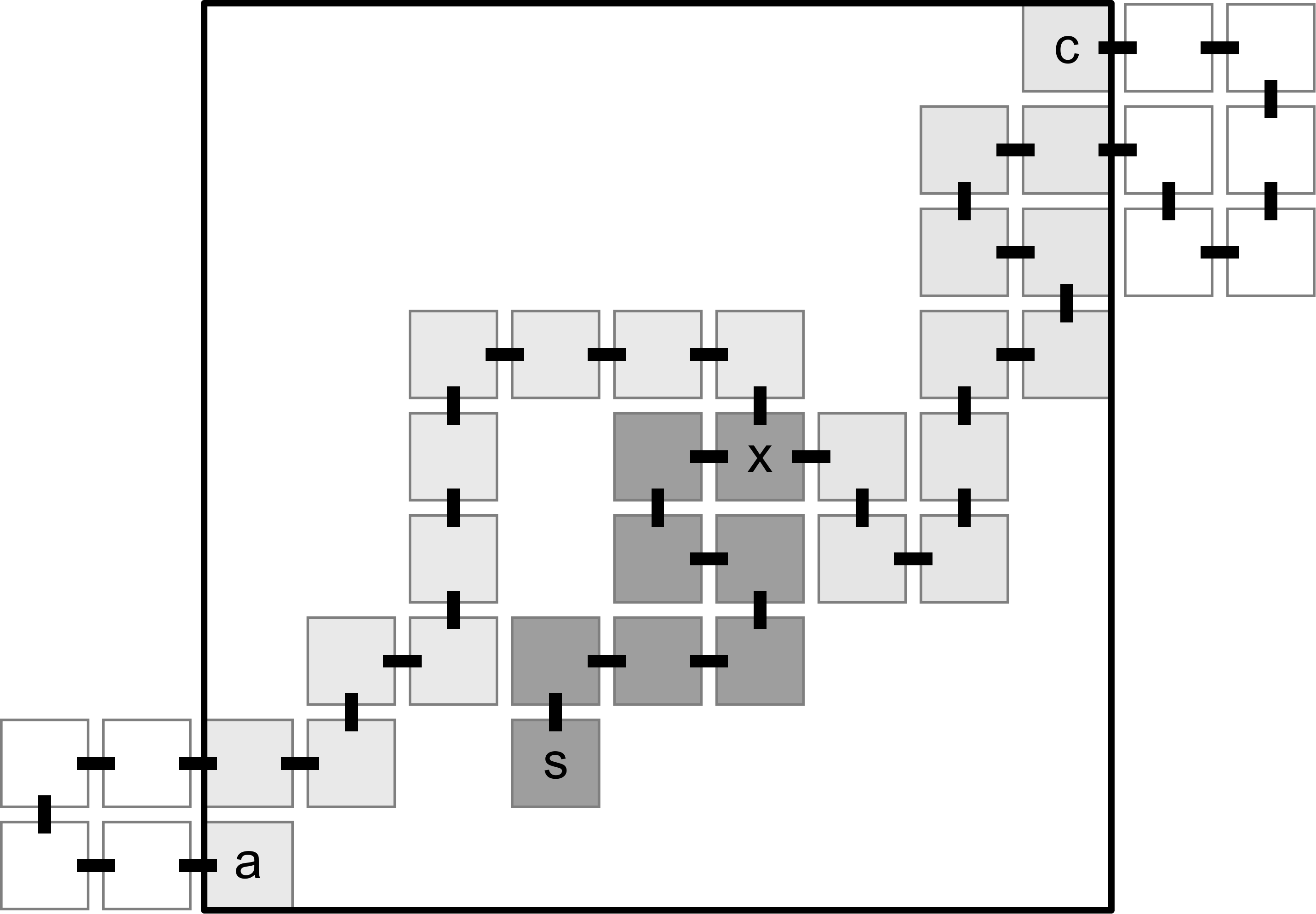}
\caption{Example paths $\pi_{\sigma}^x$ (darkest grey) and $\pi_a^c$ (a combination of light grey, dark grey, and white tiles) in $\alpha$ which weakly self-assembles a $9 \times 9$ square (i.e. of odd dimension).  White tiles represent tile types not in $B \subseteq T$, while grey and dark grey tiles represent ``black'' tile types in $B$.}
\label{fig:s-x-a-c-paths}
\end{center}

\end{figure}

We will now define a valid assembly sequence $\vec{\alpha'}$ in $\calT$ which builds an assembly $\alpha'$ which consists of modified versions of the paths $\pi_{\sigma}^x$ and $\pi_a^c$, and which places tiles of types in $B \subseteq T$ (i.e. ``black'' tiles) at points further apart than $2n-1$, proving that $\calT$ does not weakly self-assemble $S$.  Let $k=|\pi_{\sigma}^x|$ and let $\vec{\alpha'}$ start from $\sigma$ and first build a modified version of $\pi_{\sigma}^x$ in the following way.  If $k=1$, stop.  Else, without loss of generality assume that $\alpha(\pi_{\sigma}^x(1))$ attaches to the north of $\sigma$ in $\alpha$.  (If it binds on the east, south, or west, for the following argument used to construct $\alpha'$, rotate all directions used clockwise by $90$, $180$, or $270$ degrees respectively.)  For each $i$ where $0<i<k$, place a tile of the same type as at $\alpha(\pi_{\sigma}^x(i))$ but reflected so that its input side is either south or west, and its output side is either north or east.  Note that for any pair of input and output sides, it is possible to reflect a tile to meet this criteria (see Figure~\ref{fig:input-output-flips}).  Place the tile of that type and orientation so that it binds with the $(i-1)$th tile of the newly created path.  Note that this also must be possible because, following from $\sigma$, every tile is oriented so that the side used as its output side in $\pi_{\sigma}^x$ is north or east and has a glue which matches that of the newly placed tile on its (now) south or west side, respectively, and since the newly forming path grows only up and/or to the right, no previously placed tile can block it from binding.  The result of this sequence of placing the tiles from $\pi_{\sigma}^x$ is a ``stretched'' version of the path $\pi_{\sigma}^x$ which grows only up and/or right from the seed $\sigma$, grown by the valid assembly sequence $\vec{\alpha'}$.

We now extend $\vec{\alpha'}$ to grow a stretched version of $\pi_a^c$ by growing stretched versions of $\pi_x^a$ and $\pi_x^c$.  We will build one of $\pi_x^a$ or $\pi_x^c$ by stretching it up and to the left, and the other by stretching it down and to the right.  (See Figure~\ref{fig:s-x-a-c-paths-stretched} for an example.)  Since the final tile of $\pi_{\sigma}^x$ is included in $\pi_a^c$, it must be the case that it has two output sides (to bind to each of $\pi_x^a$ and $\pi_x^c$).  Also, recall that its input side must have been either south or west. We modify $\vec{\alpha'}$ so that the final tile of $\pi_{\sigma}^x$ is oriented with one output side on either the north or west, and the other on either the south or east.  See Figure~\ref{fig:input-output-flips2} for all possible orientations of this tile. Without loss of generality, assume that in this orientation, the side which serves as input to $\pi_x^a$ is on either the north or west, and that serving as input to $\pi_x^c$ is on either the south or east.  (If it is the opposite, the directions used to grow the stretched $\pi_x^a$ and $\pi_x^c$ can simply be swapped and the argument will still hold.)  We will now build a stretched version of $\pi_x^a$ up and to the left, and $\pi_x^c$ down and to the right.  To do this, we extend $\vec{\alpha'}$ to attach the tiles of $\pi_x^a$, moving along that path as it appears in $\alpha$ and placing the same tiles in the same order, but always orienting them so that each has as its input side either its south or east, and as its output side either its north or west, allowing each successive tile to correctly bind to its predecessor on the path.  As before, this is possible because of the set of reflections possible and the fact that no tile previously placed from $\pi_x^a$ or $\pi_{\sigma}^x$ can block the placement.  The fact that blocking can't occur is due to the fact that the stretched version of $\pi_x^a$ grows strictly up and/or to the left so all previous tiles placed for $\pi_x^a$ cannot block, and after the first tile placed after the final tile of $\pi_{\sigma}^x$, the newly forming path is either completely above $\pi_{\sigma}^x$ or to the left of the topmost tile of that path.  The stretched version of $\pi_x^c$ is grown analogously, with all input sides being on the north or west and output sides to the south or east.

Note that in the special case where $\vec{x} = \vec{a}$ or $\vec{x} = \vec{c}$ in $\alpha$, meaning that either the seed is in one of those corners, or the path from the seed to the path $\pi_a^c$ terminates at one of those corners, we can simplify our construction and simply grow a single direction from $\vec{x}$.

\begin{figure}[htp]
\begin{center}
	\includegraphics[width=3.5in]{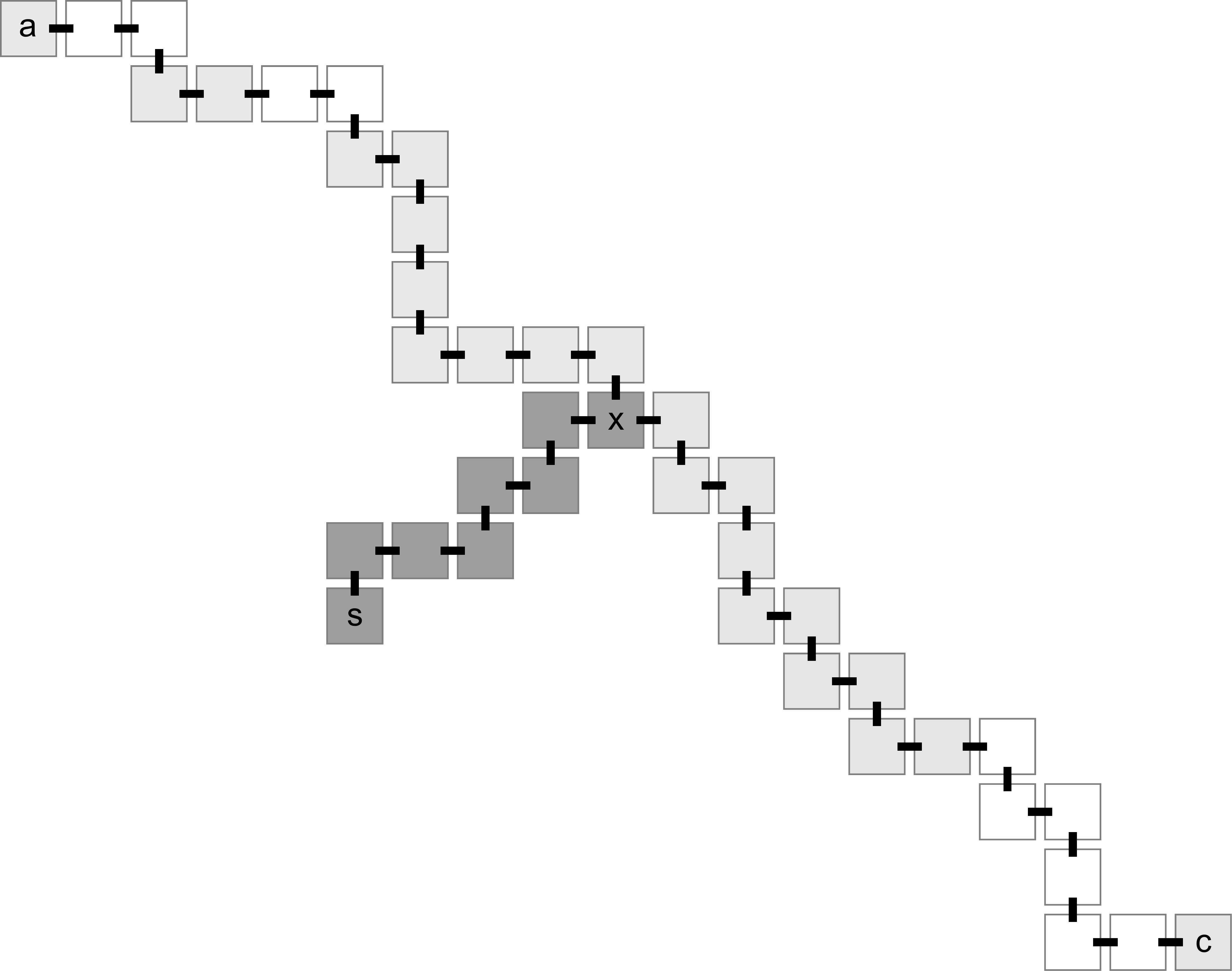}
\caption{Example paths $\pi_{\sigma}^x$ and $\pi_x^c$ stretched out in $\hat{\pi}_{\sigma}^c$.}
\label{fig:s-x-a-c-paths-stretched}
\end{center}

\end{figure}

At this point, $\vec{\alpha'}$ grows assembly $\alpha'$ which is a path that looks like that in Figure~\ref{fig:s-x-a-c-paths-stretched} and which starts from the seed and grows a stretched version of $\pi_{\sigma}^x$, then grows a stretched version of $\pi_a^c$ (as connected stretched versions of $\pi_x^a$ and $\pi_x^c$) from the end of that path.  Since in $\alpha$, $\pi_a^c$ connects the bottom-left corner of the square $S$ with the top-right corner, it must be the case that $|\pi_a^c| \ge (2n-1)$.  Since $|T| \leq (n-1)$, and $2(n-1) = 2n-2$, by the pigeonhole principle we know that along $\pi_a^c$ (and therefore also its stretched version), it must be the case that at least $3$ tiles of some same type appear.

\begin{figure}[htp]
\begin{center}
	\includegraphics[width=2.0in]{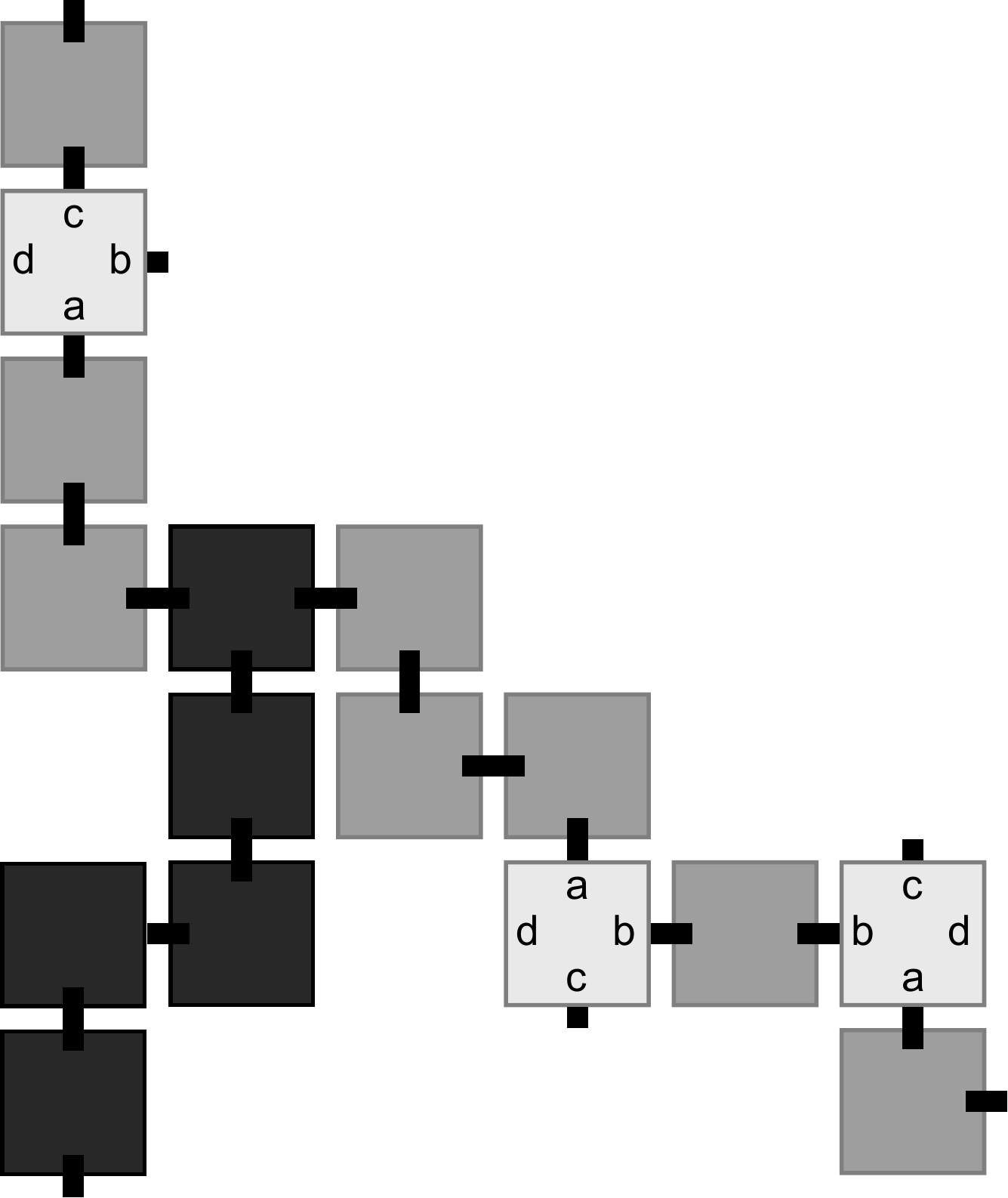}
\caption{One possible configuration of a portion of the stretched paths $\pi_{\sigma}^x$ (black) and $\pi_x^a$ and $\pi_x^c$ (grey).  Three reflected copies of the same tile type appear on $\pi_x^a$ and $\pi_x^c$ (lightest grey).  A path using different input sides can be formed between the top and bottom copies of the repeated tile type, entering from the $a$ and $b$ sides, respectively.}
\label{fig:paths-3-tile-repeats}
\end{center}

\end{figure}

With $3$ copies of the same tile type appearing along $\pi_a^c$, it must be the case that a path can be formed which connects two of the copies in such a way that it enters each of the two tiles along different sides of that tile type when it's in its default configuration.  (See Figure~\ref{fig:paths-3-tile-repeats} for an example where a path can be formed between the top copy and the bottom copy, with the top being entered via the $a$ side and the bottom via the $b$ side, while the path between the top and middle copies enters each from the $a$ side, and the path between the middle and bottom copies enters each from the $b$ side.)  Let $\vec{f}$, $\vec{g}$, and $\vec{h}$ be the locations of the repeated tile type in $\alpha'$, and let $\vec{f}$ and $\vec{g}$ be the locations which can be connected by a path entering the tiles at those locations from different sides, and let $\vec{f}$ be the furthest of $\vec{f}$ and $\vec{g}$ from the final tile of $\pi_{\sigma}^x$ (or the southernmost if they tie).  Without loss of generality, assume that $\vec{f}$ is south of $\vec{g}$ (otherwise, rotate the following directions).

We now modify $\vec{\alpha'}$ so that it builds the stretched copies of $\pi_{\sigma}^x$, $\pi_x^a$, and $\pi_x^c$ until it is about to place the tile at $\vec{f}$.  At this point, it rotates the tile for that position so that the side which serves as input for the tile at $\vec{g}$ is facing either south or east, and places that tile.  Then, since the tile at $\vec{g}$ was above and/or to the left of that at $\vec{f}$, the side it used as output is now oriented so that it can be used as output from the copy at $\vec{g}$.  We add to $\vec{\alpha'}$ the tile placements which now make a copy of the path $\pi_g^f$, again orienting the final tile to allow yet another copy.  We make $\vec{\alpha'}$ add $2n$ copies of $\pi_g^f$.  Then, we let the final tile on the final copy of $\pi_g^f$ be oriented so that the final portion of the original stretched copy of $\pi_x^c$ can grow.  At this point, $\vec{\alpha'}$, which is still a valid assembly sequence in $\calT$, forms the stretched copy of $\pi_{\sigma}^x$, as well as the stretched copy of $\pi_a^c$, with the addition that some interior portion of $\pi_a^c$ has been ``pumped'' so that the overall distance between the tiles at the end of the path are at a Manhattan distance at least $2n$ from each other.  This is because the stretched path $\pi_a^c$ is at least that long, containing $2n$ copies of $\pi_g^f$, and it grows so that each tile placed makes the endpoints strictly further from each other.  Further, the tiles placed at the endpoints of the stretched version of $\pi_a^c$ are of the same types as those placed at $\vec{a}$ and $\vec{c}$ in $\alpha$.  Since those were at the corners of the square $S$, they were tile types from $B \subseteq T$ and must therefore be inside the square $S$. However, in $S$, no two points are at a Manhattan distance greater than $2n-1$ from each other.  Thus, $\vec{\alpha'}$ does not weakly (or therefore strictly) self-assemble $S$ and neither does $\calT$.  This is a contradiction, and thus Theorem~\ref{thm:square-lower} is proven.
\end{proof}

} %

\vspace{-20pt}
\subsection{Assembling Finite Shapes in the \rtam/}\label{sec:general-shapes}
\vspace{-8pt}

In this section we first give a corollary of Theorem~\ref{thm:square-upper} showing that sufficiently symmetric shapes weakly self-assemble in the \rtam/. Then we prove $3$ theorems about assembling finite shapes in $\Z^2$ in the \rtam/. These theorems show that the assembly of finite shapes in singly seeded \rtam/ systems is quite a bit different than the assembly of finite shapes in the aTAM by singly seeded systems.

Given a shape $S$ in $\Z^2$, let $\chi_S$ denote the characteristic function of the set $S$. That is, $\chi_S\left(x,y\right) = 1$ if $x\in S$ and $\chi_S\left(x,y\right) = 0$ otherwise. Then, we say that a shape $S$ is \emph{odd-symmetric} with respect to a horizontal line $y=l$ (respectively, vertical line $x=l$)  for $l\in \Z^2$ iff for all $(a,b)\in \Z^2$, $\chi_S(a,l-b) = \chi_S(a,l+b)$ (respectively, $\chi_S(l-a,b) = \chi_S(l+a,b)$). If there exists a line such that a shape, $S$, is odd-symmetric with respect to this line, we say that the shape is odd-symmetric. Given a shape $S$, we call the smallest rectangle of points in $\Z^2$ containing $S$ the \emph{bounding-box} for $S$.

Let $R$ denote the bounding box of an odd-symmetric shape $S$, and let $n$ and $m$ in $\N$ be the dimensions of $R$. A simple modification to the proof of Theorem~\ref{thm:square-upper} where a tile is labeled with a label $B$  if and only if it corresponds to points in $S$, shows that odd-symmetric shapes can weakly assemble in the \rtam/.

\begin{corollary}\label{cor:weak-assembly-symm}
Given a shape $S$ in $\Z^2$, if $S$ is odd-symmetric, then there exists an RTAS $\mathcal{T} = (T, \sigma, \tau)$ such that $|\sigma|=1$, $\tau\geq 1$, and $\mathcal{T}$ weakly assemblies $S$.
\end{corollary}

Additionally, if one is willing to build 2 mirrored copies of the shape in each assembly, then any finite shape can be weakly self-assembled in the RTAM at $\tau=1$, along with its mirrored copy (at a cost of tile complexity approximately equal to the number of points in the shape) by simply building a central column (or row) from which identical copies of hardcoded rows (or columns) grow, so that each side grows a reflected copy of the shape in hardcoded slices.

We say that a TAS (in either the aTAM or the RTAM) $\mathcal{T}$ is called \emph{mismatch-free} if for every producible assembly $\alpha\in \prodasm{T}$ with two neighboring tiles with abutting edges $e_1$ and $e_2$, either $e_1$ and $e_2$ do not have glues or $e_1$ and $e_2$ have glues with matching labels and strengths.
Then, for singly seeded aTAM systems, any finite connected shape can be strictly assembled by a mismatch-free system. Theorems~\ref{thm:no-strict},~\ref{thm:no-directed}, and~\ref{thm:no-mismatch-free} show that assembling shapes in the \rtam/ is more complex. The proofs of these theorems can be found in Section~\ref{sec:proof-shapes}.

\begin{theorem}\label{thm:no-strict}
There exists a finite connected shape $S$ in $\Z^2$ that weakly self-assembles in a singly seeded \rtam/ system such that there exists no singly seeded \rtam/ system that strictly self-assembles $S$.
\end{theorem}

\later{
\section{Proofs of Theorems~\ref{thm:no-strict},~\ref{thm:no-directed}, and~\ref{thm:no-mismatch-free}}\label{sec:proof-shapes}
\subsection{Proof of Theorem~\ref{thm:no-strict}}

Consider the shape $S$ depicted in Figure~\ref{fig:no-strict1}. By Corollary~\ref{cor:weak-assembly-symm}, there is a singly seeded \rtam/ system that weakly self-assembles $S$. Therefore, to complete the proof, we show that this shape cannot be strictly self-assembled by a singly seeded system in the \rtam/ by contradiction. Suppose that $\mathcal{T} = (T,\sigma, \tau)$ is a singly seeded system that strictly assembles $S$, and let $\alpha$ in $\termasm{T}$ be a terminal assembly such that $\dom(\alpha) = S$. Without loss of generality, assume that the seed for $\mathcal{T}$ is at a location shown in red in Figure~\ref{fig:no-strict1}. Then, by applying a single tile reflection, we can modify the assembly sequence of $\alpha$ to obtain a producible assembly $\beta$ depicted in Figure~\ref{fig:no-strict2}, which contradicts the fact that $S$ strictly assembles in $\mathcal{T}$.

\begin{figure}[htp]
\centering
  \subfloat[][]{%
        \label{fig:no-strict1}%
        \makebox[2in][c]{ \includegraphics[width=1.5in]{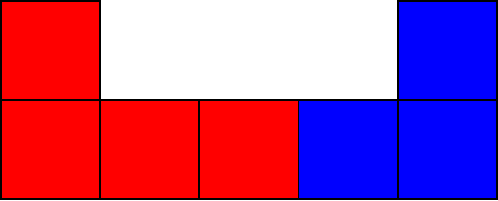}}
        }%
        \quad
  \subfloat[][]{%
        \label{fig:no-strict2}%
        \makebox[2in][c]{ \includegraphics[width=1.5in]{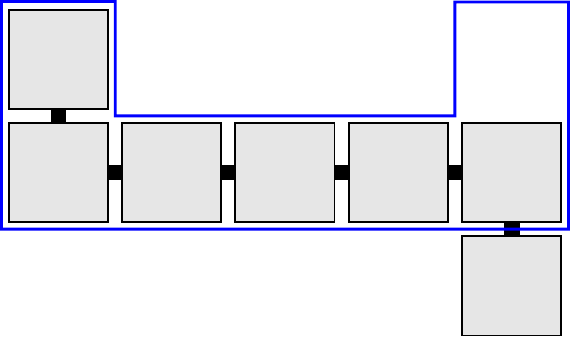}}
        }%
  \caption{(a) A shape $S$ that cannot strictly self-assemble in the \rtam/. (b) A producible assembly that can assemble in any system which can also produce an assembly whose domain is $S$ (upto translation and reflection).}
  \label{fig:no-strict}
\end{figure}

}%

\begin{theorem}\label{thm:no-directed}
There exists a finite shape $S$ in $\Z^2$ that can be strictly self-assembled by some singly seeded \rtam/ system such that every singly seeded \rtam/ system at temperature $1$ which strictly self-assembles $S$ is not directed.
\end{theorem}

\later{
\subsection{Proof of Theorem~\ref{thm:no-directed}}

\begin{figure}[htp]
\centering
  \subfloat[][]{%
        \label{fig:no-directed1}%
        \makebox[2in][c]{ \includegraphics[width=1in]{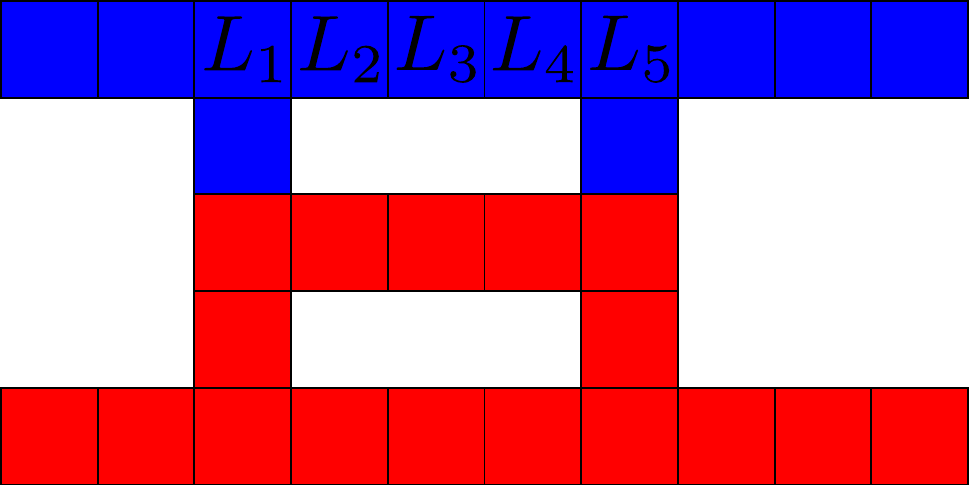}}
        }%
        \quad
  \subfloat[][]{%
        \label{fig:no-directed2}%
        \makebox[2in][c]{ \includegraphics[width=1.5in]{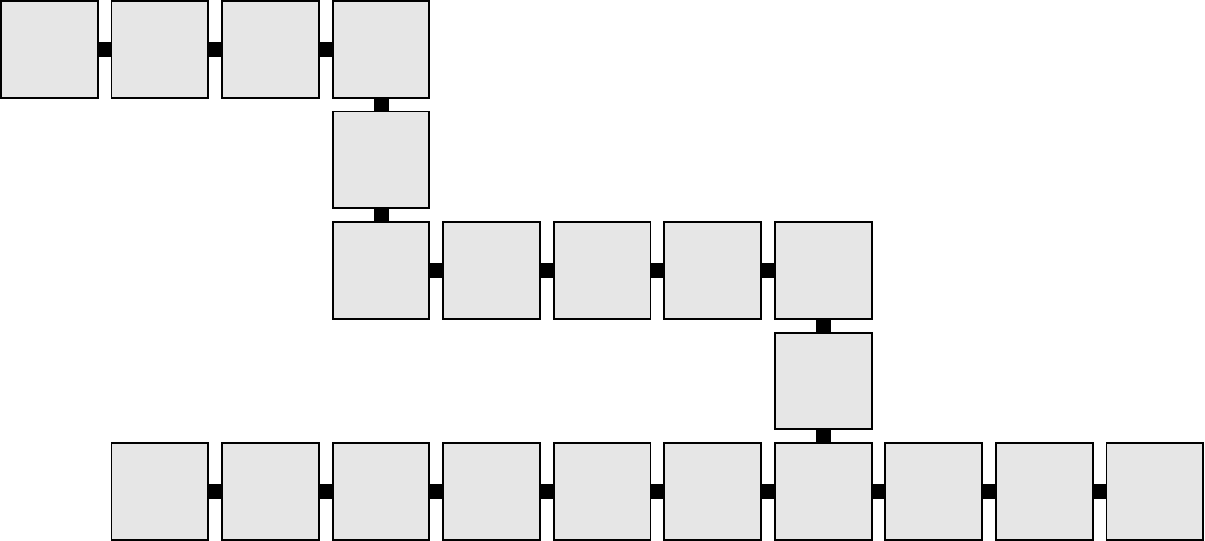}}
        }%
  \caption{(a) A shape $S$ that cannot be self-assembled in the \rtam/ by a directed system. }
  \label{fig:no-directed}
\end{figure}

Consider the shape $S$ depicted in Figure~\ref{fig:no-directed1}. First, to see that $S$ can be strictly self-assembled in the \rtam/ by a singly seeded system, consider the tile set shown in Figure~\ref{fig:no-directed-tileset}.

\begin{figure}[htp]
\centering
        \includegraphics[width=1.5in]{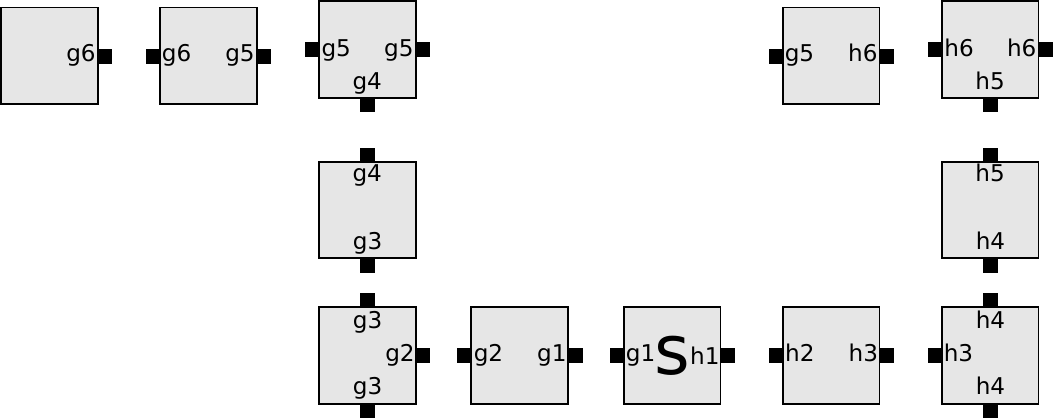}
  \caption{A tile set $T$ such that a system $\mathcal{T} = (T, \sigma, \tau)$, with $\sigma$ a single tile assembly consisting of the tile labeled $S$, strictly self-assembles the shape $S$ shown in Figure~\ref{fig:no-directed1}. Note that all of the glues depicted in this figure have strength $\tau$, and that assembly of the shape $S$ requires the reflection of these tiles.}
  \label{fig:no-directed-tileset}
\end{figure}

Then, to complete the proof, we show by contradiction that $S$ cannot be strictly self-assembled by a directed singly seeded system in the \rtam/. Suppose that $\mathcal{T} = (T,\sigma, \tau)$ is a directed singly seeded system that strictly assembles $S$, and let $\alpha$ in $\termasm{T}$ be a terminal assembly such that $\dom(\alpha) = S$. Without loss of generality, assume that the location of the seed in $\alpha$ is at a location shown in red in Figure~\ref{fig:no-directed1}.

Let $t_i$ denote the tile of $\alpha$ located at positions $L_i$.
Note that the south glue of either $t_1$ or $t_5$ is bound to $\alpha$ with strength $\tau$. Without loss of generality, assume that the south glue of $t_1$ is bound to $\alpha$ with strength $\tau$.
If $t_1$ is bound to $t_2$ with strength $\tau$, $t_2$ is bound to $t_3$ with strength $\tau$, and $t_3$ is bound to $t_4$ with strength $\tau$, then by modifying the assembly sequence for $\alpha$ and applying the appropriate reflections, we can can give an assembly that is producible in $\mathcal{T}$, and is a total of $11$ tiles wide. An example of such an assembly is shown in Figure~\ref{fig:no-directed2}. This contradicts the assumption that $\mathcal{T}$ strictly self assembles $S$.
On the other hand, if the strength of the bond between $t_i$ and $t_{i+1}$ is less than $\tau$ for $1\leq i \leq 4$, then the south glues of $t_1$ and $t_5$ must each be bound with strength $\tau$.
Therefore, the assembly $\alpha'$ depicted as gray tiles in Figure~\ref{fig:no-directed3} must be producible in $\mathcal{T}$.
\begin{figure}[htp]
\centering
        \includegraphics[width=1.5in]{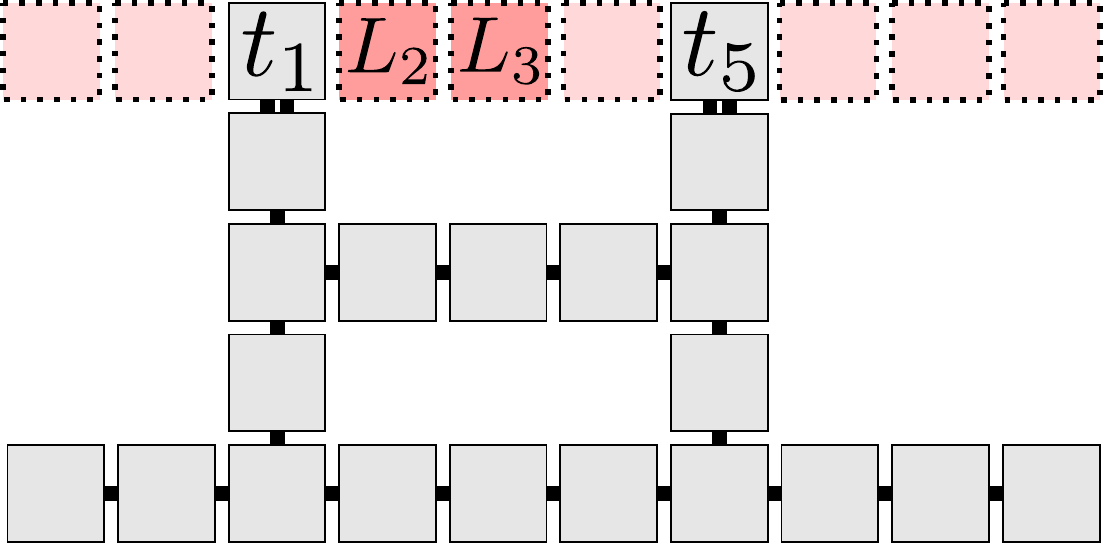}
  \caption{A shape $S$ that cannot be self-assembled in the \rtam/ by a directed system. Strength-$\tau$ glues are shown with two bars, the rest of the glues can be either strength-1 or strength-2 as long as the assembly is $\tau$-stable.}
  \label{fig:no-directed3}
\end{figure}
In the following argument, without loss of generality, we assume that for every glue, $g$, on a tile in $T$, there is a matching glue $g'$ on some (possibly the same) tile in $T$.
Then, since two tiles must be able to attach so that they occupy the two tile locations to the west of $t_1$, after applying a reflection of $t_1$, these same tiles, call them $a$ and $b$, must be able to attach with strength $\tau$ to the east to $t_1$ with the easternmost tile, $b$, located at $L_3$.  Similarly, (prior to $a$ and $b$ binding) three tiles, call them $a'$, $b'$, and $c'$, must be able to attach with strength $\tau$ to the west of $t_5$, with $a'$ located at $L_4$, $b'$ located at $L_3$, and $c'$ located at $L_2$. Notice that $b'$ must have two strength-$\tau$ glues.
Hence, nondeterministically, either $b$ or $b'$ will be at location $L_3$. Finally, note that $b$ and $b'$ cannot have the same tile type. This follows by contradiction; suppose that $b$ and $b'$ are of the same type. Then, after binding, $b$ must expose a strength-$\tau$ glue. However, when the westernmost tile that attaches to the west of $t_1$ is $b$, this strength-$\tau$ glue would allow for a tile to bind, giving a producible assembly with a domain that is not contained in $S$. This contradicts the assumption that $\mathcal{T}$ strictly self-assembles $S$. Finally, since in any terminal assembly of $\mathcal{T}$, either $b$ or $b'$ will be at location $L_3$, and $b$ and $b'$ do not have the same tile type, we conclude that
$\mathcal{T}$ is not directed.

}%

\begin{theorem}\label{thm:no-mismatch-free}
There exists a finite shape $S$ in $\Z^2$ such that every singly seeded \rtam/ system that strictly self-assembles $S$ is not mismatch-free.
\end{theorem}

\later{
\subsection{Proof of Theorem~\ref{thm:no-mismatch-free}}
Let $S$ denote the shape shown in Figure~\ref{fig:no-directed1}. Then, the proof follows from the proof of Theorem~\ref{fig:no-directed}. Note that in proof of Theorem~\ref{fig:no-directed}, if $b$ is placed at $L_3$, then the west edge of $a'$ and the east edge of $b$ must have different glues. Similarly, if $b'$ is placed at $L_3$, then the east edge of $a$ and the west edge of $b'$ must have different glues. Therefore, any singly seeded \rtam/ system that strictly self-assembles $S$ is not mismatch-free.
}%

\vspace{-20pt}
\subsection{Mismatch-free Assembly of Finite Shapes in the \rtam/}\label{sec:general-shapes}

Given a shape $S$, i.e. a finite connected subset of $\mathbb{Z}^2$, we say that a \emph{graph of $S$} is a graph $G_S = (V,E)$ with a vertex at the center of each point in $S$ and an edge between every pair of vertices at adjacent points of $S$. A \emph{tree of $S$}, $T_S$, is a graph of $S$ which is a tree.  (See Figure~\ref{fig:epsilon-symmetric-example} for examples of $S$, $G_S$, and $G_T$.)  Given a graph $G = (V,E)$, we say that an \emph{axis} of $G$ is a horizontal or vertical line of vertices such that there is an edge between each pair of adjacent points on that line. Notice that two distinct axes can be collinear. Given an axis $a$, an \emph{axial branch} of $T_S$ is a branch of $T_S$ which contains exactly one vertex $v$ on $a$ and all vertices and edges of $T_S$ which are connected to a vertex that does not lie on $a$ and is adjacent to $v$. %
We say that the branch \emph{begins from} $v$.  Intuitively, an axial branch is a connected component extending from an axis.  (See the pink highlighted portion of Figure~\ref{fig:epsilon-symmetric-shape-tree} for an example axial branch off of the axis shown in green.)

\begin{figure}[htp]
\centering
  \subfloat[][Example shape $S$]{%
        \label{fig:epsilon-symmetric-shape}%
        \makebox[1.3in][c]{\includegraphics[width=.75in]{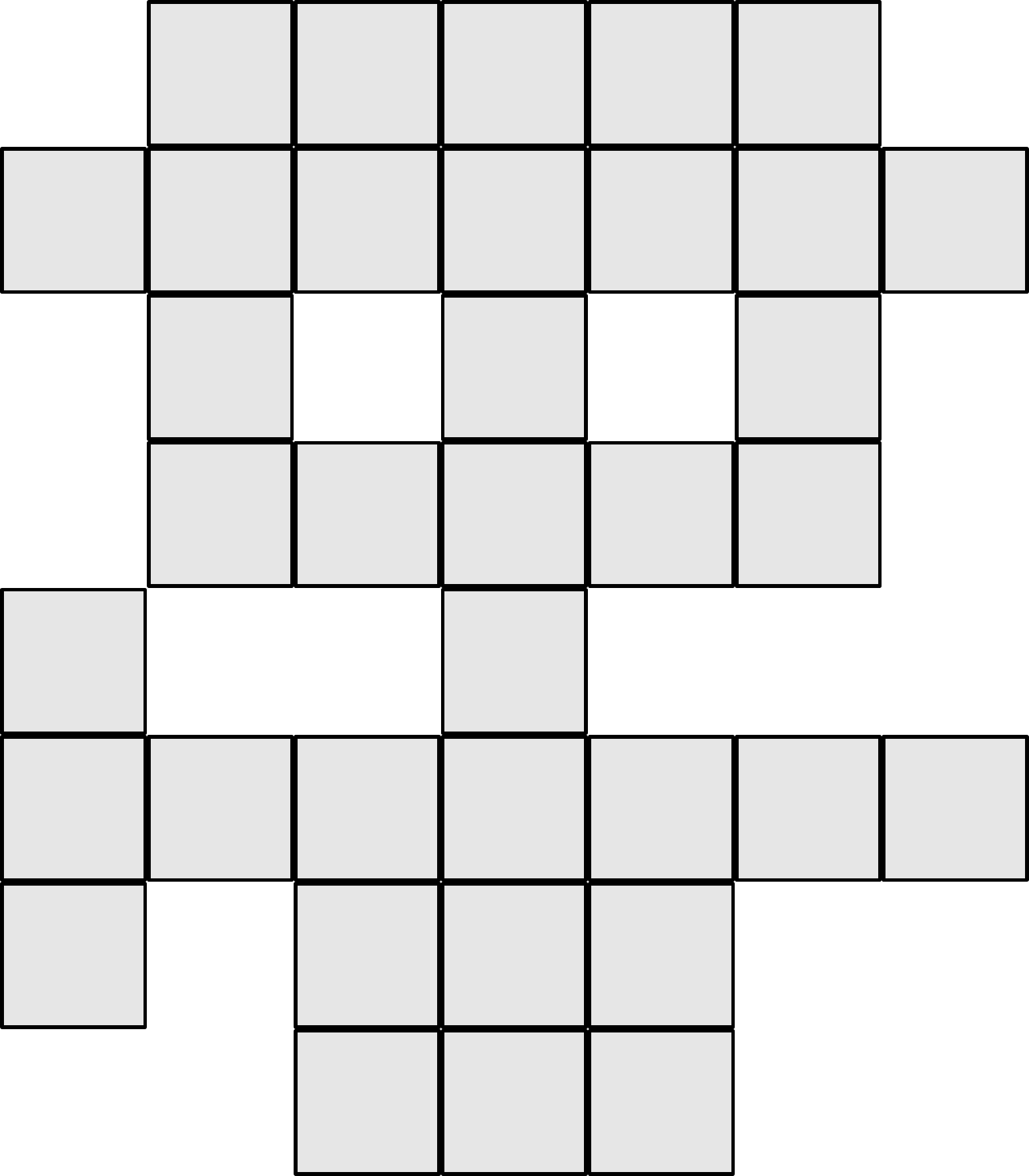}}
        }%
        \quad\quad
  \subfloat[][Graph of $S$, $G_S$]{%
        \label{fig:epsilon-symmetric-shape-graph}%
        \makebox[1.3in][c]{\includegraphics[width=.7in]{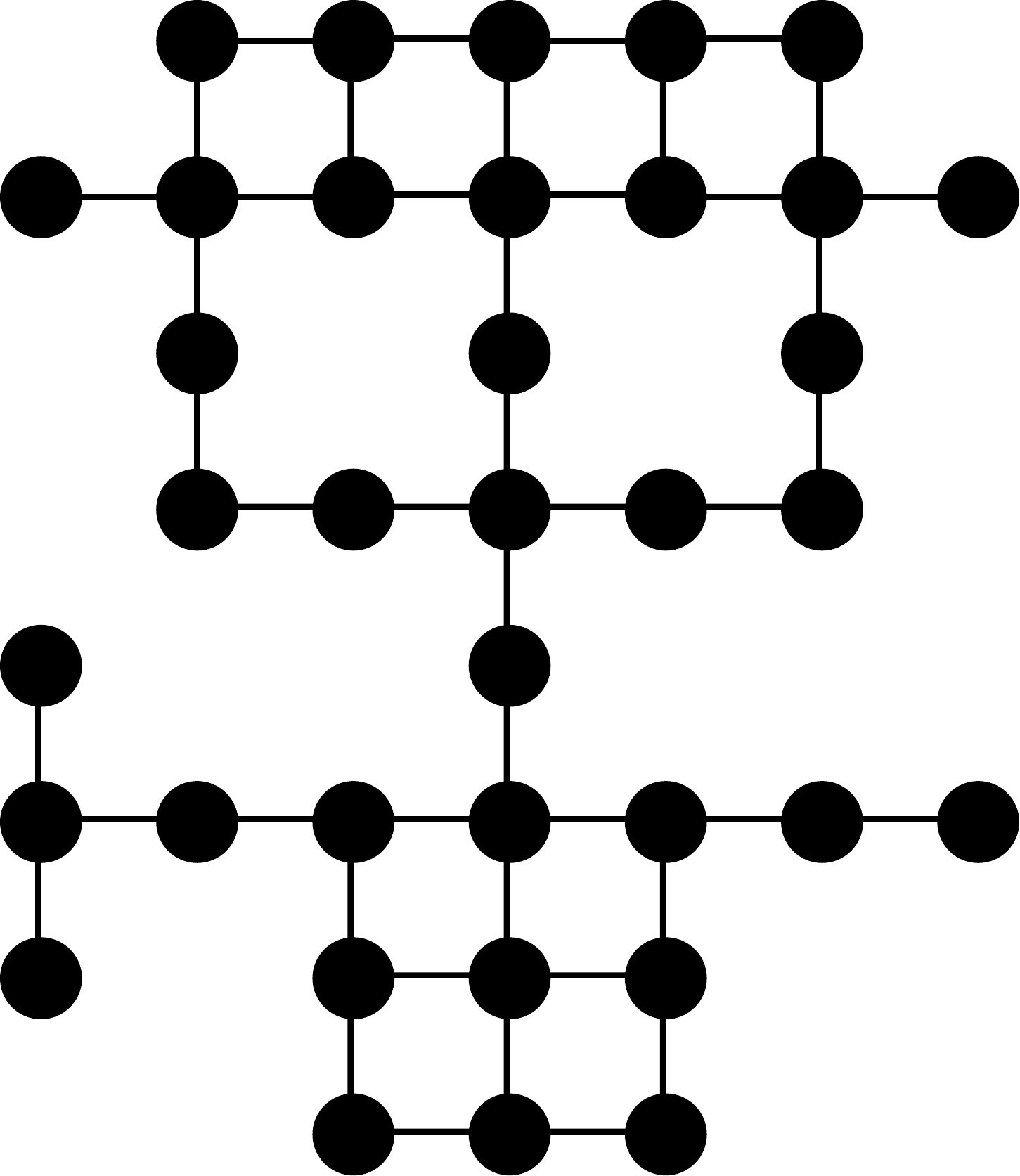}}
        }%
        \quad\quad
  \subfloat[][Tree of $S$, $T_S$]{%
        \label{fig:epsilon-symmetric-shape-tree}%
        \makebox[1.3in][c]{\includegraphics[width=.7in]{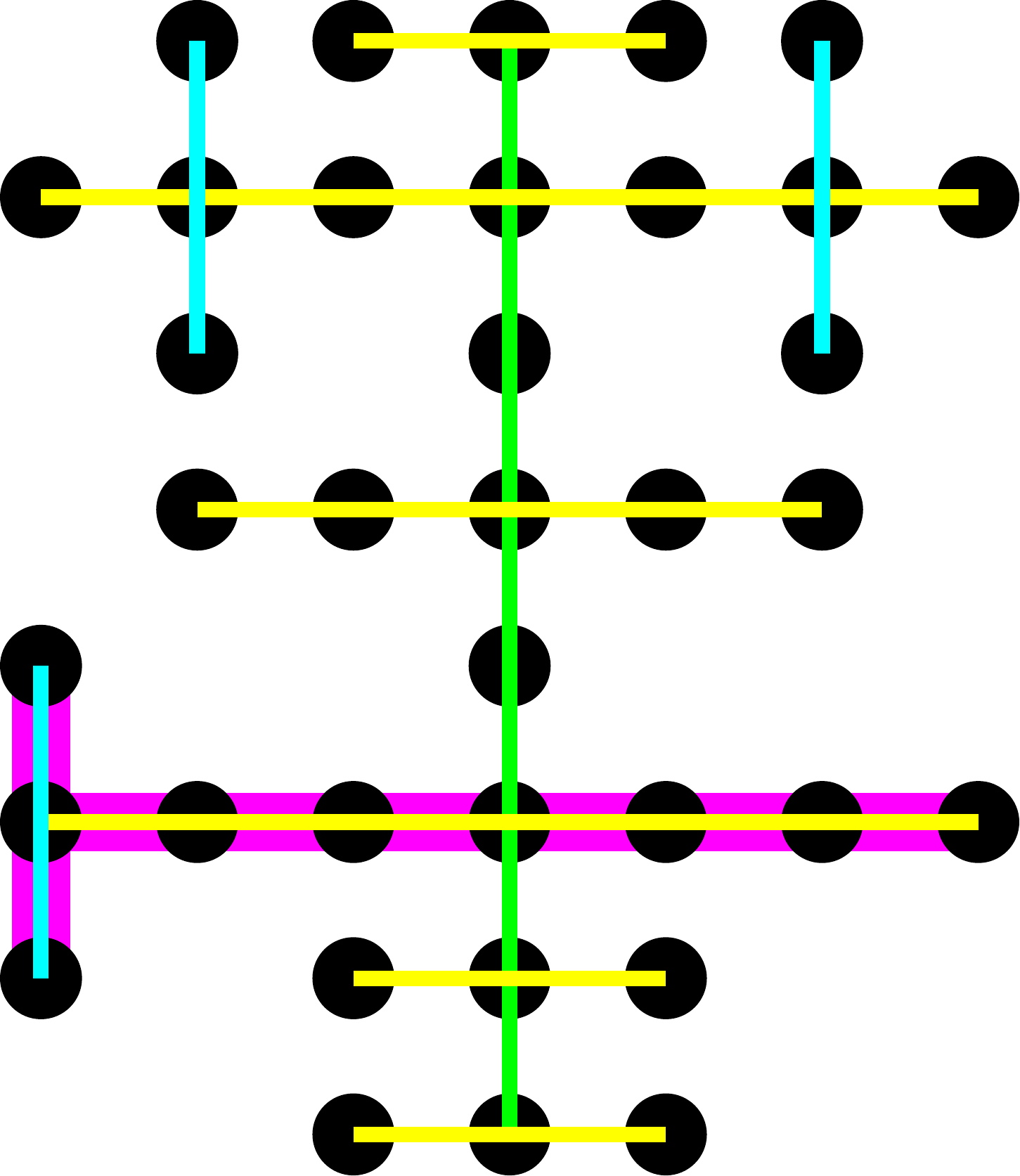}}
        }%
  \caption{An example $\epsilon$-symmetric shape.}
  \label{fig:epsilon-symmetric-example}
\end{figure}

A tree $T_S$ is symmetric across an axis $a$ if, for every vertex $v$ contained on $a$, the branches of $a$ which begin from $v$ are symmetric across $a$.
A tree $T_S$ is \emph{off-by-one symmetric} across an axis $a$ if, for every vertex $v$ except for at most 1, the branches of $a$ which begin from $v$ are symmetric across $a$. See Figure~\ref{fig:epsilon-symmetric-shape-tree} for an example of such a tree, with the axis $a$ shown in green.  %

\begin{definition}\label{def:eps-symm-tree}
A tree $T$ is \emph{$\epsilon$-symmetric} if and only if for any axis $a$ of $T$, $T$ is off-by-one symmetric across $a$.
\end{definition}

\begin{definition}\label{def:eps-symm}
Given a shape $S$ with graph $G_S$, we say that $S$ is \emph{$\epsilon$-symmetric} if and only if there exists a spanning tree, $T_S$, of $G_S$ such that $T_S$ is $\epsilon$-symmetric.
\end{definition}

For an example of an $\epsilon$-symmetric shape $S$, see Figure~\ref{fig:epsilon-symmetric-shape}.  The tree $T_S$ is off-by-one symmetric across the vertical green axis, the branches off of that axis are symmetric across the horizontal yellow axes, and the branches off of those axes are symmetric across the vertical blue axes. The following theorem gives a complete classification of finite connected shapes which can be assembled by temperature-1 singly seeded mismatch-free \rtam/ systems. The proof of this theorem is given in Section~\ref{sec:proof-temp1-shapes}. %

\begin{theorem}\label{thm:temp1-shapes}
Let $S\subset \Z^2$ be a finite connected shape. There exists a mismatch-free \rtam/ system $\mathcal{T} = (T,\sigma, 1)$ with $|\sigma| = 1$ that strictly assembles $S$ if and only if $S$ is $\epsilon$-symmetric.
\end{theorem}

\later{
\section{Proof of Theorem~\ref{thm:temp1-shapes}}\label{sec:proof-temp1-shapes}

First, we show that if $S$ is $\epsilon$-symmetric, then there is a singly seeded mismatch-free \rtam/ system that strictly assembles $S$.
Let $T_S$ be a tree for $S$ which satisfies Definition~\ref{def:eps-symm}. We use $T_S$ to define a tile set $T$ and give a seed tile $\sigma$ such that the system $\mathcal{T}= (T,\sigma, 1)$ strictly self-assembles $S$, and then we argue that for any terminal assembly $\alpha$ in $\mathcal{T}$, the domain of $\alpha$ is equal to $S$ (up to translation and reflection). $T$ and $\sigma$ are obtained as follows.

First, notice that we can give a temperature-1 aTAM system $\mathcal{T}' = (T', \sigma', 1)$ that assembles the shape $S$ with the properties
(1) $\mathcal{T}'$ is a singly seeded directed system,
(2) for $\alpha' \in \prodasm{\mathcal{T}'}$, the binding graph of $\alpha'$ is exactly the tree $T_S$, and
(3) when the seed tile $\sigma'$ is the single tile located at some point $v$, then for each point $p$ of $S$, there is a unique tile $t'$ in $T'$ such that $\alpha(p) = t'$.
More intuitively, when the seed tile of $\mathcal{T'}$ is placed at $v$, then assembly proceeds by the binding of unique tiles equipped with glues specific to each point of $S$.

Now, we define the tiles of $T$ to be copies of the tiles in $T'$ and define $\sigma'$ to be $\sigma$. Notice that in the \rtam/, if we assume that no tiles are reflected before binding as assembly proceeds, then the terminal assembly, $\alpha$, that results has a domain equal to $S$. Furthermore, since the binding graph of $\alpha$ is $T_S$, a reflected tile, $t$, of $\alpha$ lies on some axis, $a_t$ say, and corresponds to some vertex $v_t$ of that axis.
Then, since $T_S$ is $\epsilon$-symmetric and appropriately reflected tiles of $T$ can still bind to $t$ even when $t$ is reflected prior to binding, the domain of a terminal assembly of $\mathcal{T}$ will either (1) be equal to $S$ (This is the case that the axial branches beginning from $v_t$ are equivalent upto reflection about the axis $a_t$. See Figure~\ref{fig:epsilon-symmetric-assembly2}.), or (2) be equal to a reflection of $S$ about the line corresponding to the axis $a_0$ (This is the special case that the axial branches beginning from $v_t$ are not equivalent after a reflection about $a_t$. See Figure~\ref{fig:epsilon-symmetric-assembly3}.). In either case, it follows that $S$ strictly assembles in $\mathcal{T}$.

\begin{figure}[htp]
\centering
  \subfloat[][]{%
        \label{fig:epsilon-symmetric-assembly1}%
        \makebox[1.3in][c]{\includegraphics[width=1.15in]{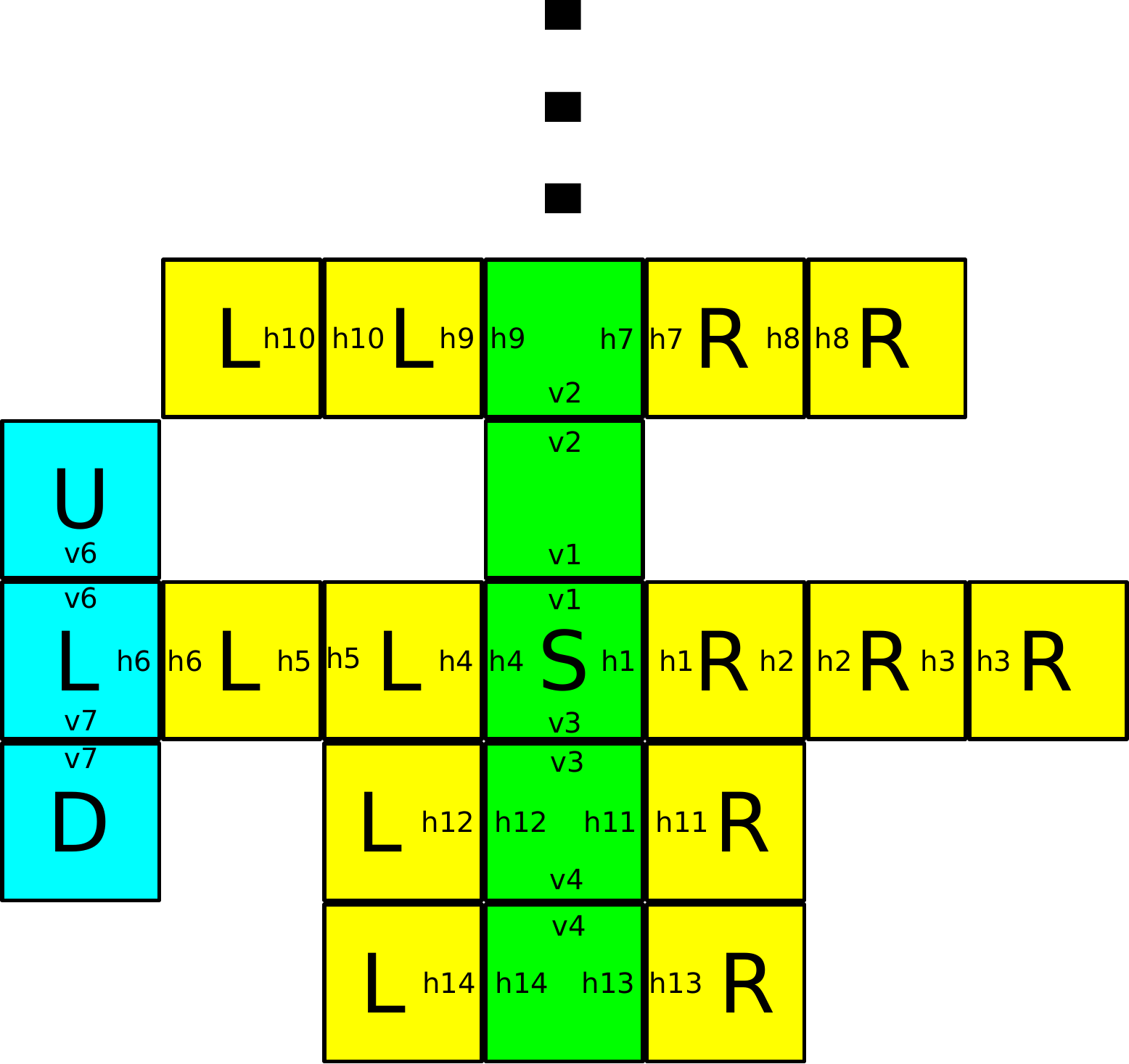}}
        }%
        \quad\quad
  \subfloat[][]{%
        \label{fig:epsilon-symmetric-assembly2}%
        \makebox[1.3in][c]{\includegraphics[width=1.1in]{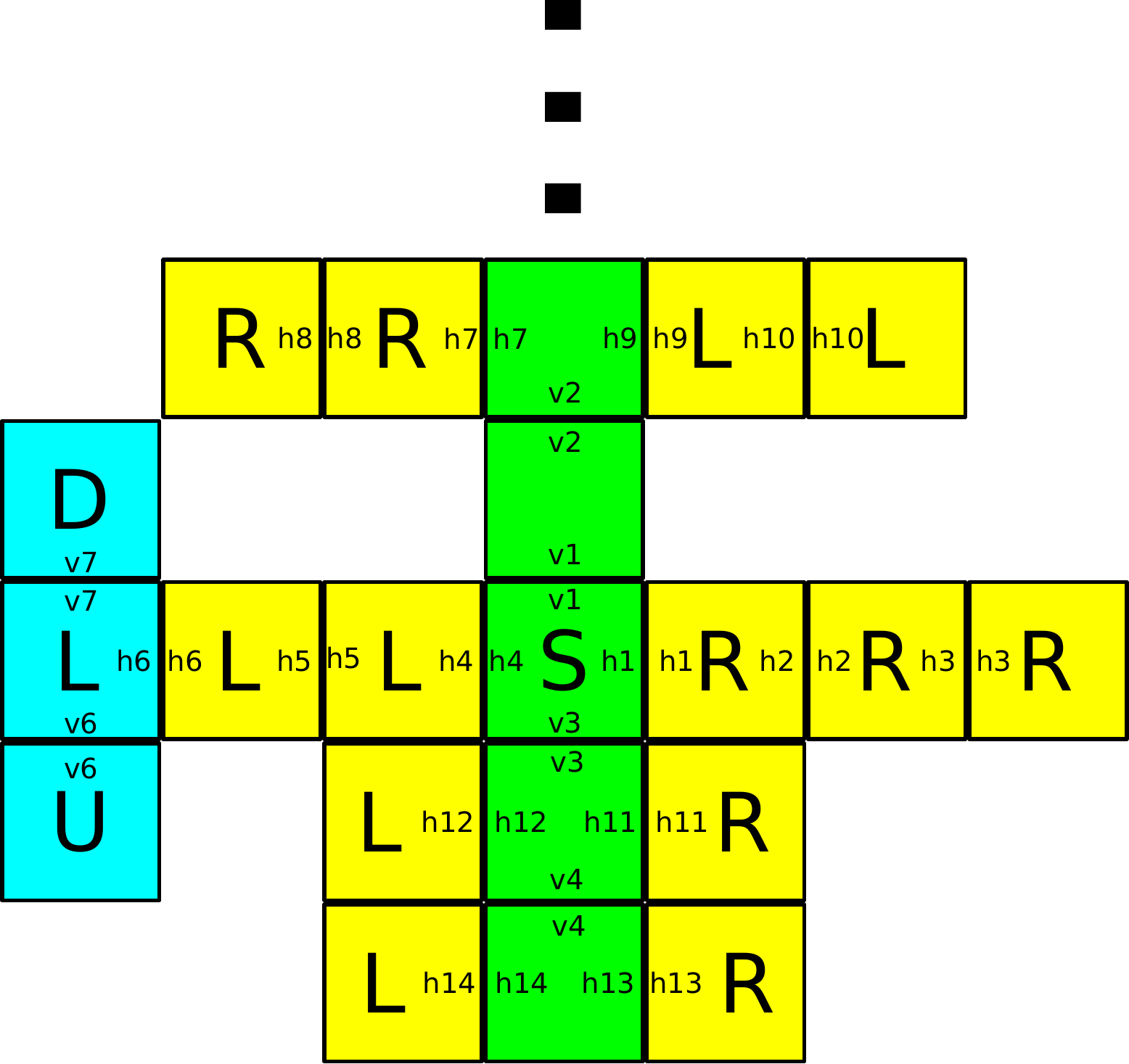}}
        }%
        \quad\quad
  \subfloat[][]{%
        \label{fig:epsilon-symmetric-assembly3}%
        \makebox[1.3in][c]{\includegraphics[width=1.1in]{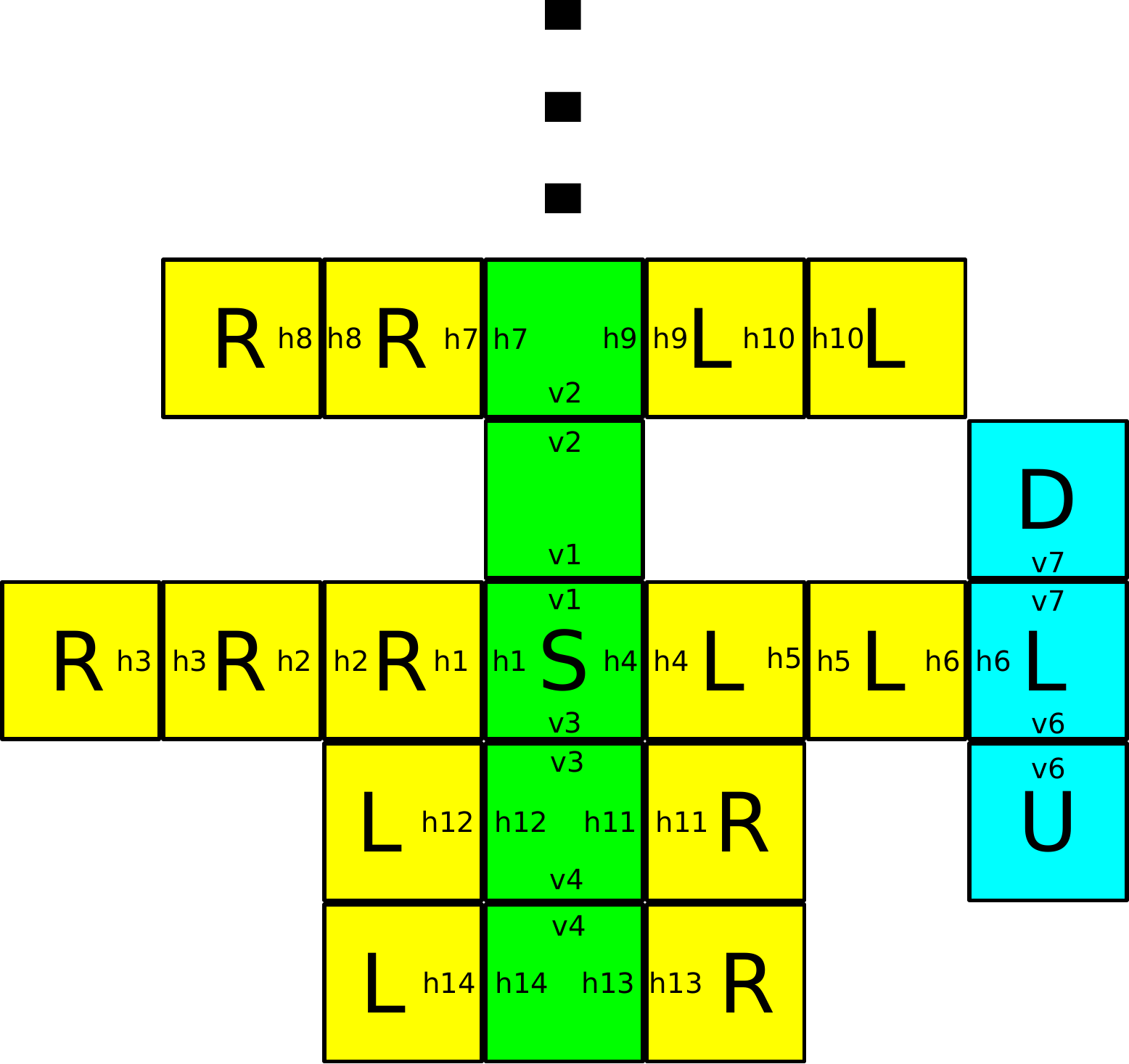}}
        }%
  \caption{An example of the assembly of a portion of an $\epsilon$-symmetric shape which assembles in $\mathcal{T}$. Relative to Figure (a), Figures (b) and (c) give two different producible assemblies of $\mathcal{T}$. }
  \label{fig:epsilon-symmetric-assembly}
\end{figure}

Now we show that if there is a singly seeded mismatch-free \rtam/ system that strictly assembles $S$, then $S$ is $\epsilon$-symmetric. Let $\mathcal{T} = (T, \sigma, 1)$ be a singly seeded mismatch-free \rtam/ system that strictly self-assembles $S$, and let $\alpha \in \termasm{T}$ be some terminal assembly of $\mathcal{T}$.

First, we note that the binding graph of $\alpha$ must be a tree, which we denote by $T_\alpha$.
This follows from the fact that if the binding graph contains a cycle, then there are infinitely many producible assemblies in $\mathcal{T}$. See Figure~\ref{fig:pump-cycle} for an example. Moreover,
$T_\alpha$ must be $\epsilon$-symmetric. This can be shown by contradiction as follows.

\begin{figure}[htp]
\centering
  \subfloat[][]{%
        \label{fig:pump-cycle1}%
        \makebox[1.3in][c]{\includegraphics[scale=.2]{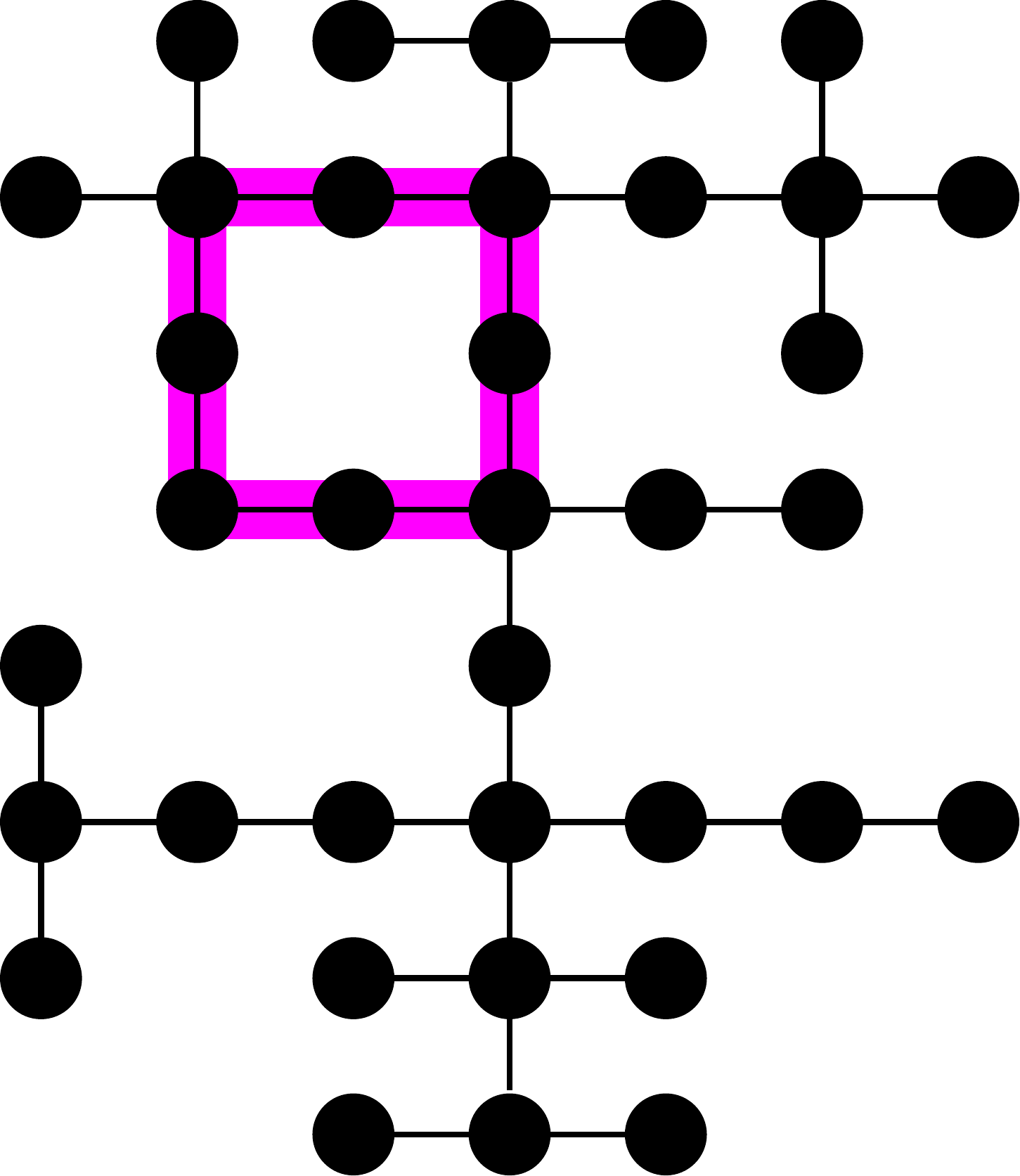}}
        }%
        \quad\quad
  \subfloat[][]{%
        \label{fig:pump-cycle2}%
        \makebox[1.3in][c]{\includegraphics[scale=.2]{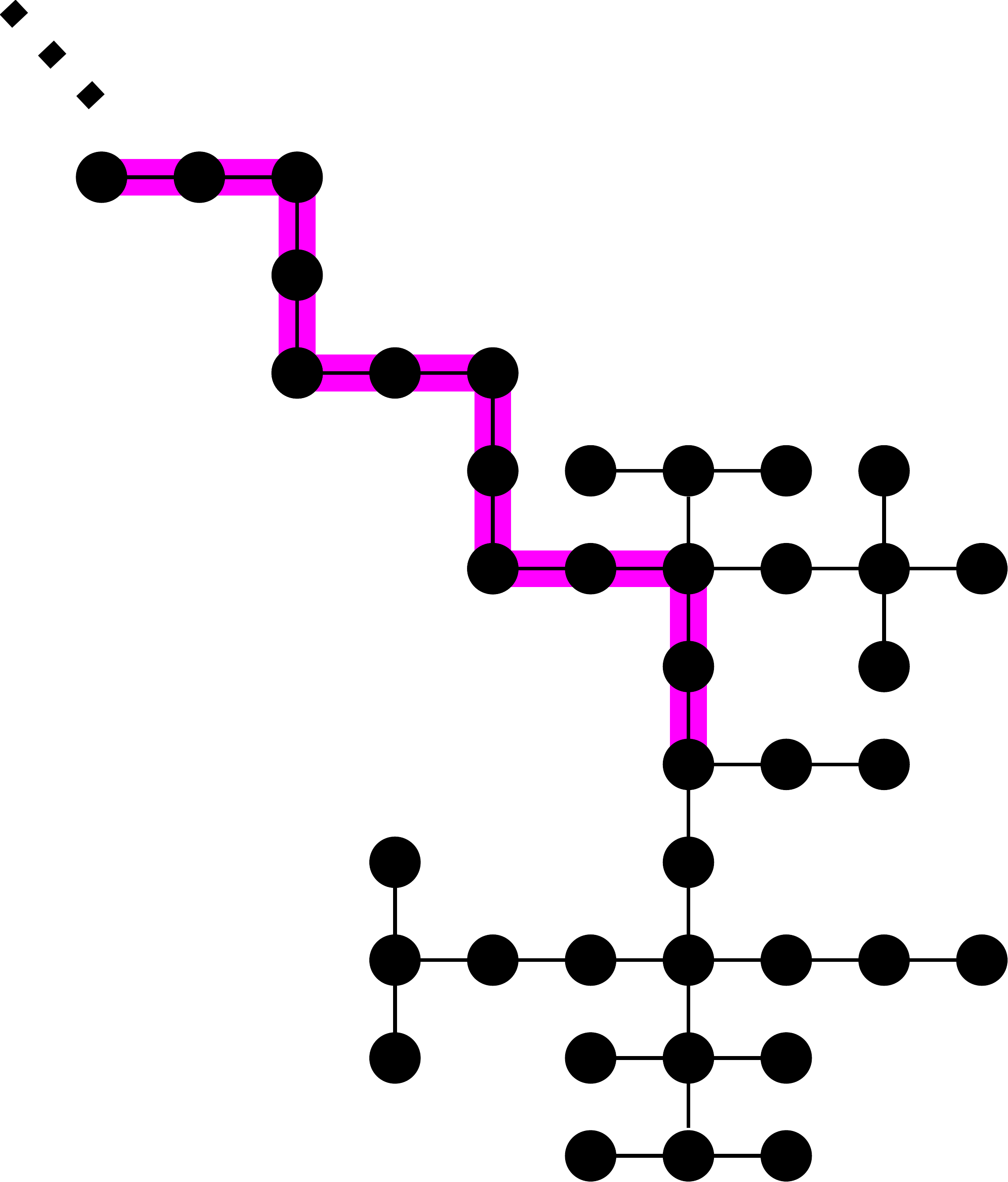}}
        }%
  \caption{(a) An example of a binding graph that contains a cycle. (b) An example of a binding graph of an assembly where tiles with type of those tiles belonging to the cycle shown in (a) are repeated an arbitrary number of times using appropriately reflected tiles.}
  \label{fig:pump-cycle}
\end{figure}

Suppose that $T_\alpha$ is not $\epsilon$-symmetric.
Then, there exists an axis $a$ of $T_\alpha$ such that $T_\alpha$ is not off-by-one symmetric across $a$.
In other words, there are at least two vertices, which we denote by $v_1$ and $v_2$, of $T_\alpha$ which lie on $a$ such that the branches of $a$ which begin from $v_1$ are not symmetric across $a$ and the branches of $a$ which begin from $v_2$ are not symmetric across $a$.
Now, we can modify the assembly sequence of $\alpha$ by reflecting the tile, $t$, located at $v_1$ prior to binding, allowing reflected tiles to assemble reflections of the subassemblies which are bound to $t$, and keeping the original assembly sequence otherwise. Now, as tiles of the reflected subassemblies originating from $t$ bind, either a mismatch occurs or the reflected subassemblies completely assemble.
If a mismatch occurs, we arrive at a contradiction since $\mathcal{T}$ is mismatch-free.
If the reflected subassemblies completely assemble, then we note that no translation and reflection of $\dom(\beta)$ is equal to $S$.  (This is because we essentially chose a sequence where the two asymmetric branches across $a$ adopted reflections across $a$ such that the sides of each which would be on the same side of $a$ in $S$ are now on different sides of $a$.) This contradicts the assumption that $\mathcal{T}$ strictly assemblies $S$.
Therefore, $T_\alpha$ is $\epsilon$-symmetric. Hence $S$ is $\epsilon$-symmetric.

\begin{figure}[htp]
\centering
  \subfloat[][]{%
        \label{fig:not-epsilon-symmetric-assembly1}%
        \makebox[1.3in][c]{\includegraphics[width=1.15in]{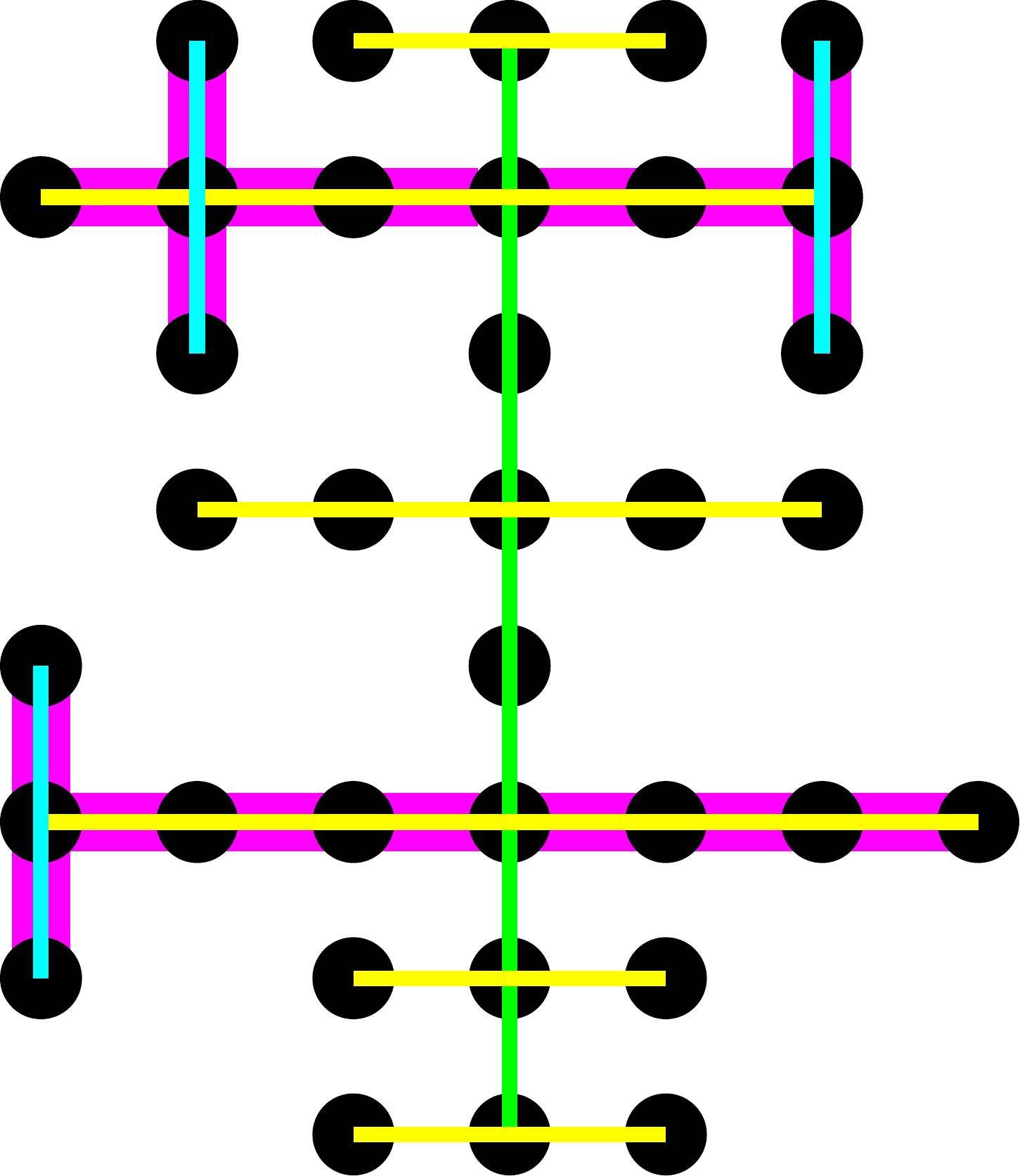}}
        }%
        \quad\quad
  \subfloat[][]{%
        \label{fig:not-epsilon-symmetric-assembly2}%
        \makebox[1.3in][c]{\includegraphics[width=1.1in]{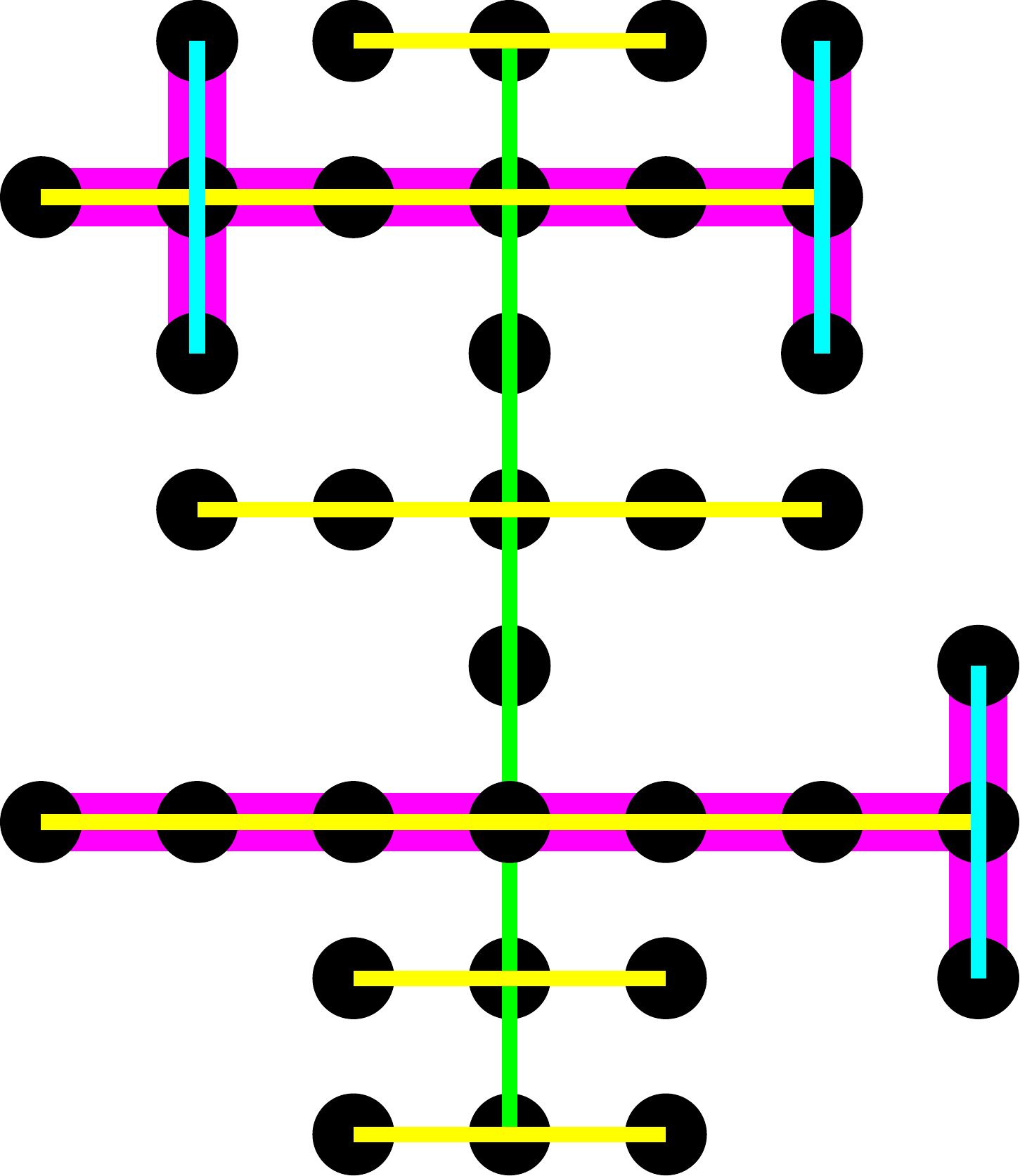}}
        }%
  \caption{An example of possible binding graphs of a shape that is not $\epsilon$-symmetric. (a) and (b) depict binding graphs of two distinct assemblies.}
  \label{fig:not-epsilon-symmetric-assembly}
\end{figure}

}%

While Theorem~\ref{thm:temp1-shapes} shows exactly which shapes can be assembled without cooperation or mismatches by singly seeded \rtam/ systems, the following theorem shows that with cooperation, \rtam/ systems can assemble arbitrary scale factor $2$ shapes. The proof of this theorem is in Section~\ref{sec:proof-temp2-shapes}.

\begin{theorem}\label{thm:temp2-shapes}
Let $S\subset \Z^2$ be a finite connected shape, and $S^2$ be $S$ at scale factor $2$. There exists a mismatch-free \rtam/ system $\mathcal{T} = (T,\sigma, 2)$ with $|\sigma| = 1$ that strictly self-assembles $S^2$.
\end{theorem}

\later{
\section{Proof of Theorem~\ref{thm:temp2-shapes}}\label{sec:proof-temp2-shapes}
Let $S$ be a finite connected shape in $\Z^2$ and $S^2$ be $S$ at scale factor $2$. Note that we can give a temperature-1 aTAM system $\mathcal{T}' = (T', \sigma', 1)$ that assembles the shape $S$ with the properties
(1) $\mathcal{T}'$ is a singly seeded directed system,
(2) for $\alpha'$ in $\termasm{T}'$, the binding graph of $\alpha'$ is a tree $T_{\alpha'}$,
(3) the location of the seed tile $\sigma'$ corresponds to a leaf of $T_{\alpha'}$, and
(4) for each point $p$ of $S$, there is a unique tile $t'$ in $T'$ such that $\alpha(p) = t'$.
For an example of such an assembly $\alpha'$ see Figure~\ref{fig:temp2-shape-example1}.

\begin{figure}[htp]
\centering
        \includegraphics[scale=.55]{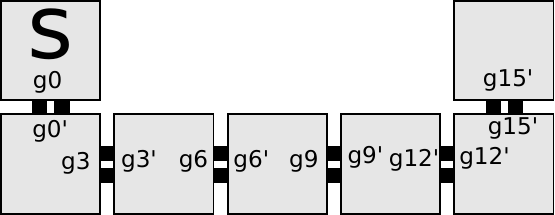}
  \caption{An example terminal assembly giving a shape assembled by an aTAM system $\mathcal{T}'$.}
  \label{fig:temp2-shape-example1}
\end{figure}

We now give an \rtam/ system $\mathcal{T}$ based on $\mathcal{T}'$ such that a terminal assembly $\alpha$ of $\mathcal{T}$ can be obtained from $\alpha'$ by replacing single tiles of $\alpha'$ by $2\times 2$ blocks of tiles, and thus assembles $S^2$. aTAM tiles and corresponding \rtam/ tiles that give rise to this block replacement scheme are described in Figure~\ref{fig:temp2-RTAM}.

\begin{figure}[htp]
\centering
  \subfloat[][]{%
        \label{fig:temp2-aTAM-seed}%
        \makebox[1.3in][c]{\includegraphics[scale=1.0]{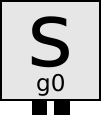}}
        }%
        \quad\quad
  \subfloat[][]{%
        \label{fig:temp2-RTAM-seed}%
        \makebox[1.3in][c]{\includegraphics[scale=1.0]{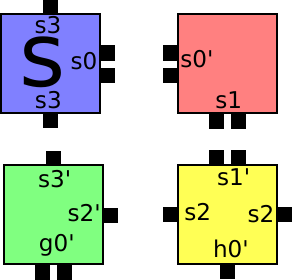}}
        }%
        \\\bigskip
  \subfloat[][]{%
        \label{fig:temp2-aTAM-tile}%
        \makebox[1.3in][c]{\includegraphics[scale=1.0]{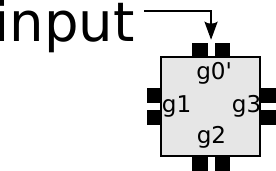}}
        }%
        \quad\quad
  \subfloat[][]{%
        \label{fig:temp2-RTAM-tile}%
        \makebox[1.3in][c]{\includegraphics[scale=1.0]{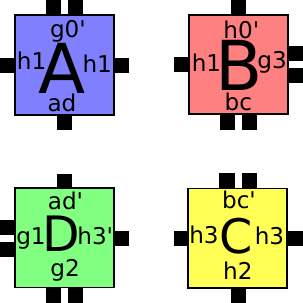}}
        }%
  \caption{(a) An example seed tile of the aTAM system $\mathcal{T}'$. (b) Tiles that represent the seed tile of (a) used to define the \rtam/ system $\mathcal{T}$. The seed of $\mathcal{T}$ consists of only the tile labeled $S$. (c) An example tile of the aTAM system $\mathcal{T}'$ that attaches via the glue $g0'$. Note that it may be the case that not all glues $g_1$, $g_2$, and $g_3$ shown here are exposed. (d) Tiles that represent the tile of (c) used to define the \rtam/ system $\mathcal{T}$. Notice that tiles bind in the order $A$, $B$, $C$ and $D$. The tile shown in (c) has ``output'' glues $g1$, $g2$, and $g3$ exposed. In the case where one or more of these glues is not exposed by a tile, we obtain tiles analogous to those in (d) that represent this single tile as follows. The glues $g1$, $g2$, and $g3$ of the tiles of (d) are exposed if and only if the respective glues of the tile in (c) are exposed. For example, if the tile in (c) exposes a $g1$ glue, then so will the tile $D$ of (d). Otherwise, $D$ will not expose a $g1$ glue. Note that the cooperative binding of $B$ fixes the orientation of $B$. Similarly, the cooperative binding of $D$ fixes the orientation of the tile $D$.}
  \label{fig:temp2-RTAM}
\end{figure}

As $2\times 2$ blocks assemble, cooperation is used to ensure that the orientations of the tiles of a $2\times 2$ block are fixed relative to the orientation of the seed. Figure~\ref{fig:temp2-shape-example2} depicts a terminal assembly of the \rtam/ system $\mathcal{T}$ based on the aTAM system $\mathcal{T}'$ whose terminal assembly is shown in Figure~\ref{fig:temp2-shape-example1}.

\begin{figure}[htp]
\centering
        \includegraphics[scale=1.0]{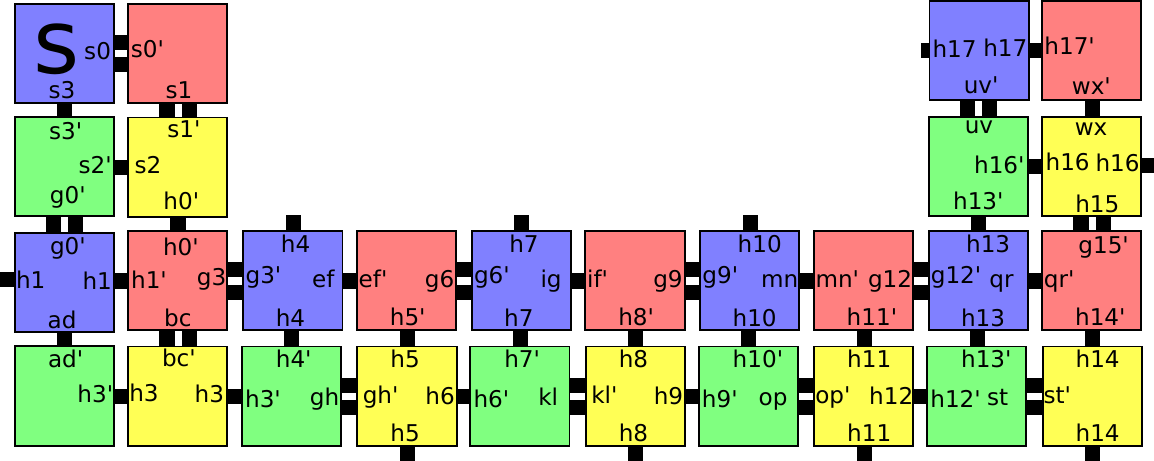}
  \caption{An example terminal assembly giving a shape assembled by an \rtam/ system $\mathcal{T}$ based on the aTAM system $\mathcal{T}'$ whose terminal assembly is shown in Figure~\ref{fig:temp2-shape-example1}.}
  \label{fig:temp2-shape-example2}
\end{figure}

}%

\ifabstract
\later{

}
\fi
\iffull

\fi
\renewcommand\refname{\vspace{-45pt}\section*{References\vspace{-20pt}}}
\bibliographystyle{plain}%
{\bibliography{tam,experimental_refs,ca}}

\ifabstract
\newpage
\appendix
\magicappendix
\fi

\end{document}